\documentclass[12pt,onecolumn]{IEEEtran}

\usepackage{latexsym}
\usepackage{graphicx}
\usepackage{array}
\usepackage{amsmath}
\usepackage{amsfonts}
\usepackage{amssymb}
\usepackage{amsthm}
\usepackage{epsfig}
\usepackage{epstopdf}
\usepackage{cite}
\usepackage{mathtools}
\usepackage[normal]{subfigure}
\usepackage[colorlinks=true]{hyperref}

\numberwithin{equation}{section}%sets equation numbers <chapter>.<section>.<index>

\newtheorem{theorem}{Theorem}[section]
\newtheorem{lemma}[theorem]{Lemma}
\newtheorem{definition}[theorem]{Definition}

\newtheorem{remark}[theorem]{Remark}

\newcommand{\sr}{\stackrel}

\newcommand{\tri}{\sr{\triangle}{=}}

\newcommand{\be}{\begin{equation}}
\newcommand{\ee}{\end{equation}}
\newcommand{\bea}{\begin{eqnarray}}
\newcommand{\eea}{\end{eqnarray}}
\newcommand{\bes}{\begin{eqnarray*}}
\newcommand{\ees}{\end{eqnarray*}}
\newcommand{\bi}{\begin{itemize}}
\newcommand{\ei}{\end{itemize}}
\newcommand{\ben}{\begin{enumerate}}
\newcommand{\een}{\end{enumerate}}
\newcommand{\bp}{\begin{problem}}
\newcommand{\ep}{\end{problem}}

\newcommand{\hst}{\hspace{.2in}}

\newcommand{\noi}{\noindent}

\begin{document}
%\bibliographystyle{ieeetr}
%\baselineskip=18pt
%\onehalfspacing

%\begin{flushright}
%\tiny{\today}
%\end{flushright}

\title{Directed Information on Abstract Spaces: Properties and Variational Equalities}
%
%\vspace*{1.0cm}
%
\author{Charalambos~D.~Charalambous and Photios~A.~Stavrou~%\IEEEmembership{Student Member,~IEEE}
\thanks{This work was financially supported by a medium size University of Cyprus grant entitled “DIMITRIS”. Parts of the material in this paper were presented at the IEEE International Symposium on Information Theory, Boston MA, July 1--6 2012 \cite{charalambous-stavrou2012}, at the IEEE International Symposium on Information Theory, Istanbul, Turkey, July 7--12 2013 \cite{stavrou-charalambous2013c}, and in book series ``Lecture Notes in Control and Information Sciences'' \cite{charalambous-stavrou-kourtellaris2015}.}
\thanks{The authors are with the Department of Electrical and Computer Engineering (ECE), University of Cyprus, 75 Kallipoleos Avenue, P.O. Box 20537, Nicosia, 1678, Cyprus, {\it Email:\{chadcha,stavrou.fotios\}@ucy.ac.cy}}}

%\thanks{Manuscript received October 9,~2012;}}%~revised....}}
%}
%
%\vspace*{0.3cm}

% The paper headers
%\markboth{Submitted to IEEE Transactions on Information Theory}%,~Vol.~X, No.~X, ~XXXX}%
%{Shell \MakeLowercase{\textit{et al.}}: Bare Demo of IEEEtran.cls for Journals}

\maketitle
%\thispagestyle{empty}
%\vspace*{1.cm}

\begin{abstract}
\par Directed information or its variants are utilized extensively in the characterization of the capacity of channels with memory and feedback, nonanticipative lossy data compression, and their generalizations to networks.
\par In this paper, we derive several functional and topological properties of directed information for general abstract alphabets (complete separable metric spaces) using the topology of weak convergence of probability measures. These include convexity of the set of consistent distributions, which uniquely define causally conditioned distributions, convexity and concavity of directed information with respect to the sets of consistent distributions, weak compactness of these sets of distributions, their joint distributions and their marginals. Furthermore, we show lower semicontinuity of directed information, and under certain conditions we also establish continuity of directed information. Finally, we derive variational equalities for directed information, including sequential versions. These may be viewed as the analogue of the variational equalities of mutual information (utilized in Blahut-Arimoto algorithm).
\par In summary, we extend the basic functional and topological properties of mutual information to directed information. These properties are discussed in the context of extremum problems of directed information. 
\end{abstract}

%\newpage

%table of contents

%\tableofcontents

%\newpage

\section{Introduction}
\par Directed information quantifies the directivity of information defined by a causal sequence of feedback and feedforward channel conditional distributions \cite{marko,massey90}. Specifically, given two sequences of Random Variables (RV's) $X^n\tri\{X_0,X_1,\ldots,X_n\}\in{\cal X}_{0,n}\tri\times_{i=0}^n{\cal X}_i$, $Y^n\tri\{Y_0,Y_1,\ldots,Y_n\}\in{\cal Y}_{0,n}\tri\times_{i=0}^n{\cal Y}_i$, where ${\cal X}_i$ and ${\cal Y}_i$ are the input and output alphabets of a channel, respectively, and ${\cal B}({\cal X}_i),~{\cal B}({\cal Y}_i)$, the corresponding measurable spaces, directed information from $X^n$ to $Y^n$ is often defined via conditional mutual information \cite{massey90,kramer1998} as follows.
\begin{align}
I(X^n\rightarrow{Y}^n)&\tri\sum_{i=0}^{n}I(X^i;Y_i|Y^{i-1})\label{equation1a}\\
&=\sum_{i=0}^n\int_{{\cal X}_{0,i}\times{\cal Y}_{0,i}}\log\bigg(\frac{dP_{Y_i|Y^{i-1},X^i}(\cdot|y^{i-1},x^i)}{dP_{Y_i|Y^{i-1}}(\cdot|y^{i-1})}(y_i)\bigg)P_{X^i,Y^i}(dx^i,dy^i)\label{equation1b}\\
&\equiv\mathbb{I}_{X^n\rightarrow{Y}^n}(P_{X_i|X^{i-1},Y^{i-1}},P_{Y_i|Y^{i-1},X^i}:i=0,1,\ldots,n)\label{equation1c}
\end{align}
where notion \eqref{equation1c} indicates that directed information $I(X^n\rightarrow{Y}^n)$ is a functional of two collections of causally conditioned distributions, $\{P_{Y_i|Y^{i-1},X^i}:i=0,\ldots,n\}$, and $\{P_{X_i|X^{i-1},Y^{i-1}}:i=0,1,\ldots,n\}$, called feedforward distribution, and feedback feedback distribution, respectively,  which uniquely define the joint distribution $\{P_{X^i,Y^i}:~i=0,1,\ldots,n\}$ and the conditional distribution $\{P_{Y_i|Y^{i-1}}:i=0,1,\ldots,n\}$ of the RV's $\{(X^i,Y^i):~i=0,1,\ldots,n\}$.
\vspace*{0.2cm}\\
\noi By Bayes' rule, for any $A_j\in{\cal B}({\cal X}_j), B_j\in{\cal B}({\cal Y}_j),~j=0,1,\ldots,i$, the joint distribution decomposes into 
\begin{align}
P_{X^i,Y^i}&(A_0,B_0,\ldots,A_i,B_i)=\int_{A_0}P_{X_0}(dx_0)\int_{B_0}P_{Y_0|X_0,Y^{-1}}(dy_0|x_0,y^{-1})\ldots\nonumber\\
&\ldots\int_{A_i}P_{X_i|X^{i-1},Y^{i-1}}(dx_i|x^{i-1},y^{i-1})\int_{B_i}P_{Y_i|Y^{i-1},X^{i}}(dy_i|y^{i-1},x^{i}),~i=0,1,\ldots,n.\label{directed:information:section:introduction:equation1}
\end{align}
\noi Formally, we represent (\ref{directed:information:section:introduction:equation1}) by $P_{X^i,Y^i}(dx^i,dy^i)=\otimes_{j=0}^i\big({P}_{X_j|X^{j-1},Y^{j-1}}\otimes{P}_{Y_j|Y^{j-1},X^j}\big)$, and we call it an $(n+1)$-fold compound probability distribution.

\noi If the distributions $\{P_{X_i|X^{i-1},Y^{i-1}},P_{Y_i|Y^{i-1},X^i}:~i=0,\ldots,n\}$ are defined with respect to the probability density functions of continuous valued  RV's $\{(X_i,Y_i):~i=0,1,\ldots,n\}$, denoted by, $\{f_{X_i|X^{i-1},Y^{i-1}},f_{Y_i|Y^{i-1},X^i}$ $:~i=0,\ldots,n\}$, then (\ref{equation1a}) reduces  to 
\begin{align*}
I(X^n\rightarrow{Y}^n)=\sum_{i=0}^n\int_{{\cal X}_{0,i}\times{\cal Y}_{0,i}}\log\Big(\frac{f_{Y_i|Y^{i-1},X^i}(y_i|y^{i-1},x^i)}{f_{Y_i|Y^{i-1}}(y_i|y^{i-1})}\Big)f_{X^i,Y^i}(x^i,y^i)dx^idy^i.
\end{align*}
\noi If the distributions $\{P_{X_i|X^{i-1},Y^{i-1}},P_{Y_i|Y^{i-1},X^i}:~i=0,\ldots,n\}$ are defined with respect to the probability mass functions of countable or finite alphabet valued RV's $\{(X_i,Y_i):~i=0,\ldots,n\}$, denoted by, $\{p_{X_i|X^{i-1},Y^{i-1}},p_{Y_i|Y^{i-1},X^i}:~i=0,\ldots,n\}$, then (\ref{equation1a}) reduces  to 
\begin{align*}
I(X^n\rightarrow{Y}^n)=\sum_{i=0}^n\sum_{(x^i,y^i)\in{\cal X}_{0,i}\times{\cal Y}_{0,i}}\log\Big(\frac{p_{Y_i|Y^{i-1},X^i}(y_i|y^{i-1},x^i)}{p_{Y_i|Y^{i-1}}(y_i|y^{i-1})}\Big)p_{X^i,Y^i}(x^i,y^i).
\end{align*}
\vspace*{0.2cm}\\
In information theory, directed information (\ref{equation1a}) or its variants are used to characterize capacity of channels with memory and feedback  \cite{tatikonda2000,chen-berger2005,yang-kavcic-tatikonda2005,permuter-cuff-vanroy-weissman2008,tatikonda-mitter2009,permuter-weissman-chen2009,permuter-weissman-goldsmith2009,shrader-permuter2009}, lossy data compression of sequential codes \cite{tatikonda2000,ma-ishwar2011}, lossy data compression with feedforward information at the decoder \cite{venkataramanan-pradhan2007}, and capacity of networks, such as, the two-way channel, the multiple access channel \cite{kramer1998,kramer2003}, etc. Some of the above references derive coding theorems for an anthology of problems of information theory, under any one of the assumptions: $(a)$ stationary ergodic processes $\{(X_i,Y_i):i=0,1,\ldots\}$, $(b)$ Dobrushin's stability of the information density $\sum_{i=0}^n\log\Big(\frac{dP_{Y_i|Y^{i-1},X^i}}{dP_{Y_i|Y^{i-1}}}\Big)$,  $(c)$ Verd\'u and Han's information spectrum methods \cite{han93}. Moreover, directed information is also utilized in a variety of problems subject to causality constraints, such as, gambling, portfolio theory, data compression and hypothesis testing \cite{permuter-kim-weissman2011}, in biology as an alternative to Granger's measure of causality \cite{solo2008,amblard-michel2011,quinn-coleman-kiyavash-hatsopoulos2011}, and in relating Bayesian filtering theory to sequential and nonanticipative RDF \cite{charalambous-stavrou-ahmed2014ieeetac,charalambous-stavrou2014ecc}. 
\vspace*{0.2cm}\\
 Directed information is initially introduced by Marko \cite{marko} by decomposing Shannon's self-mutual information into two directional parts, and then taking  expectation. Although, directed information is defined via a sequence of  conditional mutual informations (i.e., (\ref{equation1a})), for general abstract alphabets (i.e., continuous) or distributions which are not necessarily continuous (i.e., induced by mixture of continuous and finite alphabet RVs) its functional and topological properties are not well understood \cite{kramer1998}. \\
Further, for such alphabet spaces or distributions, specific functional properties of mutual information expressed as a functional $I(X^n;Y^n)\equiv\mathbb{I}_{X^n;Y^n}(P_{X^n},P_{Y^n|X^n})$, of the two distributions $\{P_{X^n},P_{Y^n|X^n}\}$, such as, convexity, concavity, and topological properties such as lower semicontinuity (with respect to the topology of weak convergence of probability measures), at first glance, do not translate into analogous properties for directed information. The reason is that directed information  $I(X^n\rightarrow{Y}^n)\equiv\mathbb{I}_{X^n\rightarrow{Y}^n}(P_{X_i|X^{i-1},Y^{i-1}},P_{Y_i|Y^{i-1},X^i}:i=0,1,\ldots,n)$ is a functional of two sequences of distributions $\{P_{X_i|X^{i-1},Y^{i-1}},P_{Y_i|Y^{i-1},X^i}:i=0,1,\ldots,n\}$, and the joint and marginal distributions are induced from these sequences of distributions. Such properties are important in extremum problems of directed information.  \\ Similarly, it is not obvious whether  the well-known variational equalities of mutual information, which involve a single maximization or minimization of  appropriate functionals over appropriate convex sets,  have counter parts, for directed information, which involve nested maximization and minimization operations of  appropriate functionals over appropriate convex sets, giving rise to sequential variational equalities.  Such sequential variational equalities, are important to develop computationally efficient  sequential algorithms to compute capacity of channels with memory and feedback, similar to  the Blahut-Arimoto algorithm \cite{blahut1987}, of memoryless channels.

These properties together with compactness of subsets of the sets of the conditional distributions $\{P_{X_i|X^{i-1},Y^{i-1}}:i=0,1,\ldots,n\}$ and $\{P_{Y_i|Y^{i-1},X^i}:i=0,1,\ldots,n\}$, are fundamental to analyze extremum problems of directed information related to channel capacity, sequential and nonanticipative RDF, their generalizations to networks, etc, for countable and abstract alphabets.\\
\noi  Recently, in \cite{ho-yeung2009ieeeit} it is demonstrated via several examples that Shannon information measures, such as, entropy, relative entropy, mutual information, and conditional mutual information, when defined on countable alphabets, are discontinuous with respect to strong topologies (i.e., induced by total variational distance metrics on the space of probability distributions). Since directed information in \eqref{equation1a} involves a sequence of conditional mutual informations, the observations in \cite{ho-yeung2009ieeeit} also apply to directed information. The lack of continuity is attributed to the fact that mutual information and directed information  are defined from relative entropy, and relative entropy is lower semicontinuous with respect to distributions \cite{ihara1993}. For such abstract alphabets problems, it was recognized many years ago (see \cite{csiszar92,fozunbal}) that the analysis of capacity formulae based on single letter mutual information formulae requires tools from the topology of weak convergence of probability measures (or equivalently the weak$^*$ topology), in order to identify global and local analytical properties of channel input distributions which maximize mutual information. 
\vspace*{0.2cm}\\  
The main objective of this paper is to derive functional properties, topological properties, and sequential variational equalities, for directed information, when the distributions are defined on abstract alphabets, and to provide appropriate conditions for these to hold. The methodology and the main results are summarized below.
\begin{itemize}
\item[R1)] Introduce an equivalent directed information definition expressed via information  divergence $\mathbb{D}(\cdot||\cdot)$, as a functional of two consistent families of conditional distributions ${\bf P}(\cdot|{\bf y})$  on ${\cal X}^{\mathbb{N}_0}\tri\times_{i=0}^{\infty}{\cal X}_i$ parametrized by ${\bf y}=(y_0,y_1,\ldots)\in{\cal Y}^{\mathbb{N}_0}\tri\times_{i=0}^{\infty}{\cal Y}_i$, and ${\bf Q}(\cdot|{\bf x})$ on ${\cal Y}^{\mathbb{N}_0}$ parametrized by ${\bf x}\in{\cal X}^{\mathbb{N}_0}$, which uniquely define $\{P_{X_i|X^{i-1},Y^{i-1}}:i\in\mathbb{N}_0\}$ and $\{P_{Y_i|Y^{i-1},X^i}:i\in\mathbb{N}_0\}$, respectively, and vice-versa, and their $(n+1)$-fold compound probability distributions $\overleftarrow{P}_{0,n}(dx^n|y^{n-1})\triangleq\otimes_{i=0}^n{P}_{X_i|X^{i-1},Y^{i-1}}$ $(dx_i|x^{i-1},y^{i-1})$, $\overrightarrow{Q}_{0,n}(dy^n|x^{n})\triangleq\otimes_{i=0}^n{P}_{Y_i|Y^{i-1},X^{i}}(dy_i|y^{i-1},x^i)$.
\item[R2)] Show convexity of the consistent families of the conditional distributions ${\bf P}(\cdot|{\bf y})$ for ${\bf y}\in{\cal Y}^{\mathbb{N}_0}$, ${\bf Q}(\cdot|{\bf x})$ for ${\bf x}\in{\cal X}^{\mathbb{N}_0}$.
\item[R3)] Show convexity and concavity of directed information as a functional with respect to the consistent families of conditional distributions ${\bf Q}(\cdot|{\bf x})$ for ${\bf x}\in{\cal X}^{\mathbb{N}_0}$, and ${\bf P}(\cdot|{\bf y})$ for ${\bf y}\in{\cal Y}^{\mathbb{N}_0}$, respectively.
\item[R4)] Show under certain conditions, weak compactness of the consistent families of conditional distributions ${\bf P}(\cdot|{\bf x})$ for ${\bf x}\in{\cal X}^{\mathbb{N}_0}$, and ${\bf Q}(\cdot|{\bf y})$ for ${\bf y}\in{\cal Y}^{\mathbb{N}_0}$, and of their marginals and joint distribution.
\item[R5)] Show lower semicontinuity of directed information as a functional of the consistent families of the conditional distributions ${\bf P}(\cdot|{\bf y})$ for ${\bf y}\in{\cal Y}^{\mathbb{N}_0}$, and ${\bf Q}(\cdot|{\bf x})$ for ${\bf x}\in{\cal X}^{\mathbb{N}_0}$, and under certain conditions, continuity of directed information as a functional of the family ${\bf P}(\cdot|{\bf y})$ for ${\bf y}\in{\cal Y}^{\mathbb{N}_0}$.
\item[R6)] Express directed information in terms of variational equalities involving sequential minimization and sequential maximization operations over conditional distributions.
\item[R7)] Illustrate that R1)--R6) extend naturally to three sequences of RV's $X^n\in{\cal X}_{0,n}$, $Y^n\in{\cal Y}_{0,n}$, $Z^n\in{\cal Z}_{0,n}$, or more, which cover directed information measures for networks, and possible problems with side information.
\item[R8)] Discuss applications of R1)-R6).
\end{itemize}
\vspace*{0.2cm}
\noi The above functional and topological properties are shown by invoking the topology of weak convergence of probability measures on Polish spaces and Prohorov's theorems \cite{dupuis-ellis97,billingsley1999}. Some of the results described above are obtained by utilizing analogies between communication channels with memory and feedback, and stochastic optimal control problems in which the control element and the controlled element are the sequences of conditional distributions, $\{P_{X_i|X^{i-1},Y^{i-1}}:i=0,1,\ldots\}$ and $\{P_{Y_i|Y^{i-1},X^i}:i=0,1,\ldots\}$, respectively, \cite{gihman-skorohod1979,bertsekas-shreve2007}.
\vspace*{0.2cm}\\
Items R1)-R7) extend various functional and topological properties of mutual information $I(X^n;Y^n)\equiv\mathbb{I}_{X^n;Y^n}(P_{X^n},P_{Y^n|X^n})$ as a functional of $\{P_{X^n},P_{Y^n|X^n}\}$ to directed information. \\
\noi From the practical point of view, there are many potential applications of R1)-R7). Below, we briefly discuss some of them.\\
\noi The concavity and convexity properties are important in deriving tight bounds in applications of converse coding theorems, in identifying properties of extremum problems involving feedback capacity \cite{kramer1998,tatikonda2009} and sequential and nonanticipative lossy data compression via the nonanticipative RDF \cite{stavrou-charalambous-kourtellaris2013a}, in relating Bayesian filtering theory and nonanticipative RDF \cite{charalambous-stavrou-ahmed2014ieeetac}, in network communication applications \cite{kramer2007,gamal-kim2011}, etc. The semicontinuity and continuity of directed information, and the compactness of the consistent families of distributions ${\bf P}(\cdot|{\bf y})$ for ${\bf y}\in{\cal Y}^{\mathbb{N}_0}$, and ${\bf Q}(\cdot|{\bf x})$ for ${\bf x}\in{\cal X}^{\mathbb{N}_0}$, are crucial, when addressing questions of existence of extremum solutions to problems involving feedback capacity, sequential and nonanticipative lossy data compression, computations of extremum solutions and their properties, and in extending existing coding theorems to abstract alphabets \cite{berger1968}. For example, the converse part of coding theorem for feedback capacity presupposes existence of optimal channel input distribution maximizing directed information, and existence of its per unit time limit. The variational equalities are important in generalizing Blahut-Arimoto computation schemes of single letter mutual information expressions \cite{blahut1972} to sequential Blahut-Arimoto schemes, involving extremum problems of directed information, such as, in problems of evaluating feedback capacity (see \cite{stavrou-charalambous-tzortzis2015}). 
\vspace*{0.2cm}\\
\noi Throughout the paper, we illustrate applications of the results to the following extremum problems. 

\noi{\bf Capacity of channels with memory and feedback.} Consider the extremum problem of channel capacity with memory and feedback. Under the assumption of stationary ergodic processes $\{(X_i,Y_i):~i=0,1,\ldots\}$ or Dobrushin's directed information stability and transmission cost stability, the operational definition of capacity is given by the following extremum problem \cite{tatikonda-mitter2009}.  
\begin{align}
C^{fb}(P)\tri\liminf_{n\rightarrow\infty}\sup_{\{P_{X_i|X^{i-1},Y^{i-1}}:~i=0,1,\ldots,n\}\in{{\cal P}_{0,n}(P)}}\frac{1}{n+1}I(X^n\rightarrow{Y^n}),\label{equation1dd}
\end{align}
where ${\cal P}_{0,n}(P)$ is the transmission cost constraint set defined by
\begin{align}
{\cal P}_{0,n}(P)&\tri\bigg\{{P}_{X_i|X^{i-1},Y^{i-1}},~i=0,1,\ldots,n:~\frac{1}{n+1}\mathbb{E}\big\{c_{0,n}(x^n,y^{n-1})\big\}\leq{P}\bigg\},~P\geq{0}\label{introduction:equation1ddd}
\end{align}
and $c_{0,n}:{\cal X}_{0,n}\times{\cal Y}_{0,n-1}\longmapsto [0,\infty), c_{0,n}(x^n,y^{n-1})\tri\sum_{i=0}^{n}{g}_{i}(x^i,y^{i-1})$ is a measurable function denoting the cost of transmitting symbols over the channel. 

\noi The task of showing existence of a sequence of probability distributions $\{P_{X_i|X^{i-1},Y^{i-1}}:i=0,1,\ldots,n\}\in{\cal P}_{0,n}(P)$ which achieves the supremum in (\ref{equation1dd}) for continuous or countable alphabet spaces is not easy. The main difficulty arises from the fact that $I(X^n\rightarrow{Y^n})$ is a functional of the two sequences of distributions $\{P_{X_i|X^{i-1},Y^{i-1}},P_{Y_i|Y^{i-1},X^i}:i=0,1,\ldots,n\}$, unlike mutual information $I(X^n;Y^n)\equiv\mathbb{I}_{X^n;Y^n}(P_{X^n},P_{Y^n|X^n})$, which inherits most of its properties from those of relative entropy between the joint distribution $P_{Y^n,X^n}$ and the product of its marginals  $P_{X^n} \times P_{Y^n}$. However, we show by utilizing some of the results described under  R1)--R6),  existence of such conditional distribution and identify several properties of the optimal conditional channel input distribution.\\%not only for (\ref{equation1dd}), but also for the sequential and nonanticipative rate distortion functions, obtain  generalizations involving three or more random processes, and identify the properties of the extremum solutions.

\noi{\bf Generalized Information Nonanticipative or Sequential RDF.} Consider the extremum problem of general information nonanticipative RDF, or sequential RDF \cite{tatikonda2000}, which is a variant of classical RDF \cite{berger}, defined by \cite{charalambous-stavrou-ahmed2014ieeetac,stavrou-kourtellaris-charalambous2015ieeeit} 
\begin{align}
R^{na}(D)\tri\limsup_{n\rightarrow\infty}\inf_{\big\{P_{Y_i|Y^{i-1},X^{i}},~i=0,1,\ldots,n\big\}\in{{\cal Q}_{0,n}(D)}}\frac{1}{n+1}I(X^n\rightarrow{Y^n}),\label{introduction:nrdf:equation1}
\end{align}
where ${\cal Q}_{0,n}(D)$ is the fidelity constraint set defined by
\begin{align}
{\cal Q}_{0,n}(D)\tri\bigg\{{Q}_{Y_i|Y^{i-1},X^{i}},~i=0,1,\ldots,n:~\frac{1}{n+1}\mathbb{E}\big\{d_{0,n}(x^n,y^{n})\big\}\leq{D}\bigg\},~D\geq{0}\label{introduction:nrdf:equation2}
\end{align}
and $d_{0,n}:{\cal X}_{0,n}\times{\cal Y}_{0,n}\longmapsto [0,\infty],~d_{0,n}(x^n,y^{n})\tri\sum_{i=0}^{n}{\rho}_{i}(x^i,y^{i})$ is a measurable function denoting the distortion function of reconstructing $x_i$ by $y_i$, $i=0,1,\ldots,n$. Note that if $P_{X_i|X^{i-1},Y^{i-1}}=P_{X_i|X^{i-1}},~a.a. (x^{i-1,y^{i-1}}),~i=0,1,\ldots,n$, then it can be shown that \eqref{introduction:nrdf:equation1}, \eqref{introduction:nrdf:equation2}  are degraded to Gorbunov and Pinsker's nonanticipatory $\epsilon$-entropy \cite{gorbunov-pinsker1973}.

\noi For both extremum problems (\ref{equation1dd}), (\ref{introduction:nrdf:equation1}), we illustrate applications of R1)--R6) in showing existence of solutions, identifying properties of optimal solutions, and in constructing sequential versions of Blahut Arimoto Algorithm (BAA) \cite{blahut1972}.
\vspace*{0.2cm}\\
The rest of the paper is structured as follows. Section~\ref{equivalent_definitions} introduces two equivalent definitions of nonanticipative channels on abstract spaces \big(R1)\big). Section~\ref{properties} derives the functional and topological properties of directed information \big(R2)--R5)\big). Section~\ref{variational} derives sequential variational equalities of directed information \big(R6)\big). 
%%%%%%%%%%%%%%%%%%%%%%%%%%%%%%%%%%%%%%%%%%%%%%%%%%%%%%%%%%%%%%%%%%%%%%%%%%%%%%%%%%%%%%

\section{Equivalent Nonanticipative Channels on Abstract  Spaces}\label{equivalent_definitions}
\par In this section, our aim is to establish two equivalent definitions of the sequence of conditional distributions or basic processes, which define any probabilistic channel with nonanticipative (causal) feedback, that relate causally the input-output behavior of the channel. This formulation is utilized extensively to establish the results stated under R1)--R7). The first definition of conditional distributions is the usual one found in many papers, e.g., \cite{kramer1998,tatikonda2000,permuter-cuff-vanroy-weissman2008,tatikonda-mitter2009,permuter-weissman-goldsmith2009,permuter-weissman-chen2009}, for finite alphabets spaces. The aforementioned definition is described via a family of multi-fold compound conditional distributions (see Fig.~\ref{Fig1}, (a)). The second definition is described via a family of conditional distributions defined on product alphabets, which satisfy a certain consistency condition (see Fig.~\ref{Fig1}, (b)). 
\begin{figure*}[htp]
\centering
\subfigure[Sequence of feedback and feedforward channels $\{P_{X_i|X^{i-1},Y^{i-1}},P_{Y_i|Y^{i-1},X^i}:i=0,1,\ldots,n\}$.]
{\includegraphics[scale=0.50]{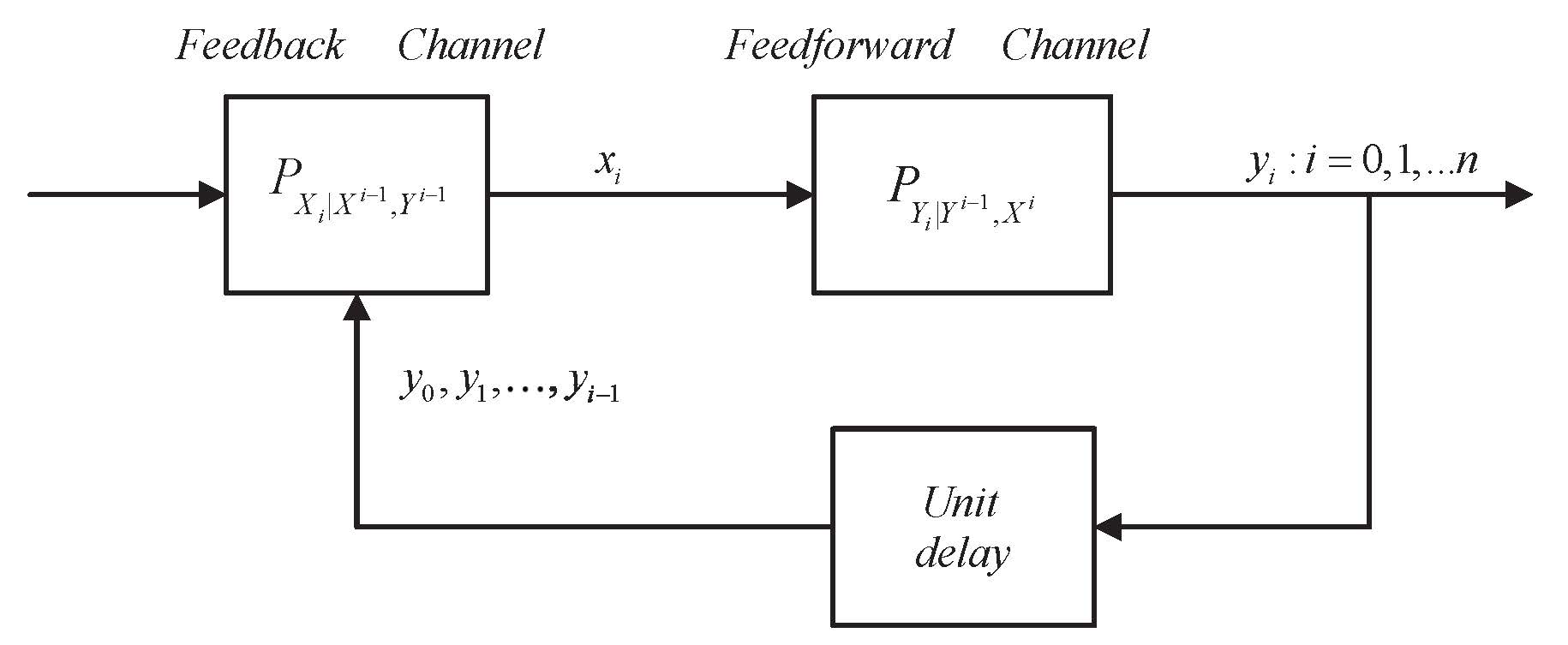}}\quad
\subfigure[Consistent families of feedback and feedforward channels $\{\protect \overleftarrow {P}_{X^n|Y^{n-1}},\protect\overrightarrow{Q}_{Y^n|X^n}:n\in{\mathbb N}_0\}$.]
{\includegraphics[scale=0.50]{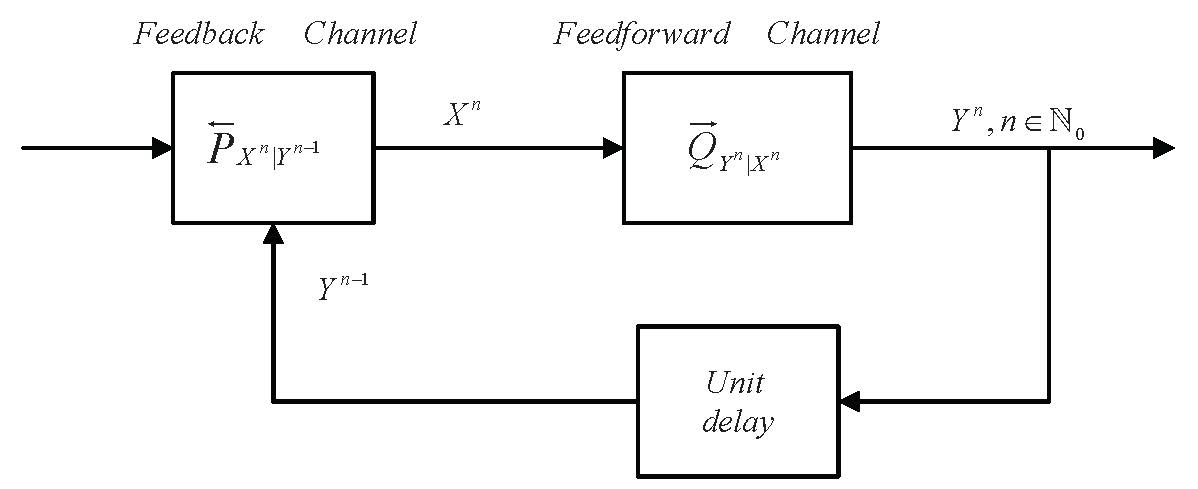}}
\caption{Equivalent Representations of Feedback/Feedforward Channels.}
\label{Fig1}
\end{figure*}
The second definition is often utilized in the stochastic control literature, in which there is a control process and a controlled process \cite{gihman-skorohod1979,bertsekas-shreve2007}. Indeed, the analogy is that $\{X_i:~i=0,1,\ldots\}$ is the control process, $\{Y_i:~i=0,1,\ldots\}$ is the controlled process, $\{P_{X_i|X^{i-1},Y^{i-1}}:~i=0,1,\ldots\}$ is the control element, and $\{P_{Y_i|Y^{i-1},X^{i}}:~i=0,1,\ldots\}$ is the controlled element. The second definition is more convenient, because the directed information density $i(X^n\rightarrow{Y^n})\tri\log\Big(\otimes_{i=0}^n\frac{dP_{Y_i|Y^{i-1},X^i}}{dP_{Y_i|Y^{i-1}}}\Big)=\sum_{i=0}^n\log\Big(\frac{dP_{Y_i|Y^{i-1},X^i}}{dP_{Y_i|Y^{i-1}}}\Big)$ corresponding to $I(X^n\rightarrow{Y^n})$, can be equivalently expressed in terms of two consistent families of conditional distributions, namely, ${\bf Q}(\cdot|{\bf x})$ on ${\cal Y}^{\mathbb{N}_0}$ given ${\bf x}=(x_0,x_1,\ldots)\in{\cal X}^{\mathbb{N}_0}$, and ${\bf P}(\cdot|{\bf y})$ on ${\cal X}^{\mathbb{N}_0}$ given ${\bf y}=(y_0,y_1,\ldots)\in{\cal Y}^{\mathbb{N}_0}$, which uniquely define  $\{P_{Y_i|Y^{i-1},X^i}:i=0,1,\ldots\}$ and $\{P_{X_i|X^{i-1},Y^{i-1}}:i=0,1,\ldots\}$, respectively, and vice-versa, such that
$i(X^n\rightarrow{Y^n})=\log\Big(\frac{d{\bf Q}(\cdot|x^n)}{d\nu^{{\bf P}\otimes{\bf Q}}(\cdot)}(y^n)\Big)-a.s.$, where $\nu^{{\bf P}\otimes{\bf Q}}(\cdot)$ is the marginal distribution on $\times_{i=0}^n{\cal Y}_i$ obtained from ${\bf P}(\cdot|{\bf y})$ and ${\bf Q}(\cdot|{\bf x})$. Once the conditions on the abstract spaces $\{({\cal Y}_i,{\cal X}_i):~i=0,1,\ldots\}$ are identified, and the consistency conditions are introduced, then it can be shown that $i(X^n\rightarrow{Y^n})$ has another version given by $i(X^n\rightarrow{Y^n})=\log\Big(\frac{d({\bf P}(\cdot|\cdot){\otimes}{\bf Q}(\cdot|\cdot))}{d({\bf P}(\cdot|\cdot)\otimes\nu^{{\bf P}\otimes{\bf Q}}(\cdot))}(x^n,y^n)\Big)-a.s.$, where $\otimes$ denotes the compound probability distribution defined by ${\bf P}(\cdot|\cdot)$ and ${\bf Q}(\cdot|\cdot)$, and similarly for the rest of the measures. Consequently, directed information can be expressed in terms of Kullback-Leibler distance $\mathbb{D}\big({\bf P}{\otimes}{\bf Q}||{\bf P}\otimes\nu^{{\bf P}\otimes{\bf Q}}\big)$\footnote{In the rest of the paper we write $\nu$ instead of $\nu^{{\bf P}\otimes{\bf Q}}$ omitting its explicit dependence on ${\bf P}(\cdot|{\bf y})$ and ${\bf Q}(\cdot|{\bf x})$.}. 
\vspace*{0.2cm}\\
\noi{\bf Notations and Preliminaries.}\\
Denote the set of nonnegative integers by $\mathbb{N}_0 \tri \{0,1,2,\ldots\},$ and the restriction of $\mathbb{N}_0$ to positive integers by $\mathbb{N}_1 \tri \{1,2,\ldots\}$, and to a finite set by $\mathbb{N}_0^n \tri \{0,1,2,\ldots,n\}$. Introduce two sequences of spaces $\{({\cal X}_n,{\cal B}({\cal X }_n)):n\in\mathbb{N}_0\}$ and $\{({\cal Y}_n,{\cal B}({\cal Y}_n)):n\in\mathbb{N}_0\},$ called basic measurable spaces, where ${\cal X}_n,{\cal Y}_n, n\in\mathbb{N}_0$ are topological spaces, and ${\cal B}({\cal X}_n)$ and ${\cal B}({\cal Y}_n)$ are Borel $\sigma-$algebras of subsets of ${\cal X}_n$ and ${\cal Y}_n,$ respectively. The set of probability measures on any measurable space $({\cal Z},{\cal B}({\cal Z}))$ is denoted by ${\cal M}_1({\cal Z})$.\\
For each $n\in\mathbb{N}_0$ define the product spaces
\begin{align*}
({\cal X}_{0,n},{\cal B}({\cal X}_{0,n}))\tri(\times_{i=0}^n{\cal X}_i,\otimes_{i=0}^n{\cal B}({\cal X}_i)), ({\cal Y}_{0,n},{\cal B}({\cal Y}_{0,n}))\tri(\times_{i=0}^n{\cal Y}_i,\otimes_{i=0}^n{\cal B}({\cal Y}_i)).
\end{align*}
\noi For each $n\in\mathbb{N}_0$, let ${\cal X}_n$ and ${\cal Y}_n$ be the spaces of all possible outcomes. Given the data up to and including the $n$th time, specifically, $(x_i,y_i)\in{\cal X}_i\times{\cal Y}_i,~i=0,1,\ldots,n$, the probability distributions at time $(n+1)$ are $p_{n+1}(A_{n+1}|x_0,\ldots,x_n,y_0,\ldots,y_n)$ and $q_{n+1}(B_{n+1}|y_0,\ldots,y_n,x_0,\ldots,x_{n+1})$, $A_{n+1}\in{\cal B}({\cal X}_{n+1})$, $B_{n+1}\in{\cal B}({\cal Y}_{n+1})$.
Hence, each possible outcome of the experiment is a sequence $\omega=(x_0,y_0,x_1,y_1,\ldots)$ with $x_n\in{\cal X}_n, y_n\in{\cal Y}_n$ for each $n\in\mathbb{N}_0$ (here, no time ordering is required).\\
Consequently, define the sample space $\Omega$ and the algebra ${\cal F}$ of all experiments by 
\begin{align*}
(\Omega,{\cal F})\tri\Big(\times_{n\in\mathbb{N}_0}({\cal X}_n\times{\cal Y}_n),\otimes_{n\in\mathbb{N}_0}\big({\cal B}({\cal X}_n)\otimes{\cal B}({\cal Y}_n)\big)\Big).
\end{align*}
Associated with the basic measurable spaces there are two basic sequences of Random Variables (RV's) $\{X_n:n\in\mathbb{N}_0\}$ and $\{Y_n:n\in\mathbb{N}_0\},$ such that for each $n\in\mathbb{N}_0,$ they take values $X_n\in{\cal X}_n$ and $Y_n\in{\cal Y}_n$. These are introduced as follows.\\
Let $X_0,Y_0,X_1,Y_1,\ldots$ be the coordinate RV's. For each $n\in\mathbb{N}_0$
\begin{align}
X_n(\omega)=x_n,~Y_n(\omega)=y_n~~\mbox{if}~~\omega=(x_0,y_0,x_1,y_1,\ldots).\nonumber
\end{align}
Clearly, $X_n:(\Omega,{\cal F})\longmapsto({\cal X}_n,{\cal B}({\cal X}_n))$, $Y_n:(\Omega,{\cal F})\longmapsto({\cal Y}_n,{\cal B}({\cal Y}_n))$, and for each outcome $\omega\in\Omega$ of the experiment, $X_n(\omega)$, $Y_n(\omega)$ are the results of the $n$th time. Similarly, $X^n\tri\{X_0,\ldots,X_n\}$ and $Y^n\tri\{Y_0,\ldots,Y_n\}$ denote the result of the trials up to and including the $n$th time; they are RV taking values in $({\cal X}_{0,n},{\cal B}({\cal X}_{0,n}))$ and $({\cal Y}_{0,n},{\cal B}({\cal Y}_{0,n}))$, respectively. The objective is to construct a measure $\mathbb{P}$ on $(\Omega,{\cal F})$ consistent with the data (e.g., measurable spaces and conditional distributions).\\
For every $n\in\mathbb{N}_0$, define the $\sigma$-algebras generated by $\{X_0,X_1,\ldots,X_n\}$ and $\{Y_0,Y_1,\ldots,Y_n\}$ by 
\begin{align}
{\cal F}(X^n)\tri\sigma\{X_0,X_1,\ldots,X_n\},~{\cal F}(Y^n)\tri\sigma\{Y_0,Y_1,\ldots,Y_n\}.\nonumber
\end{align}
Then every event $H\in{\cal F}(X^n)$ has the form 
\begin{align}
H=\Big\{(X_0,X_1,\ldots,X_n)\in{A}\Big\}=A\times{\cal X}_{n+1}\times{\cal X}_{n+2}\ldots,~A\in{\cal B}({\cal X}_{0,n})\nonumber
\end{align}
and $H$ is called a cylinder set with base $A\in{\cal B}({\cal X}_{0,n})$. Similarly, for an event $J\in{\cal F}(Y^n)$ 
\begin{align}
J=\Big\{(Y_0,Y_1,\ldots,Y_n)\in{B}\Big\}=B\times{\cal Y}_{n+1}\times{\cal Y}_{n+2}\ldots,~B\in{\cal B}({\cal Y}_{0,n})\nonumber
\end{align}
and $J$ is a cylinder set with base $B\in{\cal B}({\cal Y}_{0,n})$.\\
Points in the Cartesian countable product spaces ${\cal X}^{\mathbb{N}_0}\tri{\times_{n\in\mathbb{N}_0}}{\cal X}_n,$ ${\cal Y}^{\mathbb{N}_0}\tri{\times_{n\in\mathbb{N}_0}}{\cal Y}_n$ are denoted by ${\bf x}\tri\{x_0,x_1,\ldots\}\in{\cal X}^{\mathbb{N}_0},$ ${\bf y}\tri\{y_0,y_1,\ldots\}\in{\cal Y}^{\mathbb{N}_0},$ respectively. Similarly, for $n\in\mathbb{N}_0,$ points in ${\cal X}_{0,n}\tri\times^n_{i=0}{\cal X}_i,$ ${\cal Y}_{0,n}\tri\times^n_{i=0}{\cal Y}_i$  are denoted by $x^n\tri\{x_0,x_1,\ldots,x_n\}\in{\cal X}_{0,n},$ $y^n\tri\{y_0,y_1,\ldots,y_n\}\in{\cal Y}_{0,n},$ respectively.\\
Let ${\cal B}({\cal X}^{\mathbb{N}_0})$ and ${\cal B}({\cal Y}^{\mathbb{N}_0})$ denote the $\sigma-$algebras in ${\cal X}^{\mathbb{N}_0}$ and ${\cal Y}^{\mathbb{N}_0},$ respectively, generated by cylinder sets (e.g., ${\cal B}({\cal X}^{\mathbb{N}_0})$ is the smallest Borel $\sigma-$algebra containing all cylinder sets $\{{\bf x}=(x_0,x_1,\ldots)\in{\cal X}^{\mathbb{N}_0}:x_0\in{A}_0,x_1\in{A}_1,\ldots,x_n\in{A}_n\}, A_i\in{\cal B}({\cal X}_i), {i}\in\mathbb{N}_0^n$). The Borel $\sigma$-algebra ${\cal B}({\cal X}^{\mathbb{N}_0})$ is denoted by $\otimes_{i\in\mathbb{N}_0}{\cal B}({\cal X}_i)$. Hence, ${\cal B}({\cal X}_{0,n})$ and ${\cal B}({\cal Y}_{0,n})$ denote the $\sigma-$algebras of cylinder sets in ${\cal X}^{\mathbb{N}_0}$ and ${\cal Y}^{\mathbb{N}_0},$ respectively, with bases over $A_i\in{\cal B}({\cal X}_i),~i\in\mathbb{N}_0^n$, and $B_i\in{\cal B}({\cal Y}_i),~i\in\mathbb{N}_0^n$, respectively.
\vspace*{0.2cm}\\
\noi{\bf Backward or Feedback Channel.}\\
Suppose for each $n\in\mathbb{N}_0,$ the conditional distribution of the RV $X_n\in{\cal X}_n$ is determined provided the values of the basic processes $X^{n-1}=x^{n-1}\in{\cal X}_{0,n-1}$ and $Y^{n-1}=y^{n-1}\in{\cal Y}_{0,n-1}$ are known, and let $\{p_n(dx_n|x^{n-1},y^{n-1}):n\in\mathbb{N}_0\}$ denote the collection of these distributions. At $n=0$, the distribution is $p_0(dx_0|x^{-1},y^{-1})$, where $(x^{-1},y^{-1})$ are either fixed, or $p_0(dx_0|x^{-1},y^{-1})={p}(dx_0)$, depending on the convention used. Without loss of generality, we assume $p_0(dx_0|x^{-1},y^{-1})\tri{p}_0(x_0)$ (i.e., $\sigma\{X^{-1},Y^{-1}\}=\{\emptyset,\Omega\}$). For each $n\in\mathbb{N}_0$, the functions $p_n(\cdot|\cdot,\cdot): {\cal X}_n\times{\cal X}_{0,n-1}\times{\cal Y}_{0,n-1}\longmapsto[0,1]$ are candidates of distributions of the sequence of RV's $\{X_n:n\in\mathbb{N}_0\}$ on $\{({\cal X}_n,{\cal B}({\cal X }_n)):n\in\mathbb{N}_0\}$ if and only if the following conditions hold.
\vspace*{0.2cm}\\
{\bf i)} For every $n\in\mathbb{N}_0$, and $x^{n-1}\in{\cal X}_{0,n-1}$, $y^{n-1}\in{\cal Y}_{0,n-1}$, $p_n(\cdot|x^{n-1},y^{n-1})$ is a probability measure on ${\cal B}({\cal X }_n);$\\
{\bf ii)} For every $n\in\mathbb{N}_0$, and $A_n\in{\cal B}({\cal X}_n)$, $p_n(A_n|\cdot,\cdot)$ is an $\otimes^{n-1}_{i=0}\big({\cal B}({\cal X }_i)\otimes{\cal B}({\cal Y}_i)\big)$-measurable function of $x^{n-1}\in{\cal X}_{0,n-1},$ $y^{n-1}\in{\cal Y}_{0,n-1}$.
\vspace*{0.2cm}\\
\noi For every $n\in\mathbb{N}_0$, the set of all functions that satisfy {\bf i)}, {\bf ii)}, are called {\it stochastic kernels} on ${\cal X}_n$ given ${\cal X}_{0,n-1}\times{\cal Y}_{0,n-1}$, and these are denoted by 
\begin{align*}
{\cal Q}({\cal X}_n|{\cal X}_{0,n-1}\times{\cal Y}_{0,n-1})\tri\big\{p_n(\cdot|x^{n-1},y^{n-1})\in{\cal M}_1({\cal X}_n):~x^{n-1}\in{\cal X}_{0,n-1}, y^{n-1}\in{\cal Y}_{0,n-1}~\mbox{and}~{\bf ii)}~\mbox{holds}\big\}.
\end{align*}
Given the collection of functions $\{p_n(\cdot|\cdot,\cdot):n\in\mathbb{N}_0\}$ satisfying conditions {\bf i)}, {\bf ii)}, one can construct a family of measures on the product space $({\cal X}^{\mathbb{N}_0},{\cal B}({\cal X}^{\mathbb{N}_0}))\tri\big(\times_{i\in\mathbb{N}_0}{\cal X}_i,\otimes_{i\in\mathbb{N}_0}{\cal B}({\cal X }_i)\big)$ as follows.\\
Let $C\in{\cal B}({\cal X}_{0,n})$ be a cylinder set of the form
\begin{align}
C\tri\Big\{{\bf x}\in{\cal X}^{\mathbb{N}_0}:x_0\in{C_0},x_1\in{C_1},\ldots,x_n\in{C_n}\Big\},~C_i\in{\cal B}({\cal X }_i),~{i}\in\mathbb{N}_0^n, C_{0,n}=\times_{i=0}^n{C_i}.\nonumber
\end{align}
Define a family of measures ${\bf P}(\cdot|{\bf y})$ parametrized by ${\bf y}\in{\cal Y}^{\mathbb{N}_0}$ on ${\cal B}({\cal X}^{\mathbb{N}_0})$ by
\begin{align}
{\bf P}(C|{\bf y})&\tri\int_{C_0}p_0(dx_0)\int_{C_1}p_1(dx_1|x_0,y_0)\ldots\int_{C_n}p_n(dx_n|x^{n-1},y^{n-1})\label{equation2}\\
&\equiv{\overleftarrow{P}}_{0,n}(C_{0,n}|y^{n-1}).\label{equation4a}
\end{align}
The notation ${\overleftarrow{P}}_{0,n}(\cdot|y^{n-1})$ is used to denote the causal conditioning dependence of the measure ${\bf P}(\cdot|{\bf y})$ defined on cylinder sets $C\in{\cal B}({\cal X}_{0,n})$, for any $n\in\mathbb{N}_0$. 
The right hand side (RHS) of (\ref{equation2}) uniquely defines a measure on $({\cal X}^{\mathbb{N}_0},{\cal B}({\cal X}^{\mathbb{N}_0}))$. Moreover, for each $n\in\mathbb{N}_0$ the family of measures ${\bf P}(\cdot|{\bf y})$ parametrized by ${\bf y}\in{\cal Y}^{\mathbb{N}_0}$, satisfies the following property (inherited from condition {\bf ii)}): for $E\in{\cal B}({\cal X}^{\mathbb{N}_0}),$ ${\bf P}(E|\cdot)$ is ${\cal B}({\cal Y}^{\mathbb{N}_0})-$measurable, and for $E\in{\cal B}({\cal X}_{0,n}),$ ${\bf P}(E|\cdot)$ is ${\cal B}({\cal Y}_{0,n-1})-$measurable.\\ 
Thus, if conditions {\bf i)} and {\bf ii)} hold then for each ${\bf y}\in{\cal Y}^{\mathbb{N}_0},$ the RHS of (\ref{equation2}) defines a consistent family of finite-dimensional distribution, and hence there exists a unique measure on $({\cal X}^{\mathbb{N}_0},{\cal B}({\cal X}^{\mathbb{N}_0})),$ for which $p_n(dx_n|x^{n-1},y^{n-1})$ is obtained. This leads to the first definition of a feedback channel, as a family of functions $\{p_n(\cdot|\cdot,\cdot)\in{\cal Q}({\cal X}_n|{\cal X}_{0,n-1}\times{\cal Y}_{0,n-1}):~n\in\mathbb{N}_0\}$, i.e., satisfying conditions {\bf i)} and {\bf ii)}. This definition is used extensively by many authors \cite{kramer1998,tatikonda2000,permuter-cuff-vanroy-weissman2008,tatikonda-mitter2009,permuter-weissman-chen2009,permuter-weissman-goldsmith2009}, when the alphabet spaces have finite cardinality.
\vspace*{0.2cm}\\
 An alternative, equivalent definition of a feedback channel is established as follows. Consider a family of measures ${\bf P}(\cdot|{\bf y})$ on $({\cal X}^{\mathbb{N}_0},{\cal B}({\cal X}^{\mathbb{N}_0}))$ parametrized by ${\bf y}\in{\cal Y}^{\mathbb{N}_0}$ satisfying the following consistency condition.
\vspace*{0.2cm}\\
\noi{\bf C1}:~~If $E\in{\cal B}({\cal X}_{0,n})$ then ${\bf P}(E_{0,n}|\cdot)$ is ${\cal B}({\cal Y}_{0,n-1})-$measurable function of ${\bf y}\in{\cal Y}^{\mathbb{N}_0}$.
\vspace*{0.2cm}\\
Clearly, if conditions {\bf i)} and {\bf ii)} are satisfied, then the family of measures ${\bf P}(\cdot|{\bf y})$ defined via the RHS of (\ref{equation2}) satisfies consistency condition {\bf C1}. The question we address next is whether for any family of measures ${\bf P}(\cdot|{\bf y})$ on $({\cal X}^{\mathbb{N}_0},{\cal B}({\cal X}^{\mathbb{N}_0}))$ parametrized by ${\bf y}\in{\cal Y}^{\mathbb{N}_0}$, satisfying consistency condition {\bf C1}, one can construct a collection of functions $\{p_n(\cdot|\cdot,\cdot)\in{\cal Q}({\cal X}_n|{\cal X}_{0,n-1}\times{\cal Y}_{0,n-1}):n\in\mathbb{N}_0\}$, i.e., satisfying conditions {\bf i)} and {\bf ii)}, which are connected to ${\bf P}(\cdot|{\bf y})$ via relation (\ref{equation2}). To illustrate this point, let $A^{(n)}=\{{\bf x}\in{\cal X}^{\mathbb{N}_0}:x_n{\in}A\},$ $A\in{\cal B}({\cal X}_n),$ and let ${\bf P}(A^{(n)}|{\cal B}({\cal X}_{0,n-1})|{\bf y})$ denote the conditional probability of $A^{(n)}$ with respect to ${\cal B}({\cal X}_{0,n-1})$ calculated on the probability space $\big({\cal X}^{\mathbb{N}_0},{\cal B}({\cal X}^{\mathbb{N}_0}),{\bf P}(\cdot|{\bf y})\big)$. Then
\begin{align}
{\bf P}(A^{(n)}|{\cal B}({\cal X}_{0,n-1})|{\bf y})=p_n(A|x^{n-1},y^{n-1}),~~A^{(n)}\in{\cal B}({\cal X}_{0,n}),\label{equation3a}
\end{align}
for ${\bf P}(\cdot|{\bf y})-$almost all ${\bf x}\in{\cal X}^{\mathbb{N}_0}.$ Clearly, the function on the RHS of (\ref{equation3a}), $p_n(A|x^{n-1},y^{n-1})$ is ${\cal B}({\cal X}_{0,n-1})$-measurable for a fixed $A\in{\cal B}({\cal X}_n)$ and $y^{n-1}\in{\cal Y}_{0,n-1}$, but it cannot be claimed that $p_n(\cdot|x^{n-1},y^{n-1})$ is a probability measure on ${\cal X}_n$. However, under the general assumption that $\{({\cal X}_n,{\cal B}({\cal X }_n)):n\in\mathbb{N}_0\}$ are complete separable metric spaces (Polish spaces), with ${\cal B}({\cal X}_n)$ the $\sigma-$algebra of Borel sets, it is shown in \cite{gihman-skorohod1979}, that the RHS of (\ref{equation3a}) represents a version of conditional probability $(a.s.)$, i.e., condition {\bf i)} holds as well. Therefore, to establish the second equivalent definition of a family of measures defined by (\ref{equation2}) with elements $\{p_n(\cdot|\cdot,\cdot)\in{\cal Q}({\cal X}_n|{\cal X}_{0,n-1}\times{\cal Y}_{0,n-1}):~n\in\mathbb{N}_0\}$, we introduce the following condition on the alphabet spaces.
\vspace*{0.2cm}\\
\noi{\bf iii)} $\{{\cal X}_n:n\in\mathbb{N}_0\}$ are complete separable metric spaces and $\{{\cal B}({\cal X}_n):n\in\mathbb{N}_0\}$ are the $\sigma-$algebras of Borel sets.
\vspace*{0.2cm}\\
By \cite{gihman-skorohod1979}, if condition {\bf iii)} holds, then for any family of measures ${\bf P}(\cdot|{\bf y})$ parametrized by ${\bf y}\in{\cal Y}^{\mathbb{N}_0}$ satisfying {\bf C1} one can construct a collection of versions of conditional distributions  $\{p_n(dx_n|x^{n-1},y^{n-1}):n\in\mathbb{N}_0\}$ satisfying conditions {\bf i)} and {\bf ii)} which are connected with ${\bf P}(\cdot|{\bf y})$ via relation (\ref{equation2}), and hence the following conclusion.\\
\noi When $\{{\cal X}_n:n\in\mathbb{N}_0\}$ are Polish Spaces with $\{{\cal B}({\cal X}_n):n\in\mathbb{N}_0\}$ the $\sigma-$algebra of Borel sets, there are two equivalent definitions of a feedback channel. The first definition is the usual one given by a collection of functions $\{p_n(\cdot|\cdot,\cdot)\in{\cal Q}({\cal X}_n|{\cal X}_{0,n-1}\times{\cal Y}_{0,n-1}):~n\in\mathbb{N}_0\}$, i.e., satisfying conditions {\bf i)} and {\bf ii)}. The second definition is given by a family of measures ${\bf P}(\cdot|{\bf y})$ on $({\cal X}^{\mathbb{N}_0},{\cal B}({\cal X}^{\mathbb{N}_0}))$ depending parametrically on ${\bf y}\in{\cal Y}^{\mathbb{N}_0}$ and satisfying the consistency condition {\bf C1}.\\ %Although, the family of measures ${\bf P}(\cdot|{\bf y})$ on $({\cal X}^{\mathbb{N}_0},{\cal B}({\cal X}^{\mathbb{N}_0}))$ are finite additive probability measures, by Kolmogorov's extension theorem \cite{ash1999}, the completeness of $\{{\cal X}_n:n\in\mathbb{N}_0\}$ guarantees the existence of countable additive probability measures ${\bf P}(\cdot|{\bf y})$ on $({\cal X}^{\mathbb{N}_0},{\cal B}({\cal X}^{\mathbb{N}_0}))$, whose marginal on each ${\cal X}_{0,n}$ is ${\overleftarrow{P}}_{0,n}(\cdot|y^{n-1})$.\\
The second equivalent definition of a feedback channel, together with an analogous equivalent definition for the forward channel will be used throughout the paper.
\vspace*{0.2cm}\\
\noi{\bf Feedforward Channel.}\\
The above methodology is repeated to obtain two equivalent definitions for the forward channel as well.
Suppose for each $n\in\mathbb{N}_0,$ the conditional distribution of the RV $Y_n\in{\cal Y}_n$ is determined provided the values of the basic processes $Y^{n-1}\in{\cal Y}_{0,n-1}$ and $X^n=x^n\in{\cal X}_{0,n}$ are known, and let $\{q_n(dy_n|y^{n-1},x^n):n\in\mathbb{N}_0\}$ denotes this collection of distributions. At $n=0$, $q_0(dy_0|y_{-1},x_0)$, where $y_{-1}$ is either fixed or its distribution is fixed (depending on the convection used). Without loss of generality, we assume $q_0(dy_0|y_{-1},x_0)\tri{q}_0(dy_0|x_0)$. The functions $\{q_n(\cdot|\cdot,\cdot):~n\in\mathbb{N}_0\}$ satisfy the following conditions.
\vspace*{0.2cm}\\
{\bf iv)} For every $n\in\mathbb{N}_0$, and $y^{n-1}\in{\cal Y}_{0,n-1}, x^n\in{\cal X}_{0,n}$, $q_n(\cdot|y^{n-1},x^{n})$ is a probability measure ${\cal B}({\cal Y}_n)$;\\
{\bf v)} For every $n\in\mathbb{N}_0$, and $B_n\in{\cal B}({\cal Y}_n)$, $q_n(B_n|\cdot,\cdot)$ is an $\otimes^{n-1}_{i=0}\big({\cal B}({\cal Y}_i)\otimes{\cal B}({\cal X}_i)\big)\otimes{\cal B}({\cal X}_n)$-measurable function of $x^{n}\in{\cal X}_{0,n},$ $y^{n-1}\in{\cal Y}_{0,n-1}$.
\vspace*{0.2cm}\\ 
For every $n\in\mathbb{N}_0$, the set of all functions that satisfy {\bf iv)}, {\bf v)}, are called  stochastic kernels on ${\cal Y}_n$ given ${\cal Y}_{0,n-1}\times{\cal X}_{0,n}$, and these are denoted by 
\begin{align*}
{\cal Q}({\cal Y}_n|{\cal Y}_{0,n-1}\times{\cal X}_{0,n})=\{q_n(\cdot|y^{n-1},x^{n})\in{\cal M}_1({\cal Y}_n):~y^{n-1}\in{\cal Y}_{0,n-1}, x^{n}\in{\cal X}_{0,n}~\mbox{and}~{\bf v)}~\mbox{holds}\}.
\end{align*}
\noi Similarly as before, using the collection of functions $\{q_n(\cdot|\cdot,\cdot)\in{\cal Q}({\cal Y}_n|{\cal Y}_{0,n-1}\times{\cal X}_{0,n}):n\in\mathbb{N}_0\}$ one can construct a family of measures ${\bf Q}(\cdot|{\bf x})$ on $({\cal Y}^{\mathbb{N}_0},{\cal B}({\cal Y}^{\mathbb{N}_0}))$ which depend parametrically on ${\bf x}\in{{\cal X}^{\mathbb{N}_0}},$ as follows.\\
Consider a cylinder set $D\in{\cal B}({\cal Y}_{0,n})$ of the form
\begin{align}
D\tri\Big\{{\bf y}\in{\cal Y}^{\mathbb{N}_0}:y_0{\in}D_0,y_1{\in}D_1,\ldots,y_n{\in}D_n\Big\},~D_i\in{\cal B}({\cal Y}_i),~{n}\in\mathbb{N}_0^n, D_{0,n}=\times_{i=0}^n{D_i}.\nonumber
\end{align}
Define a family of measures on ${\cal B}({\cal Y}^{\mathbb{N}_0})$ parametrized by ${\bf x}\in{\cal X}^{\mathbb{N}_0}$ by
\begin{align}
{\bf Q}(D|{\bf x})&\tri\int_{D_0}q_0(dy_0|x_0)\int_{D_1}q_1(dy_1|y_0,x^1)\ldots\int_{D_n}q_n(dy_n|y^{n-1},x^n)\label{equation4}\\
&\equiv{\overrightarrow{Q}}_{0,n}(D_{0,n}|x^n).\label{equation4b}
\end{align}
Since, for each ${\bf x}\in{\cal X}^{\mathbb{N}_0}$ the RHS of (\ref{equation4}) defines a consistent family of finite dimensional distribution, then there exist a unique measure on $({\cal Y}^{\mathbb{N}_0},{\cal B}({\cal Y}^{\mathbb{N}_0}))$ from which the family of distributions $\{q_n(dy_n|y^{n-1},x^n):n\in\mathbb{N}_0\}$ satisfying {\bf iv)}, {\bf v)} can be obtained. Moreover, the family of measures ${\bf Q}(\cdot|{\bf x})$ parametrized by ${\bf x}\in{\cal X}^{\mathbb{N}_0}$ satisfies the following consistency condition.
\vspace*{0.2cm}\\
\noi{\bf C2}: If $F\in{\cal B}({\cal Y}_{0,n}),$ then ${\bf Q}(F|\cdot)$ is a ${\cal B}({\cal X}_{0,n})-$measurable function of ${\bf x}\in{\cal X}^{\mathbb{N}_0}.$
\vspace*{0.2cm}\\
By \cite{gihman-skorohod1979}, to obtain another equivalent definition for the forward channel introduce the following condition on the output alphabet.
\vspace*{0.2cm}\\
\noi {\bf vi)} $\{{\cal Y}_n:n\in\mathbb{N}_0\}$ are Polish Spaces and $\{{\cal B}({\cal Y}_n):n\in\mathbb{N}_0\}$ are the $\sigma-$algebra of Borel sets.
\vspace*{0.2cm}\\
If condition {\bf vi}) holds, then for any family of measures ${\bf Q}(\cdot|{\bf x})$ on $({\cal Y}^{\mathbb{N}_0},{\cal B}({\cal Y}^{\mathbb{N}_0}))$ parametrized by ${\bf x}\in{\cal X}^{\mathbb{N}_0}$ satisfying consistency condition {\bf C2}, one can construct a collection of functions $\{q_n(\cdot|\cdot,\cdot)\in{\cal Q}({\cal Y}_n|{\cal Y}_{0,n-1}\times{\cal X}_{0,n}):n\in\mathbb{N}_0\}$, i.e., satisfying conditions {\bf iv)} and {\bf v)}, which are connected with ${\bf Q}(\cdot|{\bf x})$ via relation (\ref{equation4}). Therefore, we arrive at two equivalent definitions for the forward channel as well.\\
We conclude this section by constructing the probability space $(\Omega,{\cal F},\mathbb{P})$, as stated earlier, and the sequence of RV's $\{(X_n,Y_n):~n\in\mathbb{N}_0\}$ defined on it.
Given the basic measures ${\bf P}(\cdot|{\bf y})$ on ${\cal X}^{\mathbb{N}_0}$ satisfying consistency condition {\bf C1} and ${\bf Q}(\cdot|{\bf x})$ on ${\cal Y}^{\mathbb{N}_0}$ satisfying consistency condition {\bf C2}, one can construct a sequence of RV's $\{X_n,Y_n:n\in\mathbb{N}_0\}$ or conditional distributions as follows.\\
Suppose {\bf iii)}, {\bf iv)} hold. Let $A^{(n)}=\{{\bf x}:x_n{\in}A\},$ $A\in{\cal B}({\cal X}_n)$ and $B^{(n)}=\{{\bf y}:y_n{\in}B\},$ $B\in{\cal B}({\cal Y}_n).$ In addition, let ${\bf P}(A^{(n)}|{{\cal B}({\cal X}_{0,n-1})}|{\bf y})$ denote the conditional probability of $A^{(n)}$ with respect to ${\cal B}({\cal X}_{0,n-1})$ calculated on the probability space $\big({\cal X}^{\mathbb{N}_0},{\cal B}({\cal X}^{\mathbb{N}_0}),{\bf P}(\cdot|{\bf y})\big),$ and ${\bf Q}(B^{(n)}|{{\cal B}({\cal Y}_{0,n-1})}|{\bf x})$ denote the conditional probability of $B^{(n)}$ with respect to ${\cal B}({\cal Y}_{0,n-1})$ calculated on the probability space $\big({\cal Y}^{\mathbb{N}_0},{\cal B}({\cal Y}^{\mathbb{N}_0}),{\bf Q}(\cdot|{\bf x})\big).$\\
Then for each $n\in\mathbb{N}_0$, by conditioning it follows that
\begin{align}
\mathbb{P}\big\{X_n{\in}A|X^{n-1}=x^{n-1},Y^{n-1}=y^{n-1}\big\}&={\bf P}\big(\{{\bf x}:x_n{\in}A\}|{{\cal B}({\cal X}_{0,n-1})}|{\bf y}\big),~A{\in}{\cal B}({\cal X}_n)\nonumber\\
&=p_n(A|x^{n-1},y^{n-1})\label{equation16}\\
\mathbb{P}\big\{Y_n{\in}B|Y^{n-1}=y^{n-1},X^n=x^n\big\}&={\bf Q}\big(\{{\bf y}:y_n{\in}B\}|{{\cal B}({\cal Y}_{0,n-1})}|{\bf x}\big),~B{\in}{\cal B}({\cal Y}_n)\nonumber\\
&=q_n(B|y^{n-1},x^n)\label{equation17}
\end{align}
for almost all ${\bf x}\in{\cal X}^{\mathbb{N}_0}$ in measure ${\bf P}(\cdot|{\bf y}),$ and for almost all ${\bf y}\in{\cal Y}^{\mathbb{N}_0}$ in measure ${\bf Q}(\cdot|{\bf x})$. Note that for each $n\in\mathbb{N}_0$, $p_n(\cdot;\cdot,\cdot)\in{\cal Q}({\cal X}_n|{\cal X}_{0,n-1},{\cal Y}_{0,n-1})$ and $q_n(\cdot|\cdot,\cdot)\in{\cal Q}({\cal Y}_n|{\cal Y}_{0,n-1},{\cal X}_n)$ are stochastic kernels determined from ${\bf P}(\cdot|\cdot)$ and ${\bf Q}(\cdot|\cdot),$ respectively, (e.g., they are related via (\ref{equation2}) and (\ref{equation4}), respectively).\\
Consequently, the finite dimensional distributions of the sequence of RV's $\{(X_n,Y_n):~n\in\mathbb{N}_0\}$ is defined by 
\begin{align}
\mathbb{P}\big\{X_0{\in}A_0,Y_0\in{B}_0,\ldots,X_n{\in}A_n,Y_n{\in}B_n\big\}&=\int_{A_0}p_0(dx_0)\int_{B_0}q_0(dy_0|x_0)\ldots\nonumber\\
&\int_{A_n}p_n(dx_n|x^{n-1},y^{n-1})\int_{B_n}q_n(dy_n|y^{n-1},x^n).\label{equation18a}
\end{align}
Hence, given the two Polish spaces ${\cal X}^{\mathbb{N}_0}$ and ${\cal Y}^{\mathbb{N}_0}$, for any ${\bf P}(\cdot|\cdot)$ and ${\bf Q}(\cdot|\cdot)$ satisfying the consistency conditions {\bf C1}, {\bf C2}, respectively, there exist a probability space and a sequence of RV's $\{(X_n,Y_n):~n\in\mathbb{N}_0\}$ defined on it, whose joint probability distribution is uniquely defined by (\ref{equation18a}), via ${\bf P}(\cdot|\cdot)$ and ${\bf Q}(\cdot|\cdot)$.
\vspace*{0.2cm}\\
The following remark summarizes the previous discussion on the two equivalent definitions of forward and feedback channels.
\begin{remark}\label{equivalent}{\ \\}
Suppose $\{{\cal X}_n:n\in\mathbb{N}_0\},$ $\{{\cal Y}_n:n\in\mathbb{N}_0\},$ are complete, separable metric spaces (Polish spaces) and $\{{\cal B}({\cal X}_n):n\in\mathbb{N}_0\},$ $\{{\cal B}({\cal Y}_n):n\in\mathbb{N}_0\}$ are respectively, the $\sigma-$algebras of Borel sets.\\
Then
\item[1)] The collection of stochastic kernels $\{p_n(\cdot|\cdot,\cdot)\in{\cal Q}({\cal X}_n|{\cal X}_{0,n-1}\times{\cal Y}_{0,n-1}):n\in\mathbb{N}_0\}$ uniquely define a family of probability measures on $({\cal X}^{\mathbb{N}_0},{\cal B}({\cal X}^{\mathbb{N}_0}))$ parametrized by ${\bf y}\in{\cal Y}^{\mathbb{N}_0}$ via (\ref{equation2}).
\item[2)] For any family of probability measures ${\bf P}(\cdot|{\bf y})$ on $({\cal X}^{\mathbb{N}_0},{\cal B}({\cal X}^{\mathbb{N}_0}))$ parametrized by ${\bf y}\in{\cal Y}^{\mathbb{N}_0}$, satisfying consistency condition {\bf C1} there exists a collection of stochastic kernels $\{p_n(\cdot|\cdot,\cdot)\in{\cal Q}({\cal X}_n|{\cal X}_{0,n-1}\times{\cal Y}_{0,n-1}):n\in\mathbb{N}_0\}$ connected to ${\bf P}(\cdot|\cdot)$ via (\ref{equation2}).
\item[3)] The collection of stochastic kernels $\{q_n(\cdot|\cdot,\cdot)\in{\cal Q}({\cal Y}_n|{\cal Y}_{0,n-1}\times{\cal X}_{0,n}):n\in\mathbb{N}_0\}$ uniquely define a family of probability measures on $({\cal Y}^{\mathbb{N}_0},{\cal B}({\cal Y}^{\mathbb{N}_0}))$ parametrized by ${\bf x}\in{\cal X}^{\mathbb{N}_0}$ via (\ref{equation4}).
\item[4)] For any family of probability measures ${\bf Q}(\cdot|{\bf x})$ on $({\cal Y}^{\mathbb{N}_0},{\cal B}({\cal Y}^{\mathbb{N}_0}))$ parametrized by ${\bf x}\in{\cal X}^{\mathbb{N}_0}$ satisfying consistency condition {\bf C2} there exists a collection of stochastic kernels $\{q_n(\cdot|\cdot,\cdot)\in{\cal Q}({\cal Y}_n|{\cal Y}_{0,n-1}\times{\cal X}_{0,n}):n\in\mathbb{N}_0\}$ connected to ${\bf Q}(\cdot|\cdot)$ via (\ref{equation4}).
\end{remark}

\noi The point to be made here is that directed information as defined by (\ref{equation1a})-(\ref{equation1c}) can be expressed via the equivalent definitions of Remark~\ref{equivalent}, 2) and 4) rather than 1) and 3). We use this equivalent definition of directed information, to derive the functional and topological properties of directed information on general abstract spaces.
Throughout the rest of the paper it is assumed that the conditions of Remark~\ref{equivalent} are satisfied, i.e., all spaces are Polish spaces.%, which is assumed throughout the rest of the paper. %However, if the spaces $\{({\cal X}_n,{\cal Y}_n):n\in\mathbb{N}_0\}$ are not complete, countably additivity of the family of probability measures ${\bf P}(\cdot|{\bf y})$ and ${\bf Q}(\cdot|{\bf x})$ does not fail. This is because these spaces can be homeomorphically embedded as Borel subsets in complete separable metric spaces $\{(\tilde{\cal X}_n,\tilde{\cal Y}_n):n\in\mathbb{N}_0\}$, so that Kolmogorov's extension theorem can be utilized \cite{bertsekas-shreve2007}.

%%%%%%%%%%%%%%%%%%%%%%%%%%%%%%%%%%%%%%%%%%%%%%%%%%%%%%%%%%%%%%%%%%%%%%%%%%%%%%%%%%%%%%%%%%%%

\section{Properties of Directed Information}\label{properties}

\par In this section, we define the feedforward information $I(X^n\rightarrow{Y^n})$ on abstract spaces (Polish spaces), via the Kullback-Leibler distance (or relative entropy), using the basic family of measures ${\bf P}(\cdot|{\bf y})$ on $({\cal X}^{\mathbb{N}_0},{\cal B}({\cal X}^{\mathbb{N}_0})),$ and ${\bf Q}(\cdot|{\bf x})$ on $({\cal Y}^{\mathbb{N}_0},{\cal B}({\cal Y}^{\mathbb{N}_0})),$ which satisfy consistency condition ${\bf C1}$ and ${\bf C2},$ respectively. Once this is established, then following Pinsker \cite{pinsker-book}, it will become obvious that directed information permits a representation as a supremum of relative entropy between two distributions, where the supremum is taken over all measurable partitions on a given $\sigma-$ algebra of subsets of a set $\cal Z.$ Further, in a subsequent subsection, we use the definition of directed information in terms of ${\bf P}(\cdot|{\bf y})$ and ${\bf Q}(\cdot|{\bf x}),$ to derive several of its properties, such as, convexity, concavity, lower semicontinuity, with respect to these two families of measures.\\
To present the precise expression for the directed information, we first introduce the measures of interest constructed from the basic consistent families of conditional distributions. Introduce the following notation.\\
The set of stochastic kernels by
\begin{align}
{\cal Q}^{\bf C1}({\cal X}^{\mathbb{N}_0}|{\cal Y}^{\mathbb{N}_0})&\tri\Big\{
{\bf P}(\cdot|{\bf y})\in{\cal M}_1({\cal X}^{\mathbb{N}_0}): {\bf y}\in{\cal Y}^{\mathbb{N}_0}~\mbox{and consistency condition}~{\bf C1}~\mbox{holds}\Big\}\nonumber\\
&\equiv\Big\{{\bf P}(\cdot|\cdot)\in{\cal Q}({\cal X}^{\mathbb{N}_0}|{\cal Y}^{\mathbb{N}_0}):~\mbox{consistency condition}~{\bf C1}~\mbox{holds}\Big\}.\label{section:properties:equation2}
\end{align}
Note that for each  ${\bf y}\in{\cal Y}^{\mathbb{N}_0}$, elements of this set are  probability distributions on ${\cal X}^{\mathbb{N}_0}$ denoted  by
\begin{align}
{\cal M}^{\bf C1}_1({\cal X}^{\mathbb{N}_0})&\tri\Big\{{\bf P}(\cdot|{\bf y})\in{\cal M}_1({\cal X}^{\mathbb{N}_0}):~\mbox{consistency condition ${\bf C1}$ holds}\Big\}\label{section:properties:equation1}
\end{align}
Similarly, 
\begin{align}
{\cal Q}^{\bf C2}({\cal Y}^{\mathbb{N}_0}|{\cal X}^{\mathbb{N}_0})&\tri\Big\{{\bf Q}(\cdot|{\bf x})\in{\cal M}_1({\cal Y}^{\mathbb{N}_0}): {\bf x}\in{\cal X}^{\mathbb{N}_0}~\mbox{and consistency condition}~{\bf C2}~\mbox{holds}\Big\}\nonumber\\
&\equiv\Big\{{\bf Q}(\cdot|\cdot)\in{\cal Q}({\cal Y}^{\mathbb{N}_0}|{\cal X}^{\mathbb{N}_0}):~\mbox{consistency condition}~{\bf C2}~\mbox{holds}\Big\}.\label{section:properties:equation4}
\end{align} 
and for each  ${\bf x}\in{\cal X}^{\mathbb{N}_0}$,  elemenets of this set are probability distributions on ${\cal Y}^{\mathbb{N}_0}$, denoted by
\begin{align}
{\cal M}^{\bf C2}_1({\cal Y}^{\mathbb{N}_0})&\tri\Big\{{\bf Q}(\cdot|{\bf x})\in{\cal M}_1({\cal Y}^{\mathbb{N}_0}):~\mbox{consistency condition ${\bf C2}$ holds}\Big\}\label{section:properties:equation3}  
\end{align}
The projection of ${\cal M}^{\bf C1}_1({\cal X}^{\mathbb{N}_0})$, ${\cal M}^{\bf C2}_1({\cal Y}^{\mathbb{N}_0})$, ${\cal Q}^{\bf C1}({\cal X}^{\mathbb{N}_0}|{\cal Y}^{\mathbb{N}_0})$, and ${\cal Q}^{\bf C2}({\cal Y}^{\mathbb{N}_0}|{\cal X}^{\mathbb{N}_0})$ to finite number of coordinates is denoted by ${\cal M}^{\bf C1}_1({\cal X}_{0,n})$, ${\cal M}^{\bf C2}_1({\cal Y}_{0,n})$, ${\cal Q}^{\bf C1}({\cal X}_{0,n}|{\cal Y}_{0,n-1})$, and ${\cal Q}^{\bf C2}({\cal Y}_{0,n}|{\cal X}_{0,n})$, respectively. Since the spaces are complete separable metric spaces then ${\bf P}(\cdot|{\bf y})\in{\cal M}_1({\cal X}^{\mathbb{N}_0})$, for fixed ${\bf y}\in{\cal Y}^{\mathbb{N}_0}$, and ${\bf Q}(\cdot|{\bf x})\in{\cal M}_1({\cal Y}^{\mathbb{N}_0})$, for fixed ${\bf x}\in{\cal X}^{\mathbb{N}_0}$, are regular conditional probability distributions \cite{dupuis-ellis97}.\\
Next, we define the distributions of interest. Given any ${\bf P}(\cdot|\cdot)\in{\cal Q}^{\bf C1}({\cal X}^{\mathbb{N}_0}|{\cal Y}^{\mathbb{N}_0})$ and ${\bf Q}(\cdot|\cdot)\in{\cal Q}^{\bf C2}({\cal Y}^{\mathbb{N}_0}|{\cal X}^{\mathbb{N}_0})$, by utilizing the construction of Section~\ref{equivalent_definitions}, we can define uniquely $\{p_n(\cdot|\cdot,\cdot):n\in\mathbb{N}_0\}$ and $\{q_n(\cdot|\cdot,\cdot):n\in\mathbb{N}_0\}$, \big(see (\ref{equation16}), (\ref{equation17})\big) and the following distributions.
\vspace*{0.2cm}\\
\noi{\bf P1}: The joint distribution on ${\cal X}^{\mathbb{N}_0}\times{\cal Y}^{\mathbb{N}_0}$ of the basic sequence $\{X_n,Y_n:n\in\mathbb{N}_0\}$ constructed from ${\bf P}(\cdot|\cdot)\in{\cal Q}^{\bf C1}({\cal X}^{\mathbb{N}_0}|{\cal Y}^{\mathbb{N}_0})$ and ${\bf Q}(\cdot|\cdot)\in{\cal Q}^{\bf C2}({\cal Y}^{\mathbb{N}_0}|{\cal X}^{\mathbb{N}_0}),$ defined uniquely for $A_i\in{\cal B}({\cal X}_i),$ $B_i\in{\cal B}({\cal Y}_i),$ $\forall{i}\in\mathbb{N}_0^n,$ by
\begin{align}
({\overleftarrow P}_{0,n}&\otimes{\overrightarrow Q}_{0,n})(\times^n_{i=0}(A_i{\times}B_i)){\tri}\mathbb{P}\Big\{X_0{\in}A_0,Y_0\in{B}_0,\ldots,X_n{\in}A_n,Y_n{\in}B_n\Big\}\nonumber\\
&=\int_{A_0}p_0(dx_0)\int_{B_0}q_0(dy_0|x_0)\ldots\int_{A_n}p_n(dx_n|x^{n-1},y^{n-1})\int_{B_n}q_n(dy_n|y^{n-1},x^n).\label{equation18}
\end{align}
Formally, the $(n+1)$ fold compound joint distribution defined by (\ref{equation18}) is written as $({\overleftarrow P}_{0,n}\otimes{\overrightarrow Q}_{0,n})(dx^n,dy^n)$ or ${\overleftarrow P}_{0,n}(dx^n|y^{n-1})\otimes{\overrightarrow Q}_{0,n}(dy^n|x^{n}).$
\vspace*{0.2cm}\\
\noi{\bf P2}: The marginal distributions on ${\cal X}^{\mathbb{N}_0}$ of the sequence $\{X_n:n\in\mathbb{N}_0\}$ constructed from ${\bf P}(\cdot|\cdot)\in{\cal Q}^{\bf C1}({\cal X}^{\mathbb{N}_0}|{\cal Y}^{\mathbb{N}_0})$ and ${\bf Q}(\cdot|\cdot)\in{\cal Q}^{\bf C2}({\cal Y}^{\mathbb{N}_0}|{\cal X}^{\mathbb{N}_0}),$ defined uniquely by\footnote{Actually $\mu\equiv\mu^{{\bf P}\otimes{\bf Q}}$ but we omit the superscript throughout the paper.}
\begin{align}
\mu_{0,n}(\times^n_{i=0}A_i)&\tri\mathbb{P}\Big\{X_0\in{A}_0, Y_0\in{\cal Y}_0,\ldots, X_n\in{A}_n, Y_n\in{\cal Y}_n\Big\},~A_i\in{\cal B}({\cal X}_i),~\forall{i}\in\mathbb{N}_0^n\label{equation19}\\
&=({\overleftarrow P}_{0,n}\otimes{\overrightarrow Q}_{0,n})(\times^n_{i=0}(A_i\times{\cal Y}_i))\nonumber\\
&=\int_{A_0}p_0(dx_0)\int_{{\cal Y}_0}q_0(dy_0|x_0)\ldots\int_{A_n}p_n(dx_n|x^{n-1},y^{n-1})\int_{{\cal Y}_n}q_n(dy_n|y^{n-1},x^n).\label{equation19a}
\end{align}
Formally, (\ref{equation19a}) is written as $\mu_{0,n}(dx^n)=({\overleftarrow P}_{0,n}\otimes{\overrightarrow Q}_{0,n})(dx^n,{\cal Y}_{0,n})$, and by Bayes' rule $\mu_{0,n}(dx^n)=\otimes^{n}_{i=0}\mu_{i}(dx_i|x^{i-1}).$
\vspace*{0.2cm}\\
\noi{\bf P3}: The marginal distributions on ${\cal Y}^{\mathbb{N}_0}$ of the sequence $\{Y_n:n\in\mathbb{N}_0\}$ constructed from ${\bf P}(\cdot|\cdot)\in{\cal Q}^{\bf C1}({\cal X}^{\mathbb{N}_0}|{\cal Y}^{\mathbb{N}_0})$ and ${\bf Q}(\cdot|\cdot)\in{\cal Q}^{\bf C2}({\cal Y}^{\mathbb{N}_0}|{\cal X}^{\mathbb{N}_0}),$ defined uniquely by\footnote{Similarly, $\nu\equiv\nu^{{\bf P}\otimes{\bf Q}}$.}
\begin{align}
\nu_{0,n}(\times^n_{i=0}B_i)&\tri\mathbb{P}\Big\{X_0\in{\cal X}_0, Y_0\in{B}_0,\ldots, X_n\in{\cal X}_n, Y_n\in{B}_n\Big\},~B_i\in{\cal B}({\cal Y}_i),~\forall{i}\in\mathbb{N}_0^n\label{equation20}\\
&=({\overleftarrow P}_{0,n}\otimes{\overrightarrow Q}_{0,n})(\times^n_{i=0}({\cal X}_i\times{B}_i))\nonumber\\
&=\int_{{\cal X}_0}p_0(dx_0)\int_{B_0}q_0(dy_0|x_0)\ldots\int_{{\cal X}_n}p_n(dx_n|x^{n-1},y^{n-1})\int_{B_n}q_n(dy_n|y^{n-1},x^n).\label{equation20a}
\end{align}
Formally, (\ref{equation20a}) is written as $\nu_{0,n}(dy^n)=({\overleftarrow P}_{0,n}\otimes{\overrightarrow Q}_{0,n})({\cal X}_{0,n},dy^n)$, and by Bayes' rule $\nu_{0,n}(dy^n)=\otimes^{n}_{i=0}\nu_{i}(dy_i|y^{i-1})$.
\vspace*{0.2cm}\\
\noi {\bf P4}: The distribution ${\overrightarrow\Pi}_{0,n}:{\cal B}({\cal X}_{0,n})\otimes{\cal B}({\cal Y}_{0,n})\mapsto[0,1]$ constructed from ${\overleftarrow P}_{0,n}(\cdot|\cdot)\in{\cal Q}^{\bf C1}({\cal X}_{0,n}|{\cal Y}_{0,n-1})$ and $\nu_{0,n}(dy^n)=({\overleftarrow P}_{0,n}\otimes{\overrightarrow Q}_{0,n})({\cal X}_{0,n},dy^n)\in{\cal M}_1({\cal Y}_{0,n})$ of (\ref{equation20}), defined uniquely by
\begin{align}
{\overrightarrow\Pi}_{0,n}(\times^n_{i=0}(A_i{\times}B_i))&\tri({\overleftarrow P}_{0,n}\otimes\nu_{0,n})(\times^n_{i=0}(A_i{\times}B_i)),~A_i\in{\cal B}({\cal X}_i),~B_i\in{\cal B}({\cal Y}_i),~\forall{i}\in\mathbb{N}_0^n\nonumber\\
&=\int_{A_0}p_0(dx_0)\int_{B_0}\nu_0(dy_0)\int_{A_1}p_1(dx_1|x_0,y_0)\int_{B_1}\nu_{1}(dy_1|y_0)\ldots\nonumber\\
&\ldots\int_{A_n}p_n(dx_n|x^{n-1},y^{n-1})\int_{B_n}\nu_{n}(dy_n|y^{n-1}). \label{equation21}
\end{align}
Formally, (\ref{equation21}) is written as ${\overrightarrow\Pi}_{0,n}(dx^n,dy^n)={\overleftarrow P}_{0,n}(dx^n|y^{n-1})\otimes\nu_{0,n}(dy^n)\in{\cal M}_1({\cal X}_{0,n}\times{\cal Y}_{0,n})$.
\vspace*{0.2cm}\\
\noi{\bf P5}: The distribution ${\overleftarrow\Pi}_{0,n}:{\cal B}({\cal Y}_{0,n})\otimes{\cal B}({\cal X}_{0,n})\mapsto[0,1]$ constructed from ${\overrightarrow Q}_{0,n}(\cdot|\cdot)\in{\cal Q}^{\bf C2}({\cal Y}_{0,n}|{\cal X}_{0,n})$ and $\mu_{0,n}(dx^n)=({\overleftarrow P}_{0,n}\otimes{\overrightarrow Q}_{0,n})(dx^n,{\cal Y}_{0,n})\in{\cal M}_1({\cal X}_{0,n})$ of (\ref{equation19a}), defined uniquely by
\begin{align}
{\overleftarrow\Pi}_{0,n}(\times^n_{i=0}(A_i{\times}B_i))&\tri(\mu_{0,n}\otimes{\overrightarrow Q}_{0,n})(\times^n_{i=0}(A_i{\times}B_i)),~A_i\in{\cal B}({\cal X}_i),~B_i\in{\cal B}({\cal Y}_i),~\forall{i}\in\mathbb{N}_0^n\nonumber\\
&=\int_{A_0}\mu_0(dx_0)\int_{B_0}q_0(dy_0|x_0)\int_{A_1}\mu_{1}(dx_1|x_0)\int_{B_1}q_1(dy_1|y_0,x_0)\ldots\nonumber\\
&\ldots\int_{A_n}\mu_{n}(dx_n|x^{n-1})\int_{B_n}q_n(dy_n|y^{n-1},x^{n}).\label{equation22}
\end{align}
Formally, (\ref{equation22}) is written as ${\overleftarrow\Pi}_{0,n}(dx^n,dy^n)=\mu_{0,n}(dx^n)\otimes{\overrightarrow Q}_{0,n}(dy^n|x^n)\in{\cal M}_1({\cal X}_{0,n}\times{\cal Y}_{0,n})$.
\vspace*{0.2cm}\\
\noi From the above definitions, for each $n\in\mathbb{N}_0$, an alternative way to construct the conditional distributions of $Y_n$ given $Y^{n-1}=y^{n-1}$, $\nu_{n}(\cdot|y^{n-1})\in{\cal M}_1({\cal Y}_n)$, and $X_n$ given $X^{n-1}=x^{n-1}$, $\mu_{n}(\cdot|x^{n-1})\in{\cal M}_1({\cal X}_n)$ is as follows. Let $A^{(n)}=\{{\bf x}:x_n\in{A}\},$ $A\in{\cal B}({\cal X}_n),$ $B^{(n)}=\{{\bf y}:y_n\in{B}\},$ $B\in{\cal B}({\cal Y}_n),$ and let ${\overrightarrow\Pi}_{0,n}(A^{(n)},B^{(n)}|{{\cal B}({\cal X}_{0,n-1})\otimes{\cal B}({\cal Y}_{0,n-1})})$ denote the joint conditional probability of $A^{(n)}\times{B}^{(n)}$ with respect to ${\cal B}({\cal X}_{0,n-1})\otimes{\cal B}({\cal Y}_{0,n-1})$ calculated on the probability space $\Big({\cal X}^{\mathbb{N}_0}\otimes{\cal Y}^{\mathbb{N}_0},{\cal B}({\cal X}^{\mathbb{N}_0})\otimes{\cal B}({\cal Y}^{\mathbb{N}_0}),{\overrightarrow\Pi}_{0,n}(\cdot)\Big)$. Then, for $ A\in{\cal B}({\cal X}_n)$,~$B\in{\cal B}({\cal Y}_n)$ we obtain
\begin{align}
{\overrightarrow\Pi}_{0,n}(A^{(n)},B^{(n)}|{{{\cal B}({\cal X}_{0,n-1})}\otimes{{\cal B}({\cal Y}_{0,n-1})}})&=p_n(A|x^{n-1},y^{n-1})\times\nu_{n}(B|y^{n-1}).\label{equation5a}
\end{align}
Hence, $\nu_{n}(\cdot|y^{n-1})\in{\cal M}_1({\cal Y}_n)$ is given by $\nu_{n}(dy_n|y^{n-1})=\int_{{\cal X}_n}{\overrightarrow\Pi}_{0,n}(dx_n,dy_n|x^{n-1},y^{n-1})$, from which $\nu_{0,n}(dy^n)\in{\cal M}_1({\cal Y}_{0,n})$ is also obtained.
Similarly, let ${\overleftarrow\Pi}_{0,n}(A^{(n)},B^{(n)}|{{\cal B}({\cal Y}_{0,n-1})}\otimes{{\cal B}({\cal X}_{0,n-1})})$ denote the joint conditional probability of $A^{(n)}\times{B}^{(n)}$ with respect to ${\cal B}({\cal Y}_{0,n-1})\otimes{\cal B}({\cal X}_{0,n-1})$ calculated on the probability space $\big({\cal Y}^{\mathbb{N}_0}\times{\cal X}^{\mathbb{N}_0},{\cal B}({\cal Y}^{\mathbb{N}_0})\otimes{\cal B}({\cal X}^{\mathbb{N}_0}),{\overleftarrow\Pi}_{0,n}(\cdot)\big).$ Then, for $B\in{\cal B}({\cal Y}_n)$ we have
\begin{align}
{\overleftarrow\Pi}_{0,n}(A^{(n)},B^{(n)}|{{{\cal B}({\cal X}_{0,n-1})}\otimes{{\cal B}({\cal Y}_{0,n-1})}})&=\int_{A_n}q_n(B|y^{n-1},x^n)\otimes\mu_{n}(dx_n|x^{n-1})\label{equation5b}
\end{align}
from which $\mu_{n}(\cdot|x^{n-1})\in{\cal M}_1({\cal X}_n)$ and $\mu_{0,n}(dx^n)\in{\cal M}_1({\cal X}_{0,n})$ are obtained.
Similarly, from (\ref{equation21}) and (\ref{equation5a}) we can obtain any of the individual conditional distributions $p_n(\cdot|x^{n-1},y^{n-1})\in{\cal M}_1({\cal X}_n)$ and $q_n(\cdot|y^{n-1},x^n)\in{\cal M}_1({\cal Y}_n)$ appearing in their RHS by proper conditional expectations.
\vspace*{0.2cm}\\
%%%%%%%%%%%%%%%%%%%%%%%%%%%%%%%%%%%%%%%%%%%%%%%%%%%%%%%%%%%%%%%%%%%%%%%%%%%%%%%%%%%%%%%
\noi Using the first definition of basic processes, that is, given a collection of stochastic kernels $\{p_n(\cdot|\cdot,\cdot)\in{\cal Q}({\cal X}_n|{\cal X}_{0,n-1}\times{\cal Y}_{0,n-1}):n\in\mathbb{N}_0\}$ and $\{q_n(\cdot|\cdot,\cdot)\in{\cal Q}({\cal Y}_n|{\cal Y}_{0,n-1}\times{\cal X}_{0,n}):n\in\mathbb{N}_0\}$, the joint distribution, as well as the conditional distributions are defined via ${\bf P1-P5}.$ Consequently, it is well-known that directed information is defined via relative entropy as follows \cite{tatikonda-mitter2009} 
\begin{align}
&I(X^n\rightarrow{Y}^n)\tri\sum_{i=0}^{n}I(X^i;Y_i|Y^{i-1})\nonumber\\
&=\sum_{i=0}^n\int_{{\cal Y}_{0,i-1}}\int_{{\cal X}_{0,i}\times{\cal Y}_{i}}\log\Bigg(\frac{dP_{0,i}(\cdot,\cdot|y^{i-1})}{d\big(P_{0,i}(\cdot|y^{i-1})\times\nu_{i}(\cdot|y^{i-1})\big)}(x^i,y_i)\Bigg)P_{0,i}(dx^i,dy_i|y^{i-1})P_{0,i-1}(dy^{i-1})\label{equation29}\\
&=\sum_{i=0}^n\int_{{\cal X}_{0,i}\times{\cal Y}_{0,i-1}}\mathbb{D}\big(q_i(\cdot|y^{i-1},x^i)||\nu_{i}(\cdot|y^{i-1})\big)p_i(dx_i|x^{i-1},y^{i-1})\nonumber\\
&\qquad\qquad\otimes_{j=0}^{i-1}\Big(q_j(dy_j|y^{j-1},x^j)\otimes{p}_j(dx_j|x^{j-1},y^{j-1})\Big)\label{equation29a}\\
&\equiv{{\mathbb{I}}_{X^n\rightarrow{Y}^n}}(p_i(\cdot|\cdot,\cdot),q_i(\cdot|\cdot,\cdot):~i=0,1,\ldots,n).\label{equation30}
\end{align}
The RHS in (\ref{equation29}) follows from the definition of conditional mutual information. In (\ref{equation30}), we use the notation ${{\mathbb{I}}_{X^n\rightarrow{Y}^n}}(p_i(\cdot|\cdot,\cdot),q_i(\cdot|\cdot,\cdot):~i=0,1,\ldots,n)$ to indicate that $I(X^n\rightarrow{Y}^n)$ is a functional of $\{p_i(\cdot|\cdot,\cdot)\in{\cal Q}({\cal X}_i|{\cal X}_{0,i-1}\times{\cal Y}_{0,i-1}),~q_i(\cdot|\cdot,\cdot)\in{\cal Q}({\cal Y}_i|{\cal Y}_{0,i-1}\times{\cal X}_{0,i}):~i=0,1,\ldots,n\}$. %Additional comments regarding (\ref{equation29})-(\ref{equation30}) are given in Appendix~\ref{additional_discussion} (pp. 48-50). 

%%%%%%%%%%%%%%%%%%%%%%%%%%%%%%%%%%%%%%%%%%%%%%%%%%%%%%%%%%%%%%%%%%%%%%%%%%%%%%%%%%%%%%%%%%%%

\subsection{Directed Information Functional of Consistent Conditional Distributions}

Now we consider the second definition of basic process introduced in Section~\ref{equivalent_definitions}. Given any ${\bf P}(\cdot|\cdot)\in{\cal Q}^{\bf C1}({\cal X}^{\mathbb{N}_0}|{\cal Y}^{\mathbb{N}_0})$ and ${\bf Q}(\cdot|\cdot)\in{\cal Q}^{\bf C2}({\cal X}^{\mathbb{N}_0}|{\cal Y}^{\mathbb{N}_0})$ the distributions under ${\bf P1-P5}$ are constructed. Next, we define directed information via relative entropy as often done for mutual information \cite{csiszar92}. By Lemma~\ref{absolute_contunuity}, ${\overleftarrow P}_{0,n} \otimes {\overrightarrow Q}_{0,n} << {\overleftarrow P}_{0,n} \otimes \nu_{0,n}$ if and only if ${\overrightarrow Q}_{0,n}(\cdot|x^n) << \nu_{0,n}(\cdot)$ for $\overleftarrow{P}_{0,n}-$almost all $x^n \in {\cal X}_{0,n}$. Utilizing the Radon-Nikodym derivative (RND) $\frac{d({\overleftarrow P}_{0,n}\otimes{\overrightarrow Q}_{0,n})}{d({\overleftarrow P}_{0,n}\otimes\nu_{0,n})}(x^n,y^n)$, define the relative entropy of ${\overleftarrow P}_{0,n} \otimes {\overrightarrow Q}_{0,n}$ with respect to ${\overrightarrow\Pi}_{0,n}$ as follows.
\begin{align}
{\mathbb{I}}_{X^n\rightarrow{Y}^n}({\overleftarrow P}_{0,n}, {\overrightarrow Q}_{0,n})&\tri\mathbb{D}({\overleftarrow P}_{0,n} \otimes {\overrightarrow Q}_{0,n}|| \overrightarrow\Pi_{0,n})\nonumber\\
&=\int_{{\cal X}_{0,n} \times {\cal Y }_{0,n}}\log \Big( \frac{d({\overleftarrow P}_{0,n}\otimes {\overrightarrow Q}_{0,n})}{d ( {\overleftarrow P}_{0,n}\otimes \nu_{0,n} ) }(x^n,y^n)\Big) ({\overleftarrow P}_{0,n}\otimes {\overrightarrow Q}_{0,n})(dx^n,dy^n)\label{equation202}\\
&= \int_{{\cal X}_{0,n} \times {\cal Y}_{0,n}} \log \Big( \frac{d{\overrightarrow Q}_{0,n}(\cdot|x^n)}{d\nu_{0,n}(\cdot)}(y^n)\Big)({\overleftarrow P}_{0,n}\otimes {\overrightarrow Q}_{0,n})(dx^n,dy^n)\label{equation6}\\
&\equiv\mathbb{I}_{X^n\rightarrow{Y}^n}({\overleftarrow P}_{0,n},{\overrightarrow Q}_{0,n})\label{section:properties:equations1}
\end{align}
Note that (\ref{equation6}) is obtained by utilizing the fact that if ${\overleftarrow P}_{0,n} \otimes {\overrightarrow Q}_{0,n}  << {\overleftarrow P}_{0,n} \otimes \nu_{0,n}$ then the RND $ \frac {d({\overleftarrow P}_{0,n} \otimes {\overrightarrow Q}_{0,n})}{d({\overleftarrow P}_{0,n}\otimes \nu_{0,n} )}(x^n,y^n)$ represents a version of $\frac {d{\overrightarrow Q}_{0,n}(\cdot|x^n) }{d\nu_{0,n} (\cdot)}(y^n)$, $\overleftarrow{P}_{0,n}-a.s$ for all  $x^n \in {\cal X}_{0,n}$. On the other hand, using Lemma~\ref{absolute_contunuity}, ${\overrightarrow Q}_{0,n}(\cdot|x^n)\ll\nu_{0,n}(\cdot),$ $\overleftarrow{P}_{0,n}-$almost $x^n\in{\cal X}_{0,n},$ and by Radon-Nikodym theorem, there exists a version of the RND $\bar{\xi}_{0,n}(x^n,y^n)\tri\frac{d{\overrightarrow Q}_{0,n}(\cdot|x^n)}{d\nu_{0,n}(\cdot)}(y^n)$ which is a non-negative measurable function of $(x^n,y^n)\in{\cal X}_{0,n}\times{\cal Y}_{0,n}.$ Hence another version of $\bar{\xi}_{0,n}(\cdot,\cdot)$ is $\bar{\xi}_{0,n}(x^n,y^n)=\frac{d({\overleftarrow P}_{0,n}\otimes{\overrightarrow Q}_{0,n})}{d({\overleftarrow P}_{0,n}\otimes\nu_{0,n})}(x^n,y^n)$.
We use notation ${\mathbb{I}}_{X^n\rightarrow{Y}^n}({\overleftarrow P}_{0,n}, {\overrightarrow Q}_{0,n})$ given in (\ref{section:properties:equations1}) to illustrate that $\mathbb{D}({\overleftarrow P}_{0,n} \otimes {\overrightarrow Q}_{0,n}|| \overrightarrow\Pi_{0,n})$ is a functional of $\big\{{\overleftarrow P}_{0,n}(\cdot|\cdot), {\overrightarrow Q}_{0,n}(\cdot|\cdot)\big\}\in{\cal Q}^{\bf C1}({\cal X}_{0,n}|{\cal Y}_{0,n-1})\times{\cal Q}^{\bf C2}({\cal Y}_{0,n}|{\cal X}_{0,n}).$
\vspace*{0.2cm}\\
\noi In the next Remark we summarize the equivalent definitions of directed information based on the two equivalent definitions of channels, that is, the one based on (\ref{equation29a}), (\ref{equation30}), and the one based on (\ref{equation202}), (\ref{equation6}). 
\begin{remark}\label{equivalent1}{\ \\}
Let ${\bf P}(\cdot|\cdot)\in{\cal Q}^{\bf C1}({\cal X}^{\mathbb{N}_0}|{\cal Y}^{\mathbb{N}_0})$ and ${\bf Q}(\cdot|\cdot)\in{\cal Q}^{\bf C2}({\cal Y}^{\mathbb{N}_0}|{\cal X}^{\mathbb{N}_0})$. By repeated application of Lemma~\ref{absolute_contunuity}, and the chain rule of relative entropy \cite[Theorem B.2.1., p. 326]{dupuis-ellis97}, directed information admits the following equivalent definitions.
\begin{align}
I(X^n&\rightarrow{Y}^n)\tri\sum_{i=0}^{n}I(X^i; Y_i | Y^{i-1})=\mathbb{D}({\overleftarrow P}_{0,n} \otimes {\overrightarrow Q}_{0,n}||{\overrightarrow\Pi}_{0,n})\label{equation33}\\
&=\int_{{\cal X}_{0,n} \times {\cal Y}_{0,n}} \log \Big( \frac{d{\overrightarrow Q}_{0,n}(\cdot|x^n)}{d\nu_{0,n}(\cdot)}(y^n)\Big)({\overleftarrow P}_{0,n}\otimes {\overrightarrow Q}_{0,n})(dx^n,dy^n)\equiv{\mathbb{I}}_{X^n\rightarrow{Y^n}}({\overleftarrow P}_{0,n}, {\overrightarrow Q}_{0,n}).\label{equation7a}
\end{align}
\end{remark}
\vspace*{.5cm}
\noi Clearly, (\ref{equation7a}) is valid even when  $({\overleftarrow P}_{0,n} \otimes {\overrightarrow Q}_{0,n})(dx^n,dy^n)$ is singular with respect to  $({\overleftarrow P}_{0,n}\otimes\nu_{0,n})(dx^n,dy^n)$, in which case its value is $+\infty$.
The point to be made here is that we will show the convexity, concavity, lower semicontinuity properties of directed information using the definition $I(X^n\rightarrow{Y^n})=\mathbb{D}({\overleftarrow P}_{0,n} \otimes {\overrightarrow Q}_{0,n}||{\overrightarrow\Pi}_{0,n})\equiv{\mathbb{I}}_{X^n\rightarrow{Y^n}}({\overleftarrow P}_{0,n}, {\overrightarrow Q}_{0,n}),$ as a functional of ${\overleftarrow P}_{0,n}(\cdot|y^{n-1})\in{\cal M}_1^{\bf C1}({\cal X}_{0,n})$ and ${\overrightarrow Q}_{0,n}(\cdot|x^n)\in{\cal M}_1^{\bf C2}({\cal Y}_{0,n})$. We will also use the directed information definition $\mathbb{D}({\overleftarrow P}_{0,n} \otimes {\overrightarrow Q}_{0,n}|| \overrightarrow\Pi_{0,n})$, as a functional of $\{{\overleftarrow P}_{0,n}, {\overrightarrow Q}_{0,n}\}$ to show lower semicontinuity, convexity and concavity properties. Then we will use these functional and topological properties to demonstrate how to establish existence of optimal solutions to the two extremum problems defined by (\ref{equation1dd}) and (\ref{introduction:nrdf:equation2}), respectively.\\

\subsection{Convexity and Concavity of Directed Information}

\par  First, we  show that the set of conditional distributions ${\bf P}(\cdot|{\bf y})\in{\cal M}^{\bf C1}_1({\cal X}^{\mathbb{N}_0})$ and ${\bf Q}(\cdot|{\bf x})\in{\cal M}^{\bf C2}_1({\cal Y}^{\mathbb{N}_0})$, i.e., satisfying consistency conditions {\bf C1} and {\bf C2}, are convex, and then we show convexity of directed information with respect to ${\bf Q}(\cdot|{\bf x})$ and concavity with respect to ${\bf P}(\cdot|{\bf y})$. \\
\noi Recall that the set of all distributions ${\bf P}(\cdot|{\bf y})\in{\cal M}_1({\cal X}^{\mathbb{N}_0})$ and ${\bf Q}(\cdot|{\bf x})\in{\cal M}_1({\cal Y}^{\mathbb{N}_0})$ (i.e., without imposing consistency conditions {\bf C1} and {\bf C2}) are convex, that is, given $\{{\bf P}^1(\cdot|{\bf y})$, ${\bf P}^2(\cdot|{\bf y})\}\in{\cal M}_1({\cal X}^{\mathbb{N}_0})\times{\cal M}_1({\cal X}^{\mathbb{N}_0})$, and $\lambda\in(0,1)$, there exists a probability measure $\tilde{P}$ on $({\cal X}^{\mathbb{N}_0}\times{\cal Y}^{\mathbb{N}_0},{\cal B}({\cal X}^{\mathbb{N}_0})\otimes{\cal B}({\cal Y}^{\mathbb{N}_0}))$ whose regular distribution $\tilde{P}(\cdot|{\bf y})$ satisfies $\tilde{P}(\cdot|{\bf y})=\lambda{\bf P}^1(\cdot|{\bf y})+(1-\lambda){\bf P}^2(\cdot|{\bf y})\in{\cal M}_1({\cal X}^{\mathbb{N}_0})$.\\
\noi Next, we show convexity of the sets ${\cal M}_1^{\bf C1}({\cal X}^{\mathbb{N}_0})$ and ${\cal M}_1^{\bf C2}({\cal Y}^{\mathbb{N}_0})$.
\begin{theorem}(Convexity of sets ${\cal M}_1^{\bf C1}({\cal X}^{\mathbb{N}_0})$, ${\cal M}_1^{\bf C2}({\cal Y}^{\mathbb{N}_0})$)\label{convexity_of_sets}{\ \\}
Let $\{{\cal X}_n:n\in{\mathbb{N}_0}\},$ $\{{\cal Y}_n:n\in{\mathbb{N}_0}\}$ be Polish spaces with ${\cal B}({\cal X}_n),$ ${\cal B}({\cal Y}_n)$, respectively, the $\sigma-$algebras of Borel sets. Then the sets of distributions ${\bf P}(\cdot|{\bf y})\in{\cal M}_1^{\bf C1}({\cal X}^{\mathbb{N}_0})$ and ${\bf Q}(\cdot|{\bf x})\in{\cal M}_1^{\bf C2}({\cal Y}^{\mathbb{N}_0})$ are convex, and similarly, their projection to finite number of coordinates, that is, $\overleftarrow{P}_{0,n}(\cdot|{y}^{n-1})\in{\cal M}_1^{\bf C1}({\cal X}_{0,n})$ and $\overrightarrow{Q}_{0,n}(\cdot|{x}^{n})\in{\cal M}_1^{\bf C2}({\cal Y}_{0,n})$, are also convex.
\end{theorem}
\begin{proof}
Since the methodology is similar for both sets, only the derivation for ${\cal M}_1^{\bf C1}({\cal X}^{\mathbb{N}_0})$ is given. By definition, the set of distributions ${\cal M}_1^{\bf C1}({\cal X}^{\mathbb{N}_0})$ is convex if for a given $\{{\bf P}^1(\cdot|{\bf y}), {\bf P}^2(\cdot|{\bf y})\}\in{\cal M}_1^{\bf C1}({\cal X}^{\mathbb{N}_0})\times{\cal M}_1^{\bf C1}({\cal X}^{\mathbb{N}_0})$, and a given $\lambda\in(0,1)$, there exists a probability measure ${\tilde{P}}$ on $({\cal X}^{\mathbb{N}_0}\times{\cal Y}^{\mathbb{N}_0},{\cal B}({\cal X}^{\mathbb{N}_0})\otimes{\cal B}({\cal Y}^{\mathbb{N}_0})$, whose regular conditional measure $\tilde{P}(\cdot|{\bf y})$ is a convex combination $\tilde{P}(\cdot|{\bf y})=\lambda{\bf P}^1(\cdot|{\bf y})+(1-\lambda){\bf P}^2(\cdot|{\bf y}),~a.e.~{\bf y}\in{\cal Y}^{\mathbb{N}_0}$, and consistency condition ${\bf C1}$ holds, i.e., $\lambda{\bf P}^1(\cdot|{\bf y})+(1-\lambda){\bf P}^2(\cdot|{\bf y})\in{\cal M}_1^{\bf C1}({\cal X}^{\mathbb{N}_0})$. By \cite{shiryaev1984}, the set of distributions ${\cal M}_1({\cal X}^{\mathbb{N}_0})$ is convex, and since $\{{\bf P}^1(\cdot|{\bf y}), {\bf P}^2(\cdot|{\bf y})\}\in{\cal M}_1({\cal X}^{\mathbb{N}_0})$\\ $\times{\cal M}_1({\cal X}^{\mathbb{N}_0})$, then there is a probability measure $\tilde{{P}}$ on ${\cal M}_1\big({\cal X}^{\mathbb{N}_0}\times{\cal Y}^{\mathbb{N}_0},{\cal B}({\cal X}^{\mathbb{N}_0}\otimes{\cal B}({\cal Y}^{\mathbb{N}_0}))$, whose regular distribution $\tilde{P}(\cdot|{\bf y})$, ${\bf y}\in{\cal Y}^{\mathbb{N}_0}$, satisfies
\begin{align}
\tilde{P}(\cdot|{\bf y})=\lambda{\bf P}^1(\cdot|{\bf y})+(1-\lambda){\bf P}^2(\cdot|{\bf y})\in{\cal M}_1({\cal X}^{\mathbb{N}_0}),~~\forall\lambda\in(0,1).\nonumber
\end{align}
Moreover, if ${\bf P}^1(\cdot|{\bf y})$, and ${\bf P}^2(\cdot|{\bf y})$ satisfy consistency condition ${\bf C1}$, then their convex combination also satisfies consistency condition {\bf C1}, and consequently $\lambda{{\bf P}}^1(\cdot|{\bf y})+(1-\lambda){{\bf P}}^2(\cdot|{\bf y})\in{\cal M}_1^{\bf C1}({\cal X}^{\mathbb{N}_0})$, i.e., the consistency condition ${\bf C1}$ holds. The derivation for ${\bf Q}(\cdot|{\bf x})\in{\cal M}_1^{\bf C2}({\cal Y}^{\mathbb{N}_0})$ is similar.  The derivation for the projection to finite number of coordinates is done as follows.  Let $A^{(n)}=\{{\bf x}:x_n{\in}A\},$ $A\in{\cal B}({\cal X}_n),$  and   let ${\bf P}(A^{(n)}|{{\cal B}({\cal X}_{0,n-1})}|{\bf y})$ denote the conditional probability of $A^{(n)}$ with respect to ${\cal B}({\cal X}_{0,n-1})$ calculated on the probability space $\big({\cal X}^{\mathbb{N}},{\cal B}({\cal X}^{\mathbb{N}}),{\bf  P}(\cdot|{\bf y})\big).$ From the definition of regular conditional probability measures,  it follows that
\begin{align}
\tilde{ P}(A^{(n)}|{\cal B}({\cal X}_{0,n-1})|{\bf y})&=\lambda{\bf P}^1(A^{(n)}|{\cal B}({\cal X}_{0,n-1})|{\bf y})+(1-\lambda){\bf P}^2(A^{(n)}|{\cal B}({\cal X}_{0,n-1})|{\bf y})-a.s.\nonumber\\
&=\lambda{p_n^1}(A|x^{n-1},y^{n-1})+(1-\lambda){p_n^2}(A|x^{n-1},y^{n-1})-a.s.\nonumber
\end{align}
where ${p_n^1}(\cdot|x^{n-1},y^{n-1}), {p_n^2}(\cdot|x^{n-1},y^{n-1})$ are regular conditional distributions.    Since convex combination of regular conditional distributions is also a regular conditional distribution, by Remark~\ref{equivalent}  the set $\overleftarrow{P}_{0,n}(\cdot| y^{n-1}) \in {\cal M}_1^{\bf C1}({\cal X}_{0,n})$ is convex, and the derivation is complete.
\end{proof}
\vspace*{0.5cm}
\noi Since ${\cal M}_1^{\bf C1}({\cal X}_{0,n})$ and ${\cal M}_1^{\bf C2}({\cal Y}_{0,n})$ are convex, then we proceed further to show that directed information ${\mathbb I}_{X^n\rightarrow{Y}^n}({\overleftarrow P}_{0,n},{\overrightarrow Q}_{0,n})$, as a functional of ${\overleftarrow P}_{0,n}(\cdot|y^{n-1})\in{\cal M}_1^{\bf C1}({\cal X}_{0,n})$, for a fixed ${\overrightarrow Q}_{0,n}(\cdot|x^n)\in{\cal M}_1^{\bf C2}({\cal Y}_{0,n})$, is concave, and as a functional of ${\overrightarrow Q}_{0,n}(\cdot|x^n)\in{\cal M}_1^{\bf C2}({\cal Y}_{0,n})$, for a fixed ${\overleftarrow P}_{0,n}(\cdot|y^{n-1})\in{\cal M}_1^{\bf C1}({\cal X}_{0,n})$, is convex. These results are shown in the next theorem.
\begin{theorem}(Convexity of conditional distributions)\label{convexity1}{\ \\}
Let $\{{\cal X}_n:n\in{\mathbb{N}_0}\},$ $\{{\cal Y}_n:n\in{\mathbb{N}_0}\}$ be Polish spaces with ${\cal B}({\cal X}_n),$ ${\cal B}({\cal Y}_n)$, respectively, the $\sigma-$algebras of Borel sets. Consider the directed information functional $I(X^n\rightarrow{Y}^n)={\mathbb I}_{X^n\rightarrow{Y}^n}$ $({\overleftarrow P}_{0,n},{\overrightarrow Q}_{0,n}),$ ${\mathbb I}_{X^n\rightarrow{Y}^n}:{\cal M}_1^{\bf C1}({\cal X}_{0,n})\times{\cal M}_1^{\bf C2}({\cal Y}_{0,n})\mapsto[0,\infty]$ defined by (\ref{equation7a}).\\
Then the following hold.
\item[1)] ${\mathbb{I}}_{X^n\rightarrow{Y}^n}({\overleftarrow P}_{0,n},{\overrightarrow Q}_{0,n})$ is a convex functional of ${\overrightarrow Q}_{0,n}(\cdot|x^n)\in{\cal M}_1^{\bf C2}({\cal Y}_{0,n})$ for a fixed ${\overleftarrow P}_{0,n}(\cdot|y^{n-1})\in{\cal M}_1^{\bf C1}({\cal X}_{0,n})$.
\item[2)] ${\mathbb{I}}_{X^n\rightarrow{Y}^n}({\overleftarrow P}_{0,n},{\overrightarrow Q}_{0,n})$ is a concave functional of ${\overleftarrow P}_{0,n}(\cdot|y^{n-1})\in{\cal M}_1^{\bf C1}({\cal X}_{0,n})$ for a fixed ${\overrightarrow Q}_{0,n}(\cdot|x^n)\in{\cal M}_1^{\bf C2}({\cal Y}_{0,n})$.
\item[3)] ${\mathbb{I}}_{X^n\rightarrow{Y}^n}({\overleftarrow P}_{0,n},\cdot)$ is a strictly convex functional on the set $\big\{{\overrightarrow Q}_{0,n}(\cdot|x^n)\in{\cal M}_1^{\bf C2}({\cal Y}_{0,n}):{\mathbb{I}}_{X^n\rightarrow{Y}^n}({\overleftarrow P}_{0,n},{\overrightarrow Q}_{0,n})$ $<\infty\big\}$ for a fixed ${\overleftarrow P}_{0,n}(\cdot|y^{n-1})\in{\cal M}_1^{\bf C1}({\cal X}_{0,n})$.
\end{theorem}
\begin{proof} By Theorem~\ref{convexity_of_sets}, the sets ${\cal M}_1^{\bf C1}({\cal X}_{0,n})$ and ${\cal M}_1^{\bf C2}({\cal Y}_{0,n})$ are convex. Therefore, to show parts 1), 2), 3) we utilize the consistency of the two families of conditional distributions and we apply the log-sum formulae, and the existence of certain Radon-Nikodym Derivatives (RNDs). The complete derivation is given in Appendix~\ref{convexity_of_functionals}.  
\end{proof}
\vspace*{.2cm}
\noi Theorem~\ref{convexity1} is analogous to mutual information $I(X^n;Y^n)\equiv{\mathbb{I}}_{X^n;{Y}^n}(P_{X^n},P_{Y^n|X^n})$, expressed as a functional of input distribution $P_{X^n}(\cdot)\in{\cal M}_1({\cal X}_{0,n})$ and the channel $P_{Y^n|X^n}(\cdot|x^n)\in{\cal M}_1({\cal Y}_{0,n})$, which is known to be a convex (respectively concave) functional of $P_{Y^n|X^n}(\cdot|x^n)\in{\cal M}_1({\cal Y}_{0,n})$ \big(respectively $P_{X^n}(\cdot)\in{\cal M}_1({\cal X}_{0,n})$\big), for a fixed $P_{X^n}(\cdot)\in{\cal M}_1({\cal X}_{0,n})$ \big(respectively $P_{Y^n|X^n}(\cdot|x^n)\in{\cal M}_1({\cal Y}_{0,n})$\big). It is important to point out that if one considers the alternative definition of directed information (\ref{equation29}), (\ref{equation30}), as a functional of the sequence of input channel distributions, $I(X^n\rightarrow{Y^n})\equiv\mathbb{I}_{X^n\rightarrow{Y^n}}(p_i(\cdot|\cdot,\cdot),q_i(\cdot|\cdot,\cdot):~i=0,1,\ldots,n)$, then it is not clear to us whether it is possible to establish convexity and concavity with respect to $q_i$ and $p_i$.   \\
\noi For finite alphabet spaces, the convexity of the set of causally conditioned probability mass functions $P(x^n||y^{n-1})\tri\prod_{i=0}^n{p}(x_i|x^{i-1},y^{i-1})$ and $Q(y^n||x^n)\tri\prod_{i=0}^n{q}(y_i|y^{i-1},x^i)$ is shown in \cite[Lemma 1]{permuter-asnani-weissman2014ieeeit}, under the assumption that for each $n\in\mathbb{N}_0$, the ratios $\frac{P(x^n||y^{n-1})}{P(x^{n-1}||y^{n-1})}$ and $\frac{Q(y^n||x^{n})}{Q(y^{n-1}||x^{n-1})}$ exist, and they are given by ${p}(x_n|x^{n-1},y^{n-1})$ and ${q}(y_n|y^{n-1},x^n)$, respectively. The derivation in \cite{permuter-asnani-weissman2014ieeeit} is based on showing that the set of all causally conditioned distributions $P(x^n||y^{n-1})$ is a polyhedron. The method described in \cite{permuter-asnani-weissman2014ieeeit} does not apply to conditional distributions defined on continuous alphabets. Theorem~\ref{convexity_of_sets} and Theorem~\ref{convexity1}, hold for general conditional distributions defined on abstract alphabet spaces, and they do not require existence of probability density functions (corresponding to the causally conditioned distributions for each $n\in\mathbb{N}_0$), hence they compliment the work in \cite{permuter-asnani-weissman2014ieeeit}.

%%%%%%%%%%%%%%%%%%%%%%%%%%%%%%%%%%%%%%%%%%%%%%%%%%%%%%%%%%%%%%%%%%%%%%%%%%%%%%%%%%%%%%%%%%%%

\subsection{Weak Convergence and Compactness of Conditional Distributions}

\par In this section we give general sufficient conditions for weak compactness of the set of probability distributions $\overleftarrow{P}_{0,n}(\cdot|y^{n-1})\in{\cal M}_1^{\bf C1}({\cal X}_{0,n})$ and $\overrightarrow{Q}_{0,n}(\cdot|x^n)\in{{\cal M}_1^{\bf C2}}({\cal Y}_{0,n}),$ and compactness of the set of joint and marginal measures with respect to the topology of weak convergence of probability measures. These conditions are sufficient to show lower semicontinuity of ${\mathbb{I}}_{X^n\rightarrow{Y}^n}({\overleftarrow P}_{0,n},{\overrightarrow Q}_{0,n})$ for fixed ${\overleftarrow P}_{0,n}(\cdot|y^{n-1})\in{\cal M}_1^{\bf C1}({\cal X}_{0,n})$ \big(respectively ${\overrightarrow Q}_{0,n}(\cdot|x^n)\in{\cal M}_1^{\bf C2}({\cal Y}_{0,n})$\big) with respect to ${\overrightarrow Q}_{0,n}(\cdot|x^n)\in{\cal M}_1^{\bf C2}({\cal Y}_{0,n})$ \big(respectively ${\overleftarrow P}_{0,n}(\cdot|y^{n-1})\in{\cal M}_1^{\bf C1}({\cal X}_{0,n})$\big). The lower semicontinuity of directed information is the analogue of the lower semicontinuity of mutual information, extensively utilized in information theory and statistics (see \cite{csiszar74,csiszar92}). 
\vspace*{0.2cm}\\
\noi Before we state the main theorem, we introduce the following notation.
Let $BC({\cal X})$ denote the set of bounded, continuous real-valued function $f$ defined on a metric space $({\cal X},d)$ endowed with the supremum norm $||f||=\sup_{x\in{\cal X}}|f(x)|$. A sequence of probability measures $\{P_\alpha:~\alpha=1,2,\ldots\}\subset{\cal M}_1({\cal X})$ is said to {\it converge weakly} to a probability measure $P\in{\cal M}_1({\cal X})$ if \cite{billingsley1999}
\begin{align}
\lim_{\alpha\rightarrow\infty}\int_{\cal X}f(x)dP_\alpha(x)=\int_{\cal X}f(x)dP(x),~\forall{f}\in{BC}({\cal X}).\nonumber  
\end{align}
Weak convergence of $\{P_\alpha:~\alpha=1,2,\ldots\}$ to $P$ is denoted by $P_\alpha\buildrel w \over \longrightarrow{P}$. A family of probability measures $M\subset{\cal M}_1({\cal X})$ is called {\it relatively compact or weakly compact} if every sequence in $M$ contains a weakly convergent subsequence that converges to ${\cal M}_1({\cal X})$ but not necessarily to $M$. Appendix~\ref{backround_material} summarizes well-known theorems of weak convergence, compactness, tightness, and Prohorov's theorem, which we invoke to derive the results of this section. \\
Throughout sequences of points in ${\cal X}^{\mathbb{N}_0}$ and ${\cal Y}^{\mathbb{N}_0}$ are denoted by ${\bf x}^{(\alpha)}\tri\{x^{(\alpha)}_0,x^{(\alpha)}_1,\ldots\}\in{\cal X}^{\mathbb{N}_0},$ ${\bf y}^{(\alpha)}\tri\{y^{(\alpha)}_0,y^{(\alpha)}_1,\ldots\}\in{\cal Y}^{\mathbb{N}_0},$ $\alpha=1,2,\ldots$ Moreover, a sequence of points ${\bf x}^{(\alpha)}\in{\cal X}^{\mathbb{N}_0},~\alpha=1,2,\ldots$ is said to converge to ${\bf x}^{(o)}\in{\cal X}^{\mathbb{N}_0}$ as $\alpha\longrightarrow\infty$, if $\lim_{\alpha\longrightarrow\infty}x^{(\alpha)}_n={x}^{(o)}_n~\mbox{for every}~n\in\mathbb{N}_0$. Sequences of such  points in ${\cal X}_{0,n}\tri\times^{n}_{i=0}{\cal X}_i$ and ${\cal Y}_{0,n}\tri \times^{n}_{i=0}{\cal Y}_i$ are denoted by $x^{n,(\alpha)}\tri\{x_0^{(\alpha)},x_1^{(\alpha)},\ldots,x_n^{(\alpha)}\}$ and $y^{n,(\alpha)}\tri\{{y_0^{(\alpha)},y_1^{(\alpha)},\ldots,y_n^{(\alpha)}}\},~\alpha=1,2,\ldots$.\\
\noi The next remark, is introduced to illustrate that in applications of weak convergence of probability distributions, weak continuity of probability distributions is natural, when analyzing conditional distributions with discontinuities, such as, distributions induced by mixture of discrete and continuous RVs. 

\begin{remark}(Weak continuity vs. Strong continuity){\ \\}
Let $q(\cdot|\cdot)\in{\cal Q}({\cal Y}|{\cal X})$ be a conditional distribution, and suppose there is a distribution $\mu(dx)\in{\cal M}_1({\cal X})$ such that for every $x\in{\cal X}$, $q(\cdot|x)$ has a density $\bar{q}(\cdot|x)$ with respect to $\mu(\cdot)$, i.e.,
\begin{align*}
q(B|x)=\int_{B}\bar{q}(y|x)\mu(dx),~\forall{B}\in{\cal B}({\cal Y}),~\forall{x}\in{\cal X}.
\end{align*}
For example, if ${\cal X}\in\mathbb{R}$ then $\mu(dx)=dx$ is the Lebesgue measure on $\mathbb{R}$. If $\bar{q}(y|\cdot)$ is continuous on ${\cal X}$ for every $y\in{\cal Y}$, then $q(\cdot|\cdot)\in{\cal Q}({\cal Y}|{\cal X})$ is strongly continuous \big(i.e., $q(B|\cdot)$ is continuous on ${\cal X}$ for every $B\in{\cal B}({\cal Y})$\big). Strong continuity of channel models is rather restrictive, because it rules out conditional distributions which have discontinuities, such as, additive noise channels, in which noise is a mixture of a continuous RV (i.e., Gaussian distributed RV) and a finite alphabet valued   RV.\\
Consider a channel model with feedback described by the nonlinear recursive equation
\begin{align*}
Y_n=h_n(Y^{n-1},X_n,V_n),~Y^{-1}=y^{-1},~n=0,1,\ldots
\end{align*}
where $\{h_n:{\cal Y}_{0,n-1}\times{\cal X}_n\times{\cal V}_n\longmapsto{\cal Y}_n: n=0,1, \ldots,\}$, is a sequence of measurable functions and $\{V_n:~n=0,1,\ldots\}$ is a sequence of $\{{\cal V}_n:~n=0,1,\ldots\}$-valued RV's, representing the channel noise.\\
Suppose the following condition holds.
\begin{align*}
P_{V_n|V^{n-1},X^n,Y^{n-1}}(dv_n|v^{n-1},x^n,y^{n-1})=P_{V_n}(v_n),~n=0,1,\ldots.
\end{align*}
Then the channel distribution induced by the above model is
\begin{align*}
q_n(B|y^{n-1},x^n)&=\mathbb{P}\{Y_n\in{B}|Y^{n-1}=y^{n-1},X^n=x^n\},~B\in{\cal B}({\cal Y})\\
&=\mathbb{P}\{h_n(Y^{n-1},X_n,V_n)\in{B}|Y^{n-1}=y^{n-1},X^n=x^n\},\\
&={\bf P}(\{v_n\in{V}_n:~h_n(y^{n-1},x_n,v_n)\in{B}\})=\int_{{\cal V}_n}I_B\big(h_n(y^{n-1},x_n,v_n)\big)P_{V_n}(dv_n)\\
&\equiv{q}_n(B|y^{n-1},x_n)
\end{align*}
where $I_B(\cdot)$ is the indicator function. If for each $n$, the function $h_n(\cdot,\cdot,v_n)$ is continuous on ${\cal Y}_{0,n-1}\times{\cal X}_n$ for every $v_n\in{\cal V}_n$, $n=0,1,\ldots$, then by bounded convergence theorem $\{q_n(\cdot|\cdot,\cdot)\in{\cal Q}({\cal Y}_n|{\cal Y}_{0,n-1}\times{\cal X}_n):~n=0,1,\ldots\}$ is weakly continuous (see Definition~\ref{App_SW}), i.e., for each sequence $\{(y^{n-1, (\alpha)}, x_n^{(\alpha)}): \alpha =1, \ldots\} \subset {\cal Y}_{0,n-1} \times {\cal X}_{n}$ such that $(y^{n-1, (\alpha)}, x_n^{(\alpha)}) \longrightarrow (y^{n-1, (o)}, x_n^{(o)})$, then $\lim_{\alpha \longrightarrow \infty}  \int_{{\cal Y}_n} g(y_n) q_n(dy_n|y^{n-1, (\alpha)},x_n^{(\alpha)}) =\int_{{\cal Y}_n} g(y_n) q_n(dy_n|y^{n-1, (o)},x_n^{(o)})$, for all bounded continuous functions $g(\cdot) \in BC({\cal Y}_n)$. Hence, no requirement is imposed on the distribution of $\{P_{V_n}(\cdot)\in{\cal M}_1({\cal V}_n):~n=0,1,\ldots\}$.\\
On the other hand, consider the special case of an additive channel,  of the form 
\begin{align*}
Y_n=\bar{h}_n(Y^{n-1},X_n)+V_n,~n=0,1,\ldots
\end{align*} 
where  $P_{V_n}(dv_n)$ is assumed to have a density, $\bar{p}(v_n)$, i.e.,  $P_{V_n}(dv_n)=\bar{p}(v_n)dv_n,~n=0,1,\ldots$. Then $\{q_n(\cdot|\cdot,\cdot)\in{\cal Q}({\cal Y}_n|{\cal Y}_{0,n-1}\times{\cal X}_n):~n=0,1,\ldots\}$ is strongly continuous if for each $n$, $\bar{h}(\cdot,\cdot)$ is continuous on ${\cal Y}_{0,n-1}\times{\cal X}_n$ and $\bar{p}(\cdot)$ is continuous on ${\cal V}_n$, for $n=0,1,\ldots$.\\
Clearly, when proving properties of mutual information or directed information, weak continuity is more general (less restrictive), compared to strong continuity, which by definition rules out many interesting application examples. 
\end{remark}

\noi Next, we state the main theorem which is also used to show lower semicontinuity of directed information. The theorem consists of two parts depending on whether, A) ${\cal Y}_{0,n}$ is compact and $p_n(dx_n|\cdot,\cdot)$ as a function of $(x^{n-1},y^{n-1})\in{\cal X}_{0,n-1}\times{\cal Y}_{0,n-1}$ is weakly continuous, and B) ${\cal X}_{0,n}$ is compact and $q_n(dy_n|\cdot,\cdot)$ as a function of $(y^{n-1},x^{n})\in{\cal Y}_{0,n-1}\times{\cal X}_{0,n}$ is weakly continuous. In applications of information theory either one of them or both maybe required, depending on the context of the application considered.
\begin{theorem}\label{weak_convergence}{\ \\}
{\bf Part A.} For each $n\in\mathbb{N}_0$, let ${\cal Y}_{0,n}$ be a compact Polish space, ${\cal X}_{0,n}$ a Polish space, and assume the collection of conditional distributions $\{p_n(\cdot|\cdot,\cdot)\in{\cal Q}({\cal X}_n|{\cal X}_{0,n-1}\times{\cal Y}_{0,n-1}):~n\in\mathbb{N}_0\}$ satisfy the following condition.
\vspace*{0.2cm}\\
\noi{\bf CA}: For all $g(\cdot){\in}BC({\cal X}_{0,n})$, the function
\begin{align}
(x^{n-1},y^{n-1})\in{\cal X}_{0,n-1}\times{\cal Y}_{0,n-1}\longmapsto\int_{{\cal X}_n}g(x)p_n(dx|x^{n-1},y^{n-1})\in\mathbb{R}\label{condition1}
\end{align}
is continuous jointly in the variables $(x^{n-1},y^{n-1})\in{\cal X}_{0,n-1}\times{\cal Y}_{0,n-1}$.\\
\vspace*{0.2cm}
\noi Then the following hold.
\item[A1)] Let ${\overleftarrow P}_{0,n}(\cdot|y^{n-1})\in{\cal M}_1^{\bf C1}({\cal X}_{0,n})$ and consider a sequence of forward channels $\big\{{\overrightarrow Q}_{0,n}^{\alpha}(\cdot|{ x^n}):\alpha=1,2,\ldots\big\}\subset{\cal M}_1^{\bf C2}({\cal Y}_{0,n})$. Then the sequence of joint measures $\{({\overleftarrow P}_{0,n}\otimes\overrightarrow{Q}^{\alpha}_{0,n}):~\alpha=1,2,\ldots\}$ converges weakly to a  joint measure $P^o(dx^n,dy^n)$, that is, 
\begin{align}
({\overleftarrow P}_{0,n}\otimes\overrightarrow{Q}^{\alpha}_{0,n})(dx^n,dy^n)\buildrel w \over\longrightarrow{P}^o(dx^n,dy^n)=({\overleftarrow P}_{0,n}\otimes\bar{Q}^o_{0,n})(dx^n,dy^n)\in{\cal M}_1({\cal X}_{0,n}\times{\cal Y}_{0,n})\label{theorem:weak:compactness:eq.1}
\end{align}
where the joint measure ${P}^o(dx^n,dy^n)$ corresponds to the same backward channel $\overleftarrow{P}_{0,n}(\cdot|y^{n-1})\in{\cal M}_1^{\bf C1}({\cal X}_{0,n})$ and a forward channel $\bar{Q}_{0,n}^o(\cdot|x^n)\in{\cal M}_1({\cal Y}_{0,n})$ (i.e., not necessarily in ${\cal M}_1^{\bf C2}({\cal Y}_{0,n})$). Equivalently, $\{({\overleftarrow P}_{0,n}\otimes\overrightarrow{Q}^{\alpha}_{0,n}):~\alpha=1,2,\ldots\}$ is relatively or weakly compact.\\
\noi Moreover, the corresponding sequence of marginal measures $\{\nu_{0,n}^{\alpha}(\cdot)\in{\cal M}_1({\cal Y}_{0,n}):\alpha=1,2,\ldots\}$ on ${\cal Y}_{0,n}$ and $\{\mu_{0,n}^{\alpha}(\cdot)\in{\cal M}_1({\cal X}_{0,n}):\alpha=1,2,\ldots\}$ on ${\cal X}_{0,n}$, converges weakly, that is, 
\begin{align*}
\nu_{0,n}^{\alpha}(dy^n)\buildrel w \over\longrightarrow\nu_{0,n}^o(dy^n)~\mbox{and}~\mu_{0,n}^{\alpha}(dx^n)\buildrel w \over\longrightarrow\mu_{0,n}^o(dx^n)
\end{align*}
where $\nu_{0,n}^o(\cdot)\in{\cal M}_1({\cal Y}_{0,n})$ and $\mu_{0,n}^o(\cdot)\in{\cal M}_1({\cal X}_{0,n})$ are the marginals of the joint measure in \eqref{theorem:weak:compactness:eq.1}.
\item[A2)] The  set of measures $\overleftarrow{P}_{0,n}(\cdot|y^{n-1})\in{\cal M}_1^{\bf C1}({\cal X}_{0,n})$ is uniformly tight.
% for each $y^{n-1}\in{\cal Y}_{0,n-1}$.
\item[A3)] The  set of measures $\overrightarrow{Q}_{0,n}(\cdot|x^n)\in{\cal M}_1^{\bf C2}({\cal Y}_{0,n})$ is relatively compact.
\item[A4)] Let ${\overleftarrow P}_{0,n}(\cdot|y^{n-1})\in{\cal M}_1^{\bf C1}({\cal X}_{0,n}),$ $\big\{{\overrightarrow Q}_{0,n}^{\alpha}(\cdot|{ x^n}):\alpha=1,2,\ldots\big\}\subset{\cal M}_1^{\bf C2}({\cal Y}_{0,n})$, where $\{\nu_{0,n}^{\alpha}(\cdot)\in{\cal M}_1({\cal Y}_{0,n}):\alpha=1,2,\ldots\}$ are the marginals of $\big\{({\overleftarrow P}_{0,n}\otimes{\overrightarrow Q}_{0,n}^\alpha)(dx^n,dy^n)\in{\cal M}_1({\cal X}_{0,n}\times{\cal Y}_{0,n}):\alpha=1,2,\ldots\big\}.$  Then 
\begin{align*}
\overrightarrow{\Pi}_{0,n}^\alpha(dx^n,dy^n)\equiv{\overleftarrow P}_{0,n}(dx^n|dy^{n-1})\otimes\nu_{0,n}^{\alpha}(dy^n)\buildrel w \over\longrightarrow{\overleftarrow P}_{0,n}(dx^n|dy^{n-1})\otimes\nu_{0,n}^o(dy^n)\equiv\overrightarrow{\Pi}_{0,n}^o(dx^n,dy^n)
\end{align*}
where $\nu_{0,n}^o(\cdot)\in{\cal M}_1({\cal Y}_{0,n})$ is the weak limit of the marginal in \eqref{theorem:weak:compactness:eq.1}.
\vspace*{0.2cm}\\
\noi{\bf Part B.} For each $n\in\mathbb{N}_0$, let ${\cal X}_{0,n}$ be a compact Polish space, ${\cal Y}_{0,n}$ a Polish space, and assume the collection of conditional distributions $\{q_n(\cdot|\cdot,\cdot)\in{\cal Q}({\cal Y}_n|{\cal Y}_{0,n-1}\times{\cal X}_{0,n}):~n\in\mathbb{N}_0\}$ satisfy the following condition.
\vspace*{0.2cm}\\
\noi {\bf CB}: For all $h(\cdot){\in}BC({\cal Y}_{0,n})$, the function
\begin{align}
(x^{n},y^{n-1})\in{\cal X}_{0,n}\times{\cal Y}_{0,n-1}\longmapsto\int_{{\cal Y}_n}h(y)q_n(dy|y^{n-1},x^n)\in\mathbb{R}\label{condition2}
\end{align}
is continuous jointly in the variables $(x^{n},y^{n-1})\in{\cal X}_{0,n}\times{\cal Y}_{0,n-1}$.\\
\vspace*{0.2cm}
\noi Then the following hold.
\item[B1)] Let ${\overrightarrow Q}_{0,n}(\cdot|x^{n})\in{\cal M}_1^{\bf C2}({\cal Y}_{0,n})$ and consider a sequence of backward channels $\big\{{\overleftarrow P}_{0,n}^{\alpha}(\cdot| y^{n-1}):\alpha=1,2,\ldots\big\}\subset{\cal M}_1^{\bf C1}({\cal X}_{0,n})$. Then, the joint measures $\{(\overleftarrow{P}_{0,n}^{\alpha}\otimes{\overrightarrow Q}_{0,n}):~\alpha=1,2,\ldots\}$ converges weakly to a joint measure $P^o(dx^n,dy^n)$, that is,
\begin{align}
(\overleftarrow{P}_{0,n}^{\alpha}\otimes{\overrightarrow Q}_{0,n})(dx^n,dy^n)\buildrel w \over\longrightarrow{P}^o(dx^n,dy^n)=(\bar{P}_{0,n}^o\otimes{\overrightarrow Q}_{0,n})(dx^n,dy^n)\in{\cal M}_1({\cal X}_{0,n}\times{\cal Y}_{0,n})
\label{theorem:weak:compactness:eq.2}
\end{align}
where the joint measure $P^o(dx^n,dy^n)$ corresponds to the same forward channel ${\overrightarrow Q}_{0,n}(\cdot|x^{n})\in{\cal M}_1^{\bf C2}({\cal Y}_{0,n})$ and a backward channel $\bar{P}_{0,n}^o(\cdot|y^{n-1})\in{\cal M}_1({\cal X}_{0,n})$ (i.e., not necessarily in ${\cal M}_1^{\bf C1}({\cal X}_{0,n})$). Equivalently, $\{(\overleftarrow{P}^{\alpha}_{0,n}\otimes\overrightarrow{Q}_{0,n}):~\alpha=1,2,\ldots\}$ is relatively or weakly compact.\\
\noi  Moreover, the corresponding sequence of marginal measures $\{\nu_{0,n}^{\alpha}(\cdot)\in{\cal M}_1({\cal Y}_{0,n}):\alpha=1,2,\ldots\}$ on ${\cal Y}_{0,n}$ and $\{\mu_{0,n}^{\alpha}(\cdot)\in{\cal M}_1({\cal X}_{0,n}):\alpha=1,2,\ldots\}$ on ${\cal X}_{0,n}$, converges weakly, that is, 
\begin{align*}
\nu_{0,n}^{\alpha}(dy^n)\buildrel w \over\longrightarrow\nu_{0,n}^o(dy^n)~\mbox{and}~\mu_{0,n}^{\alpha}(dx^n)\buildrel w \over\longrightarrow\mu_{0,n}^o(dx^n)
\end{align*}
where $\nu_{0,n}^o(\cdot)\in{\cal M}_1({\cal Y}_{0,n})$ and $\mu_{0,n}^o(\cdot)\in{\cal M}_1({\cal X}_{0,n})$ are the marginals of \eqref{theorem:weak:compactness:eq.2}.
\item[B2)] The  set of measures $\overrightarrow{Q}_{0,n}(\cdot|x^n)\in{\cal M}_1^{\bf C2}({\cal Y}_{0,n})$ in uniformly tight.
\item[B3)] The  set of measures $\overleftarrow{P}_{0,n}(\cdot|y^{n-1})\in{\cal M}_1^{\bf C1}({\cal X}_{0,n})$ is relatively compact.
\item[B4)] Let ${\overrightarrow Q}_{0,n}(\cdot|x^{n})\in{\cal M}_1^{\bf C2}({\cal Y}_{0,n})$, $\big\{{\overleftarrow P}_{0,n}^{\alpha}(\cdot|{ y^{n-1}}):\alpha=1,2,\ldots\big\}\subset{\cal M}_1^{\bf C1}({\cal X}_{0,n})$, where $\{\mu_{0,n}^{\alpha}(\cdot)\in{\cal M}_1({\cal X}_{0,n}):\alpha=1,2,\ldots\}$ are the marginals of $\big\{({\overleftarrow P}_{0,n}^\alpha\otimes{\overrightarrow Q}_{0,n})(dx^n,dy^n)\in{\cal M}_1({\cal X}_{0,n}\times{\cal Y}_{0,n}):\alpha=1,2,\ldots\big\}$. Then
\begin{align*}
\overleftarrow{\Pi}^\alpha(dx^n,dy^n)\equiv{\overrightarrow Q}_{0,n}(dy^n|dx^{n})\otimes\mu_{0,n}^{\alpha}(dx^n)\buildrel w \over\longrightarrow{\overrightarrow Q}_{0,n}(dy^n|dx^{n})\otimes\mu_{0,n}^o(dx^n)\equiv\overleftarrow{\Pi}^o(dx^n,dy^n)
\end{align*}
where $\mu_{0,n}^o(\cdot)\in{\cal M}_1({\cal X}_{0,n})$ is the weak limit of the marginal in \eqref{theorem:weak:compactness:eq.2}. 
\end{theorem}
\begin{proof}
See Appendix~\ref{proof_weak_convergence}.
\end{proof}
\vspace*{.2cm}
\noi Note that additional conditions are required to show that the limiting joint distribution \eqref{theorem:weak:compactness:eq.1} (respectively, \eqref{theorem:weak:compactness:eq.2}) corresponds to a $\bar{Q}^o(\cdot|x^n)\in{\cal M}_1^{\bf C2}({\cal Y}_{0,n})$ (respectively, $\bar{P}^o(\cdot|y^{n-1})\in{\cal M}_1^{\bf C1}({\cal X}_{0,n})$). Conditions for this to hold are given in Section~\ref{section:applications}.\\
\noi Below, we illustrate analogies and differences between Theorem~\ref{weak_convergence} and currently known results regarding mutual information found in \cite{csiszar74,csiszar92}. To this end, consider {\bf Part B.}, B1). If we use mutual information \cite[Lemma 2]{csiszar92}, then the sequence of joint measures is defined by $P^{\alpha}_{X^n,Y^n}(dx^n,dy^n)\tri{P}_{Y^n|X^n}(dy^n|x^n)\otimes{P}^{\alpha}_{X^n}(dx^n)$, and showing weak convergence of this family is much simpler compared to the sequence of joint distributions $(\overleftarrow{P}_{0,n}^{\alpha}\otimes{\overrightarrow Q}_{0,n})(dx^n,dy^n)$, because ${P}_{X^n}(dx^n)$ is not conditioned on $y^n\in{\cal Y}_{0,n}$. Clearly, if the mapping $x^n\longrightarrow{P}_{Y^n|X^n}(\cdot|x^n)$ is weakly continuous (i.e., special case of \ref{condition2}), and  $P^{\alpha}_{X^n}(dx^n)$ converges weakly to $P^{o}_{X^n}(x^n)$, then $P^{\alpha}_{X^n,Y^n}(dx^n,dy^n)$ converges weakly to $P_{Y^n|X^n}(dy^n|x^n)\otimes{P}^{o}_{X^n}(dx^n)=P^o_{X^n,Y^n}(dx^n,dy^n)$, and so does its marginal on ${\cal Y}_{0,n}$. On the other hand, if we use directed information, then the joint measure $P_{X^n,Y^n}(dx^n,dy^n)\tri\otimes_{i=0}^n{P}_{Y_i|Y^{i-1},X^i}(dy_i|y^{i-1},x^i)\otimes{P}_{X_i|X^{i-1},Y^{i-1}}(dx_i|x^{i-1},y^{i-1})$ involves an $(n+1)$-fold compound probability distribution defined by (\ref{directed:information:section:introduction:equation1}), and ${P}_{X_i|X^{i-1},Y^{i-1}}(\cdot|\cdot,\cdot)$ is a function of $y^{n-1}\in{\cal Y}_{0,n-1}$, hence a significant level of additional complexity incurs, compared to mutual information. Nevertheless, condition {\bf CB} is the natural generalization to causally conditioned $(n+1)$-fold compound probability distributions of the weak continuity of the mapping $x^n\longrightarrow{P}_{Y^n|X^n}(\cdot|x^n)$, assumed for the mutual information by Csisz\'ar in \cite{csiszar92}.
\vspace*{0.2cm}\\
\noi Theorem~\ref{weak_convergence} is important for several extremum problems involving directed information. Such applications are discussed in the next section.

\subsection{Applications of Theorem~\ref{weak_convergence}}\label{section:applications}

\par In this section, we discuss applications of Theorem~\ref{weak_convergence} to the extremum problems of feedback capacity and nonanticipative RDF, defined by (\ref{equation1dd}) and (\ref{introduction:nrdf:equation1}), respectively.

\vspace*{0.2cm}

\noi{\bf Existence of optimal channel input distribution for channels with memory and feedback.} Consider extremum problems of capacity of channels with memory and feedback defined by (\ref{equation1dd}), without any transmission cost constraint. The aim is to show existence  of a channel input conditional distribution $\overleftarrow{P}(\cdot|y^{n-1})\in{\cal M}_1^{\bf C1}({\cal X}_{0,n}), y^{n-1} \in {\cal Y}_{0,n-1}$, which achieves the supremum of directed information. To show that such a conditional distribution exists, it is sufficient to show compactness of the set of channel input conditional distributions  (i.e., this set is closed and uniformly tight) and upper semicontinuity (or continuity) of $\mathbb{I}_{X^n\rightarrow{Y}^n}(\overleftarrow{P}_{0,n},\overrightarrow{Q}_{0,n})$, with respect to $\overleftarrow{P}(\cdot|y^{n-1})\in{\cal M}_1^{\bf C1}({\cal X}_{0,n})$ for a fixed channel $\overrightarrow{Q}_{0,n}(\cdot|x^n)\in{\cal M}_1^{\bf C2}({\cal Y}_{0,n})$. Since Theorem~\ref{continuity}, {\bf Part A.} A2) uniform tightness of $\overleftarrow{P}(\cdot|y^{n-1})\in{\cal M}_1^{\bf C1}({\cal X}_{0,n})$, it remains to show this set is closed. This is shown in the next lemma, by introducing additional assumptions.
\begin{lemma}(Compactness of $\overleftarrow{P}(\cdot|y^{n-1})\in{\cal M}_1^{\bf C1}({\cal X}_{0,n})$)\label{capacity:closedness:lemma}{\ \\}
Suppose the conditions of Theorem~\ref{weak_convergence}, {\bf Part A.} hold, and for each compact subset $K_{0,i-1}\subset{\cal X}_{0,i-1},$ and each $h_i(\cdot)\in{BC}({\cal X}_i)$,
\begin{align}
\lim_{\alpha\longrightarrow\infty}\sup_{x^{i-1}\in{K}_{0,i-1}}\Bigg{|}\int_{{\cal X}_i}h_i(x)p_i^{\alpha}(dx|x^{i-1},y^{i-1})-\int_{{\cal X}_i}h_i(x)p_i(dx|x^{i-1},y^{i-1})\Bigg{|}=0,~i=0,1,\ldots,n\label{capacity:applications:lemma:closedness}
\end{align}
Then, 
\begin{align}
\overleftarrow{P}^{\alpha}_{0,n}(dx^n|y^{n-1})\buildrel w \over\longrightarrow\overleftarrow{P}_{0,n}^o(dx^n|y^{n-1})~\mbox{for each}~y^{n-1}\in{\cal Y}_{0,n-1}\label{application:capacity:closedness:equation00}
\end{align}
i.e., the set $\overleftarrow{P}(\cdot|y^{n-1})\in {\cal M}_1^{\bf C1}({\cal X}_{0,n})$ is closed with respect to the topology of weak convergence, and  moreover, it is also  is compact (i.e., closed and tight).
\end{lemma}
\begin{proof}
See Appendix~\ref{section:proof:closedness:feedback:capacity}.
\end{proof}
\begin{remark}(Compactness of channel input distributions with transmission cost){\ \\}
In the presence of power constraints $\overleftarrow{P}(\cdot|y^{n-1})\in{\cal{P}}_{0,n}(P)\subset{\cal Q}^{\bf C1}({{\cal X}_{0,n}|{\cal Y}_{0,n-1}})$, by Prohorov's theorem (Appendix~\ref{backround_material}, Theorem~\ref{corollary_of_prohorov}), to show compactness of ${\cal{P}}_{0,n}(P)$, it is sufficient to show that this set is closed and uniformly tight. By invoking Lemma~\ref{capacity:closedness:lemma}, it suffices to show ${\cal{P}}_{0,n}(P)$ is a closed subset of the weakly compact set ${\cal M}^{\bf C1}({\cal X}_{0,n})$ (as a closed subset of a weakly compact set is weakly compact).
\end{remark}

\noi {\bf Existence of optimal reproduction distribution of nonanticipative RDF.} Consider a special case of extremum problems of nonanticipative RDF defined by (\ref{introduction:nrdf:equation1}), with distortion constraint defined by (\ref{introduction:nrdf:equation2}), when the source distribution is causally independent of past reproduction symbols, that is, $p_i(dx_i|x^{i-1},y^{i-1})=\mu_i(dx_i|x^{i-1}),-a.a. (x^{i-1},y^{i-1}),~i=0,1,\ldots,n$. Then, the finite time version of (\ref{introduction:nrdf:equation1}) is given by
\begin{align}
R^{na}_{0,n}(D)&=\inf_{\overrightarrow{Q}_{0,n}(dy^n|x^n)\in{\cal Q}_{0,n}(D)}\int_{{\cal X}_{0,n}\times{\cal Y}_{0,n}}\log\Big(\frac{d\overrightarrow{Q}_{0,n}(\cdot|x^n)}{d\nu_{0,n}(\cdot)}(y^n)\Big)\overrightarrow{Q}_{0,n}(dy^n|x^n)\otimes\mu_{0,n}(dx^n)\label{nrdf:application:equation1}\\
&\equiv\inf_{\overrightarrow{Q}_{0,n}(dy^n|x^n)\in{\cal Q}_{0,n}(D)}\mathbb{I}_{X^n\rightarrow{Y^n}}(\mu_{0,n},\overrightarrow{Q}_{0,n})\label{nrdf:application:equation2}
\end{align} 
where $\mu_{0,n}(dx^n)=\otimes_{i=0}^n\mu_i(dx_i|x^{i-1})$, $\nu_{0,n}(dy^n)=\int_{{\cal X}_{0,n}}\overrightarrow{Q}_{0,n}(dy^n|x^n)\otimes\mu_{0,n}(dx^n)$, and the fidelity constraint is defined by
\begin{align}
{\cal Q}_{0,n}(D)\tri&\bigg\{\overrightarrow{Q}_{0,n}(dy^n|x^n)\in{\cal M}_1^{\bf C2}({\cal Y}_{0,n}):\nonumber\\
&\qquad~\frac{1}{n+1}\int_{{\cal X}_{0,n}\times{\cal Y}_{0,n}}d_{0,n}(x^n,y^{n})\overrightarrow{Q}_{0,n}(dy^n|x^n)\otimes\mu_{0,n}(dx^n)\leq{D}\bigg\},~D\geq{0}\label{nrdf:application:equation3}
\end{align}
and $d_{0,n}:{\cal X}_{0,n}\times{\cal Y}_{0,n}\mapsto [0,\infty],~d_{0,n}(x^n,y^{n})\tri\sum_{i=0}^{n}{\rho}_{i}(x^i,y^{i})$ is a measurable function denoting the distortion function of reconstructing $x_i$ by $y_i$, $i=0,1,\ldots,n$.\\
The information nonanticipative RDF defined by (\ref{nrdf:application:equation1}), (\ref{nrdf:application:equation3}), is an equivalent notion to the nonanticipative epsilon entropy  investigated by Gorbunov and Pinsker \cite{gorbunov-pinsker1973} (see Charalambous et al. in \cite{charalambous-stavrou-ahmed2014ieeetac} for relations to filtering theory).\\
\noi The aim is to show existence of a conditional distribution $\overrightarrow{Q}_{0,n}(\cdot|x^{n})\in{\cal M}_1^{\bf C2}({\cal Y}_{0,n})$, which achieves infimum in (\ref{nrdf:application:equation1}). Since ${\cal Q}_{0,n}(D)\subset{\cal M}_1^{\bf C2}({\cal Y}_{0,n})$, to show such a conditional distribution exists, it is sufficient to show compactness of ${\cal M}_1^{\bf C2}({\cal Y}_{0,n})$ (closed and uniformly tight), the set ${\cal Q}_{0,n}(D)$ is a closed subset of ${\cal M}_1^{\bf C2}({\cal Y}_{0,n})$, and $\mathbb{I}_{X^n\rightarrow{Y}^n}(\overleftarrow{P}_{0,n},\overrightarrow{Q}_{0,n})$ is lower semicontinuous with respect to $\overrightarrow{Q}(\cdot|x^{n})\in{\cal M}_1^{\bf C2}({\cal Y}_{0,n})$, for a fixed $\mu_{0,n}(dx^n)\in{\cal M}_1({\cal X}_{0,n})$. This can be done by invoking a combination of the assumptions of Theorem~\ref{weak_convergence} {\bf Part A.} or {\bf Part B.}, depending on whether ${\cal Y}_{0,n}$ is compact and ${\cal X}_{0,n}$ is arbitrary or ${\cal X}_{0,n}$ is compact and ${\cal Y}_{0,n}$ is arbitrary, respectively. Since in general, ${\cal Y}_{0,n}\subseteq{\cal X}_{0,n}$, it is more appropriate to assume ${\cal Y}_{0,n}$ is compact.  
\begin{lemma}(Compactness of $\overrightarrow{Q}(\cdot|x^{n})\in{\cal{Q}}_{0,n}(D)$)\label{closedness:nrdf:lemma}{\ \\}
{\bf(1)} Suppose ${\cal X}_{0,n}$ are Polish spaces, and ${\cal Y}_{0,n}$ is compact, the sequence $\{q_n(\cdot|\cdot,\cdot)\in{\cal Q}({\cal Y}_n|{\cal Y}_{0,n-1}\times{\cal X}_{0,n}):~n\in\mathbb{N}_0\}$ is weakly continuous, i.e., it satisfies (\ref{condition2}), and for each compact subset $\Phi_{0,i-1}\subset{\cal Y}_{0,i-1}$, and each $h_i(\cdot)\in{BC}({\cal Y}_i)$,
\begin{align}
\lim_{\alpha\longrightarrow\infty}\sup_{y^{i-1}\in{\Phi}_{0,i-1}}\Bigg{|}\int_{{\cal Y}_i}h_i(x)q_i^{\alpha}(dy|y^{i-1},x^{i})-\int_{{\cal Y}_i}h_i(y)q_i(dy|y^{i-1},x^{i})\Bigg{|}=0,~\forall{x}^i\in{\cal X}_{0,i},~i=0,1,\ldots,n.\label{nrdf:applications:lemma:closedness}
\end{align}
Then, 
\begin{align*}
\overrightarrow{Q}^{\alpha}_{0,n}(dy^n|x^{n})\buildrel w \over\longrightarrow\overrightarrow{Q}_{0,n}^o(dy^n|x^{n})~\mbox{for each}~x^n\in{\cal X}_{0,n}
\end{align*}
i.e., the set ${\cal M}_1^{\bf C2}({\cal Y}_{0,n})$ is closed with respect to the topology of weak convergence. Moreover,  ${\cal M}_1^{\bf C2}({\cal Y}_{0,n})$ is compact (closed and tight).\\
{\bf(2)} In addition, suppose the distortion function $d_{0,n}(x^n,\cdot):{\cal X}_{0,n}\times{\cal Y}_{0,n}\longmapsto[0,\infty]$ is Borel measurable relative to ${\cal B}({\cal X}_{0,n})\otimes{\cal B}({\cal Y}_{0,n})$ and continuous on $y^n\in{\cal Y}_{0,n}$.\\
Then, the fidelity set ${\cal{Q}}_{0,n}(D)$ is compact (it is a closed subset of the compact set ${\cal M}_1^{\bf C2}({\cal Y}_{0,n})$).
\end{lemma}
\begin{proof}
See Appendix~\ref{proof:closedness:nrdf}.
\end{proof}
Theorem~\ref{weak_convergence} gives the flexibility of choosing either ${\cal X}_{0,n}$ or ${\cal Y}_{0,n}$ to be compact; it has several applications in other extremum problems of directed information. In the following remark, we discuss such applications.
\begin{remark}(Additional Applications)
\begin{itemize}
\item[{\bf(1)}] Consider extremum problems of capacity for a class of channels with memory and feedback, such as, arbitrary varying channels \cite{csiszar92}. Such problems are defined by the max-min operations of directed information, where the minimizer is over the class of channels \cite{lapidoth-narayan1998}. To investigate such capacity problems one has to establish coding theorems, and showing compactness over the class of channel conditional distributions, in addition to channel input distributions is very helpful. Theorem~\ref{weak_convergence}, {\bf Part B.}, B3) gives conditions of weak compactness of channels $\overrightarrow{Q}_{0,n}(\cdot|x^n)\in{\cal M}_{1}^{\bf C2}({\cal Y}_{0,n})$.
\item[{\bf(2)}] Consider extremum problems of sequential or nonanticipative lossy data compression for a class of sources. Then such problems are defined by mini-max operations of directed information, where the maximizer is over the class of source distributions \cite{sakrison1969}. To investigate such data compression problems, one has to establish coding theorems, and to show compactness over the class of source distributions, in addition to the reproduction distributions, Theorem~\ref{weak_convergence}, {\bf Part A.}, A3) is crucial. 
\end{itemize}
\end{remark}

%%%%%%%%%%%%%%%%%%%%%%%%%%%%%%%%%%%%%%%%%%%%%%%%%%%%%%%%%%%%%%%%%%%%%%%%%%%%%%%%%%%%%%%5

\subsection{Lower Semicontinuity of Directed Information}

We are now ready to utilize the results of Theorem~\ref{weak_convergence}, to show lower semicontinuity of directed information $I(X^n\rightarrow{Y}^n)\equiv{\mathbb{I}}_{X^n\rightarrow{Y}^n}({\overleftarrow P}_{0,n},{\overrightarrow Q}_{0,n})$. This may be viewed as a generalization of lower semicontinuity of mutual information $I(X^n;{Y}^n)\equiv\mathbb{I}_{X^n;Y^n}(P_{X^n},Q_{Y^n|X^n})$, with respect to $P_{X^n}$ for fixed $Q_{Y^n|X^n}$, and with respect to $Q_{Y^n|X^n}$ for fixed $P_{X^n}$.
\begin{theorem}(Lower semicontinuity)\label{lower_semicontinuity}{\ \\}
1) Suppose the conditions in Theorem~\ref{weak_convergence}, {\bf Part A.}, hold.\\
For fixed ${\overleftarrow P}_{0,n}(\cdot|y^{n-1})\in{\cal M}_1^{\bf C1}({\cal X}_{0,n})$, if the family ${\cal M}_1^{\bf C2}({\cal Y}_{0,n})$ is closed \big(i.e., $\{\overrightarrow{Q}^{\alpha}_{0,n}(\cdot|x^n):~\alpha=1,2,\ldots\}\in{\cal M}_1^{\bf C2}({\cal Y}_{0,n})$ converges weakly to $\overrightarrow{Q}^o_{0,n}(\cdot|x^n)\in{\cal M}_1^{\bf C2}({\cal Y}_{0,n})$\big) then
\begin{align*}
\mathbb{I}_{X^n\rightarrow{Y}^n}(\overleftarrow{P}_{0,n},\overrightarrow{Q}^o_{0,n})\leq\liminf_{\alpha\longrightarrow\infty}\mathbb{I}_{X^n\rightarrow{Y}^n}(\overleftarrow{P}_{0,n},\overrightarrow{Q}^\alpha_{0,n})
\end{align*}
i.e., ${\mathbb{I}}_{X^n\rightarrow{Y^n}}(\cdot, {\overrightarrow Q}_{0,n})$ is lower semicontinuous on ${\overrightarrow Q}_{0,n}(\cdot|x^n)\in{\cal M}_1^{\bf C2}({\cal Y}_{0,n})$.\\
2) Suppose the conditions in Theorem~\ref{weak_convergence}, {\bf Part B.}, hold.\\ 
For fixed ${\overrightarrow Q}_{0,n}(\cdot|x^n)\in{\cal M}_1^{\bf C2}({\cal Y}_{0,n})$, if the family ${\cal M}_1^{\bf C1}({\cal X}_{0,n})$ is closed \big(i.e.,$\{\overleftarrow{P}^{\alpha}_{0,n}(\cdot|y^{n-1}):~\alpha=1,2,\ldots\}\in{\cal M}_1^{\bf C1}({\cal X}_{0,n})$ converges weakly to $\overleftarrow{P}^o_{0,n}(\cdot|y^{n-1})\in{\cal M}_1^{\bf C1}({\cal X}_{0,n})$\big) then
\begin{align*}
\mathbb{I}_{X^n\rightarrow{Y}^n}(\overleftarrow{P}^o_{0,n},\overrightarrow{Q}_{0,n})\leq\liminf_{\alpha\longrightarrow\infty}\mathbb{I}_{X^n\rightarrow{Y}^n}(\overleftarrow{P}^\alpha_{0,n},\overrightarrow{Q}_{0,n})
\end{align*}
i.e., ${\mathbb{I}}_{X^n\rightarrow{Y^n}}(\overleftarrow{P}_{0,n},\cdot)$ is lower semicontinuous on ${\overleftarrow P}_{0,n}(\cdot|y^{n-1})\in{\cal M}_1^{\bf C1}({\cal X}_{0,n})$.
\end{theorem}
\begin{proof}
See Appendix~\ref{proof:lower:semicontinuity}.
\end{proof}
\vspace*{0.2cm}
\noi Recall that conditions for the sets ${\cal M}_1^{\bf C1}({\cal X}_{0,n})$, ${\cal M}_1^{\bf C2}({\cal Y}_{0,n})$ to be closed are given in Lemma~\ref{capacity:closedness:lemma} and Lemma~\ref{closedness:nrdf:lemma}, respectively.\\
\noi Comparing Theorem~\ref{lower_semicontinuity}, 1), with the lower semicontinuity of mutual information $I(X^n;Y^n)\equiv\mathbb{I}_{X^n;Y^n}$  $(P_{X^n},Q_{Y^n|X^n})$, it is clear that directed information requires additional assumptions for its derivation (e.g., those given in Theorem~\ref{weak_convergence}).\\
\noi Theorem~\ref{weak_convergence} together with Theorem~\ref{lower_semicontinuity} are important to establish existence of the optimal reproduction distribution for the nonanticipative rate distortion functions defined by (\ref{introduction:nrdf:equation1}) \cite{charalambous-stavrou-ahmed2014ieeetac,stavrou-kourtellaris-charalambous2015ieeeit} (by utilizing Weierstrass' Theorem) and in general extremum problems of directed information involving minimization over $\overrightarrow{Q}_{0,n}(\cdot|x^n)$ in some subset of ${\cal M}_1^{\bf C2}({\cal Y}_{0,n})$. This is formally stated in the next theorem.
\begin{theorem}(Existence of information nonanticipative RDF)\label{nrdf:existence}{\ \\}
Under the conditions of Lemma~\ref{closedness:nrdf:lemma} and Theorem~\ref{lower_semicontinuity}, the infimum over ${\overrightarrow{Q}}_{0,n}(\cdot|x^{n})\in{\cal Q}_{0,n}(D)$ in $R^{na}_{0,n}(D)$, defined by (\ref{nrdf:application:equation1}), is achieved by some $\overrightarrow{Q}_{0,n}^*(\cdot|x^{n})\in{\cal Q}_{0,n}(D)$.
\end{theorem}

%%%%%%%%%%%%%%%%%%%%%%%%%%%%%%%%%%%%%%%%%%%%%%%%%%%%%%%%%%%%%%%%%

\subsection{Continuity of Directed Information}

\noi  Many problems in information theory involve extremum problems defined as maximizations of directed information, with respect to the feedback channels $\{p_i(dx_i|x^{i-1},y^{i-1})\in{\cal M}_1({\cal X}_i):~i=0,1,\ldots,n\}$, such as, extremum problems of feedback capacity of channels with memory with transmission cost constraint defined by (\ref{equation1dd}). For such problems it is desirable to have upper semicontinuity of directed information with respect to $\overleftarrow{P}_{0,n}(\cdot|y^{n-1})\in{\cal M}_1^{\bf C1}({\cal X}_{0,n})$. Since by Theorem~\ref{lower_semicontinuity}, directed  information is lower semicontinuous with respect to $\overleftarrow{P}_{0,n}(\cdot|y^{n-1})\in{\cal M}_1^{\bf C1}({\cal X}_{0,n})$, to investigate extremum problems involving feedback capacity (maximization problems), it is sufficient to show continuity of the functional ${\mathbb{I}}_{X^n\rightarrow{Y^n}}({\overleftarrow P}_{0,n}, {\overrightarrow Q}_{0,n})$ with respect to $\overleftarrow{P}_{0,n}(\cdot|y^{n-1})\in{\cal M}_1^{\bf C1}({\cal X}_{0,n})$ for a fixed $\overrightarrow{Q}_{0,n}(\cdot|x^n)\in{\cal M}_1^{\bf C2}({\cal Y}_{0,n})$. Continuity of mutual information based on single letter expression is shown in \cite[Lemma 7]{csiszar92}, and under weaker conditions in \cite[Theorem 3.2]{fozunbal}. Here, we show continuity of directed information by following the procedure in \cite{fozunbal}, generalized to the directed information functional ${\mathbb{I}}_{X^n\rightarrow{Y^n}}({\overleftarrow P}_{0,n}, {\overrightarrow Q}_{0,n})$. First, we shall need the following Lemma.
\begin{lemma}\label{helpful_lemma}{\ \\}
For a given $\overleftarrow{P}_{0,n}(\cdot|\cdot)\in{\cal Q}^{\bf C1}({\cal X}_{0,n}|{\cal Y}_{0,n-1})$ and $\overrightarrow{Q}_{0,n}(\cdot|\cdot)\in{\cal Q}^{\bf C2}({\cal Y}_{0,n}|{\cal X}_{0,n})$ define 
\begin{align}
\big{|}{\mathbb{I}}_{X^n\rightarrow{Y}^n}\big{|}({\overleftarrow P}_{0,n}, {\overrightarrow Q}_{0,n})\tri\int_{{\cal X}_{0,n} \times {\cal Y }_{0,n}}\Bigg{|}\log \Big( \frac{d({\overleftarrow P}_{0,n}\otimes {\overrightarrow Q}_{0,n})}{d ( {\overleftarrow P}_{0,n}\otimes \nu_{0,n} ) }\Big)\Bigg{|} d({\overleftarrow P}_{0,n}\otimes {\overrightarrow Q}_{0,n}).\nonumber
\end{align}
Then the following inequalities hold.
\begin{align}
{\mathbb{I}}_{X^n\rightarrow{Y}^n}({\overleftarrow P}_{0,n}, {\overrightarrow Q}_{0,n})\leq{|\mathbb{I}_{X^n\rightarrow{Y}^n}|}({\overleftarrow P}_{0,n}, {\overrightarrow Q}_{0,n})\leq{\mathbb{I}}_{X^n\rightarrow{Y}^n}({\overleftarrow P}_{0,n}, {\overrightarrow Q}_{0,n})+\frac{2}{e\ln2}.\label{equation66i}
\end{align}
\end{lemma}
\begin{proof}
Recall directed information defined in Remark~\ref{equivalent1}. Then
\begin{align}
{\mathbb{I}}_{X^n\rightarrow{Y}^n}({\overleftarrow P}_{0,n}, {\overrightarrow Q}_{0,n})&=\int_{{\cal X}_{0,n} \times {\cal Y }_{0,n}}\log \Bigg( \frac{d({\overleftarrow P}_{0,n}\otimes {\overrightarrow Q}_{0,n})}{d ( {\overleftarrow P}_{0,n}\otimes \nu_{0,n} ) }\Bigg) d({\overleftarrow P}_{0,n}\otimes {\overrightarrow Q}_{0,n})\nonumber\\
&=\int_{{\cal X}_{0,n} \times {\cal Y }_{0,n}}\log \Bigg( \frac{d({\overleftarrow P}_{0,n}\otimes {\overrightarrow Q}_{0,n})}{d ( {\overleftarrow P}_{0,n}\otimes \nu_{0,n} ) }\Bigg)\Bigg(\frac{d({\overleftarrow P}_{0,n}\otimes {\overrightarrow Q}_{0,n})}{d ( {\overleftarrow P}_{0,n}\otimes \nu_{0,n})}\Bigg)d ( {\overleftarrow P}_{0,n}\otimes \nu_{0,n}).\label{equation110i}
\end{align}
The first inequality in (\ref{equation66i}) is obvious. To show the second inequality in (\ref{equation66i}), recall the inequality \cite[Section 2.3, p. 13]{pinsker-book} $-\frac{1}{e\ln2}\leq{x}\log_2{x}$, $x\in[0,\infty)$ ($0\log0$ is assumed to be $0$). Then,
\begin{align}
|x\log_{2}x|\leq{x\log_{2}x}+\frac{2}{e\ln2}.\label{equation155}
\end{align}
Using \eqref{equation155} in (\ref{equation110i}), with $x=\Big(\frac{d({\overleftarrow P}_{0,n}\otimes {\overrightarrow Q}_{0,n})}{d ( {\overleftarrow P}_{0,n}\otimes \nu_{0,n})}\Big)$, establishes the second inequality in (\ref{equation66i}).
\end{proof}
\vspace*{0.2cm}
\noi Now, we are ready to state the Theorem, which establishes continuity with respect to weak convergence of ${\mathbb{I}}_{X^n\rightarrow{Y^n}}({\overleftarrow P}_{0,n}, {\overrightarrow Q}_{0,n})$ for a fixed $\overrightarrow{Q}_{0,n}(\cdot|x^n)\in{\cal M}_1^{\bf C2}({\cal Y}_{0,n})$, as a functional of $\overleftarrow{P}_{0,n}(\cdot|y^{n-1})\in{\cal M}_1^{\bf C1}({\cal X}_{0,n})$.
\begin{theorem}(Continuity)\label{continuity}{\ \\}
Consider a forward channel ${\overrightarrow Q}_{0,n}(\cdot|x^n)\in{\cal M}_1^{\bf C2}({\cal Y}_{0,n})$, and a closed family of feedback channels $\overleftarrow{P}_{0,n}(\cdot|y^{n-1})\in{\cal M}_1^{{\bf C1},cl}({\cal X}_{0,n})\subseteq{\cal M}_1^{{\bf C1}}({\cal X}_{0,n})$. Suppose the following conditions hold.
\item[A)] There exists a measure $\bar{\nu}_{0,n}(dy^n)$ on ${\cal Y}_{0,n}$ such that ${\overrightarrow Q}_{0,n}(\cdot|x^n)\ll{\bar{\nu}}_{0,n}(dy^n)$ with RND or density $\xi_{\bar{\nu}_{0,n}}(x^n,y^n)\tri\frac{d{\overrightarrow Q}_{0,n}(\cdot|x^n)}{d\bar{\nu}_{0,n}(\cdot)}(y^n)$. 
\item[B)] The RND $\xi_{\bar{\nu}_{0,n}}(x^n,y^n)$ is continuous on ${\cal X}_{0,n}\times{\cal Y}_{0,n},$ and $\xi_{\bar{\nu}_{0,n}}(x^n,y^n)\log\xi_{\bar{\nu}_{0,n}}(x^n,y^n)$ is uniformly integrable over $\Big\{\big(\bar{\nu}_{0,n}\otimes{\overleftarrow P}_{0,n}\big)(dx^n,dy^n):{\overleftarrow P}_{0,n}(\cdot|y^{n-1})\in{\cal M}_1^{{\bf C1},cl}({\cal X}_{0,n})\Big\}.$
\item[C)] For a fixed $y^n\in{\cal Y}_{0,n},$ the RND $\xi_{\bar{\nu}_{0,n}}(x^n,y^n)$ is uniformly integrable over ${\cal M}_1^{{\bf C1},cl}({\cal X}_{0,n})$.\\
Then, ${\mathbb{I}}_{X^n\rightarrow{Y^n}}({\overleftarrow P}_{0,n},{\overrightarrow Q}_{0,n})$ as a functional of ${\overleftarrow P}_{0,n}(\cdot|\cdot)\in{\cal M}_1^{{\bf C1},cl}({\cal X}_{0,n})$ is bounded and weakly continuous over ${\cal M}_1^{{\bf C1},cl}({\cal X}_{0,n})$, for fixed ${\overrightarrow Q}_{0,n}(\cdot|x^n)\in{\cal M}_1^{\bf C2}({\cal Y}_{0,n})$.
\end{theorem}
\begin{proof}
The derivation is shown in Appendix~\ref{continuity1}.
\end{proof}
\vspace*{0.2cm}
\noi Note that Theorem~\ref{weak_convergence} gives conditions for weak compactness of $\overleftarrow{P}_{0,n}(\cdot|y^{n-1})\in{\cal M}_1^{\bf C1}({\cal X}_{0,n})$, and Lemma~\ref{capacity:closedness:lemma} gives conditions for compactness of $\overleftarrow{P}_{0,n}(\cdot|y^{n-1})\in{\cal M}_1^{\bf C1}({\cal X}_{0,n})$. In addition, Theorem~\ref{continuity} gives conditions of weak continuity of $\mathbb{I}_{X^n\rightarrow{Y^n}}(\overleftarrow{P}_{0,n},\overrightarrow{Q}_{0,n})$ with respect to $\overleftarrow{P}_{0,n}(\cdot|y^{n-1})\in{\cal M}_1^{\bf C1}({\cal X}_{0,n})$, for fixed ${\overrightarrow Q}_{0,n}(\cdot|x^n)\in{\cal M}_1^{\bf C2}({\cal Y}_{0,n})$. Hence, sufficient conditions are identified to address existence of solution to the extremum problem of feedback capacity. This is stated in the next theorem.
\begin{theorem}(Existence of information feedback capacity without transmission cost constraint){\ \\}
Under the conditions of Lemma~\ref{capacity:closedness:lemma} and Theorem~\ref{continuity}, the supremum over $\overleftarrow{P}_{0,n}(\cdot|y^{n-1})\in{\cal M}_1^{\bf C1}({\cal X}_{0,n})$ in the extremum problem of information feedback capacity 
\begin{align}
C_{0,n}^{fb}\tri\sup_{\{P_{X_i|X^{i-1},Y^{i-1}}:~i=0,1,\ldots,n\}\in{\cal M}_1^{\bf C1}({\cal X}_{0,n})}\frac{1}{n+1}I(X^n\rightarrow{Y^n})\label{existence:extremum:capacity:feedback:eq.1}
\end{align}
is achieved by some $\overleftarrow{P}^*_{0,n}(\cdot|y^{n-1})\in{\cal M}_1^{\bf C1}({\cal X}_{0,n})$.
\end{theorem}

\subsection{Extension of Directed Information to Arbitrary  Number of Sequences of RV's}

\noi In this section, we demonstrate how the previous results are easily generalized to three, or more, sequences of RV's. These extensions have implications in communication networks, and in communication with side information at either the transmitter or the receiver \cite{kramer2007,gamal-kim2011}.\\
To facilitate the demonstration, first consider the following case.
\vspace*{0.2cm}\\
\noi{\bf Case~1}: The sequence of RV's $X^n\in{\cal X}_{0,n}$ is defined by $X^n=(X^{1,n},X^{2,n})\in{\cal X}^1_{0,n}\times{\cal X}^2_{0,n}\equiv{\cal X}_{0,n}$, where $X^{1,n}=\{X^1_i:~i=0,1,\ldots,n\}$ and $X^{2,n}=\{X^2_i:~i=0,1,\ldots,n\}$.
\vspace*{0.2cm}\\
\noi Then, the two sequences of conditional distributions are $\{p_i(dx^1_i,dx^2_i|x^{1,i-1},x^{2,i-1},y^{i-1}):i=0,1,\ldots,n\}$ and $\{q_i(dy^1_i|y^{1,i-1},x^{1,i},x^{2,i}):i=0,1,\ldots,n\}$, respectively. Consequently, the constructions of consistent families of conditional distributions, and the results obtained so far, extend naturally to directed information $\mathbb{I}_{(X^{1,n},X^{2,n})\rightarrow{Y^n}}(\overleftarrow{P}_{0,n},\overrightarrow{Q}_{0,n})$, where $\overleftarrow{P}_{0,n}(dx^{1,n},dx^{2,n}|y^{n-1})=\otimes_{i=0}^n{p}_i(dx_i^1,dx_i^2|x^{1,i-1},x^{2,i-i},y^{i-1})$, and $\overrightarrow{Q}_{0,n}(dy^n|x^{1,n},x^{2,n})=\otimes_{i=0}^n{q}_i(dy_i|y^{i-1},x^{1,i},x^{2,i})$.\\
Next, we consider the following case.
\vspace*{0.2cm}\\
\noi{\bf Case~2}: The sequence of RV's $Y^n\in{\cal Y}_{0,n}$ is defined by $Y^n\tri(Y^{1,n},Y^{2,n})\in{\cal Y}^1_{0,n}\times{\cal Y}^2_{0,n}\equiv{\cal Y}_{0,n}$, where $Y^{1,n}=\{Y^1_i:~i=0,1,\ldots,n\}$ and $Y^{2,n}=\{Y^2_i:~i=0,1,\ldots,n\}$.
\vspace*{0.2cm}\\
\noi Then, the two sequences of conditional distributions are $\{p_i(dx_i|x^{i-1},y^{1,i-1},y^{2,i-1}):i=0,1,\ldots,n\}$ and $\{q_i(dy^1_i,dy^2_i|y^{1,i-1},y^{2,i-1},x^{i}):i=0,1,\ldots,n\}$, respectively. Consequently, the constructions of consistent families of conditional distributions, and the results obtained so far, extend naturally to directed information $\mathbb{I}_{X^{n}\rightarrow(Y^{1,n},Y^{2,n})}(\overleftarrow{P}_{0,n},\overrightarrow{Q}_{0,n})$, where $\overleftarrow{P}_{0,n}(dx^{n}|y^{1,n-1},y^{2,n-1})=\otimes_{i=0}^n{p}_i(dx_i|x^{i-1},y^{1,i-i},y^{2,i-1})$, and $\overrightarrow{Q}_{0,n}(dy^{1,n},dy^{2,n}|x^{n})=\otimes_{i=0}^n{q}_i(dy_i^1,dy_i^2|y^{1,i-1},y^{2,i-1},x^{i})$.\\
Clearly, {\bf Case 1} and {\bf Case 2} can be generalized to an arbitrary number of sequences of RV's.
%%%%%%%%%%%%%%%%%%%%%%%%%%%%%%%%%%%%%%%%%%%%%%%%%%%%%%%%%%%%%%%%%%%%%%%%%%%%%%%%%%%%%%%%%%

\section{Sequential Variational Equalities of Directed Information}\label{variational}

\par In this section we derive variational equalities including their sequential versions for directed information. Moreover, we illustrate an application of these variational equalities in feedback capacity computation, by developing the main ingredient of a sequential algorithm using dynamic programming.\\
\noi The variational equalities of directed information may be viewed as generalizations of the well-known variational equalities of mutual information $I(X^n;Y^n)\equiv\mathbb{I}_{X^n;Y^n}(P_{X^n},P_{Y^n|X^n})$, expressed as minimizations or maximizations of relative entropy functionals, as follows \cite{blahut1987}.
\vspace*{0.2cm}\\
\noi{\bf Min:} Given a channel distribution $P_{Y^n|X^n}(dy^n|x^n)$, a source distribution $P_{X^n}$, and any arbitrary distribution $V_{Y^n}(dy^n)$ on ${\cal Y}_{0,n}$ then
\begin{align}
\mathbb{I}_{X^n;Y^n}(P_{X^n},P_{Y^n|X^n})=\inf_{V_{Y^n}(dy^n)\in{\cal M}_1({\cal Y}_{0,n})}\int_{{\cal X}_{0,n}\times{\cal Y}_{0,n}}\log\bigg(\frac{dP_{Y^n|X^n}(\cdot|x^n)}{dV_{Y^n}(\cdot)}(y^n)\bigg)P_{Y^n|X^n}(dy^n|x^n)\otimes{P}_{X^n}(dx^n)\label{equation2d}
\end{align}
and the infimum is achieved at $V_{Y^n}(dy^n)\equiv{P}_{Y^n}(dy^n)$ given by
\begin{align}
P_{Y^n}(dy^n)=\int_{{\cal X}_{0,n}}P_{Y^n|X^n}(dy^n|x^n)\otimes{P}_{X^n}(dx^n).\label{section:variational_equalities:eq1}
\end{align}
\vspace*{0.2cm}\\
\noi{\bf Max:} Given a channel distribution $P_{Y^n|X^n}(dy^n|x^n)$, a source distribution $P_{X^n}(dx^n)$, and any arbitrary conditional distribution $V_{X^n|Y^n}(dx^n|y^n)$ on ${\cal X}_{0,n}$ parametrized by $y^n\in{\cal Y}_{0,n}$ then
\begin{align}
\mathbb{I}_{X^n;Y^n}(P_{X^n},P_{Y^n|X^n})=\sup_{\substack{V_{X^n|Y^n}(dx^n|y^n)\\ \in{\cal M}_1({\cal X}_{0,n})}}\int_{{\cal X}_{0,n}\times{\cal Y}_{0,n}}\log\bigg(\frac{dV_{X^n|Y^n}(\cdot|y^n)}{dP_{X^n}(\cdot)}(x^n)\bigg)P_{Y^n|X^n}(dy^n|x^n)\otimes{P}_{X^n}(dx^n)\label{equation2c}
\end{align}
\noi and the supremum is achieved at $V_{X^n|Y^n}(dx^n|y^n)\equiv{P}_{X^n|Y^n}(dx^n|y^n)$ given by
\begin{align}
P_{X^n|Y^n}(dx^n|y^n)=\frac{P_{Y^n|X^n}(dy^n|x^n)\otimes{P}_{X^n}(dx^n)}{\int_{{\cal X}_{0,n}}P_{Y^n|X^n}(dy^n|x^n)\otimes{P}_{X^n}(dx^n)}.\label{section:variational_equalities:eq2}
\end{align}
\vspace*{0.2cm}\\
\noi That is, in \eqref{equation2d} and \eqref{equation2c} the optimal distribution is generated by the joint distribution induced by $\{P_{Y^n|X^n},P_{X^n}\}$. Both variational equalities are used in the Blahut-Arimoto algorithm (BAA) \cite{blahut1987,blahut1972} to derive iterative computational schemes for channel capacity of memoryless channels, via max-max operations, and for RDF of memoryless sources via mini-min operations.\\
\noi Recently, a version of (\ref{equation2c}) is applied in \cite[eq. (9)]{naiss-permuter2013a} to develop a BAA for capacity of channels with memory and feedback, defined on finite alphabet spaces. Specifically, the authors in \cite{naiss-permuter2013a} consider causally conditioned probability mass functions, $P(x^n||y^{n-1})\tri\Pi_{i=0}^n{p}(x_i|x^{i-1},y^{i-1})$, $Q(y^n||x^n)\tri\Pi_{i=0}^n{q}(y_i|y^{i-1},x^{i})$, where $P(y^n)=\Pi_{i=0}^n{p}(y_i|y^{i-1})$ is generated by $P(x^n,y^n)\tri{P}(x^n||y^{n-1})\otimes{Q}(y^n||x^n)=\Pi_{i=0}^n{p}(x_i|x^{i-1},y^{i-1})\otimes{q}(y_i|y^{i-1},x^{i})$, and utilize the identity $P(x^n|y^n)=\Pi_{i=0}^n{p}(x_i|x^{i-1},y^n)$, to rewrite $\frac{Q(y^n||x^n)}{{P}(y^n)}=\frac{P(x^n|y^n)}{{P}(x^n||y^{n-1})}$, and to express (\ref{equation2c}) as follows.
\begin{align}
{I}(X^n\rightarrow{Y}^n)=\sup_{P(x^n|y^{n})}\sum_{(x^n,y^n)\in{\cal X}_{0,n}\times{\cal Y}_{0,n}}\log\bigg(\frac{P(x^n|y^n)}{P(x^n||y^{n-1})}\bigg)P(x^n||y^{n-1})\otimes{Q}(y^n||x^n).\label{section:variational_equalities:eq3}
\end{align}
Based on (\ref{section:variational_equalities:eq3}), the authors in \cite{naiss-permuter2013a} developed an algorithm, which computes the causally conditioned product $P^*(x^n||y^{n-1})$ that maximizes (\ref{section:variational_equalities:eq3}), similar to the BAA \cite{blahut1972,blahut1987}, over the product space ${\cal X}_{0,n}=\times_{i=0}^n{\cal X}_i$. The variational equalities introduced in this paper and the envisioned applications compliment \cite{naiss-permuter2013a}, in the sense that our emphasis is on generalizing classical variational equalities, by developing  sequential variational equalities, which can be used to develop sequential computational algorithms.  

\subsection{Variational Equalities of Directed Information}
\label{VE_DI}
In this section,  our emphasis is to develop variational equalities of directed information, and equivalent sequential variational equalities.\\
\noi The variational equalities of directed information are based on two families of distributions, similar to ${\bf P}(\cdot|\cdot)\in{\cal Q}^{\bf C1}({\cal X}^{\mathbb{N}_0}|{\cal Y}^{\mathbb{N}_0})$ and ${\bf Q}(\cdot|\cdot)\in{\cal Q}^{\bf C2}({\cal Y}^{\mathbb{N}_0}|{\cal X}^{\mathbb{N}_0})$, which are introduced below.\\
\noi Let $P_{0,n}(dx^n,dy^n)={\overleftarrow P}_{0,n}(dx^n|y^{n-1})\otimes{\overrightarrow Q}_{0,n}(dy^n|x^n)$ be the given distribution constructed from the basic feedback channel ${\bf P}(\cdot|{\bf y})\in{\cal M}_1^{\bf C1}({\cal X}^{\mathbb{N}_0})$ and forward channel ${\bf Q}(\cdot|{\bf x})\in{\cal M}_1^{\bf C2}({\cal Y}^{\mathbb{N}_0})$ (by projection onto finite number of coordinates).\\
\noi Let ${\bf S}(\cdot|{\bf x})$ be any probability measure on $({\cal Y}^{\mathbb{N}_0},{\cal B}({\cal Y}^{\mathbb{N}_0}))$ depending parametrically on ${\bf x}\in{\cal X}^{\mathbb{N}_0}$ satisfying the following consistency condition.
\vspace*{0.2cm}\\
\noi{\bf C3:} If $F\in{\cal B}({\cal Y}_{0,n}),$ then ${\bf S}(F|{\bf x})$ is a ${\cal B}({\cal X}_{0,n-1})-$measurable.
\vspace*{0.2cm}\\
\noi For fixed ${\bf x}\in{\cal X}^{\mathbb{N}_0}$, the set of measures on $({\cal Y}^{\mathbb{N}_0},{\cal B}({\cal Y}^{\mathbb{N}_0}))$ satisfying consistency condition {\bf C3} is denoted by ${\cal M}_1^{\bf C3}({\cal Y}^{\mathbb{N}_0})$ and the corresponding family by ${\cal Q}^{\bf C3}({\cal Y}^{\mathbb{N}_0}|{\cal X}^{\mathbb{N}_0})$.  By Remark~\ref{equivalent}, for any family of probability measures ${\bf S}(\cdot|{\bf x})$ on $({\cal Y}^{\mathbb{N}_0},{\cal B}({\cal Y}^{\mathbb{N}_0}))$ parametrized by ${\bf x}\in{\cal X}^{\mathbb{N}_0}$, satisfying consistency condition ${\bf C3}$, there exists a collection of stochastic kernels $\{s_n(\cdot|\cdot,\cdot)\in{\cal Q}({\cal Y}_n|{\cal Y}_{0,n-1}\times{\cal X}_{0,n-1}):n\in\mathbb{N}_0\}$ connected to ${\bf S}(\cdot|{\bf x})$ as follows.\\
\begin{align}
{\bf S}(D|{\bf x})&=\int_{D_0}s_0(dy_0)\int_{D_1}s_1(dy_1|y_0,x_0)\ldots\int_{D_n}s_n(dy_n|y^{n-1},x^{n-1})\equiv\overleftarrow{S}_{0,n}(\times_{i=0}^n{D_i}|x^{n-1})\label{equation110}
\end{align}
where
\begin{align}
D\tri\{{\bf y}\in{\cal Y}^{\mathbb{N}_0}:y_0\in{D}_0,y_1\in{D}_1,\ldots,y_n\in{D}_n\},~D_i\in{\cal B}({\cal Y}_i),~\forall{i}\in\mathbb{N}_0^{n}.\nonumber
\end{align}
Note that $\overleftarrow{S}_{0,n}(\cdot|x^{n-1})\in{\cal M}^{\bf C3}_1({\cal Y}_{0,n})$ is conditioned on $x^{n-1}\in{\cal X}_{0,n-1}$, unlike $\overrightarrow{Q}_{0,n}(\cdot|x^n)\in{\cal M}^{\bf C2}_1({\cal Y}_{0,n})$, which is conditioned on $x^n\in{\cal X}_{0,n}$.\\
\noi Let ${\bf R}(\cdot|{\bf y})$ be any family of probability measures on $({\cal X}^{\mathbb{N}_0},{\cal B}({\cal X}^{\mathbb{N}_0}))$ depending parametrically on ${\bf y}\in{\cal Y}^{\mathbb{N}_0}$ satisfying the following consistency condition.
\vspace*{0.2cm}\\
\noi{\bf C4:} If $E\in{\cal B}({\cal X}_{0,n}),$ then ${\bf R}(E|{\bf y})$ is a ${\cal B}({\cal Y}_{0,n})-$measurable.
\vspace*{0.2cm}\\
\noi For fixed ${\bf y}\in{\cal Y}^{\mathbb{N}_0}$, the set of measures on $({\cal Y}^{\mathbb{N}_0},{\cal B}({\cal Y}^{\mathbb{N}_0}))$ satisfying consistency condition {\bf C4} is denoted by ${\cal M}_1^{\bf C4}({\cal X}^{\mathbb{N}_0})$ and the corresponding family by ${\cal Q}^{\bf C4}({\cal X}^{\mathbb{N}_0}|{\cal Y}^{\mathbb{N}_0})$.  Similarly as before, by Remark~\ref{equivalent}, for any family of measures ${\bf R}(\cdot|{\bf y})$ on $({\cal X}^{\mathbb{N}_0},{\cal B}({\cal X}^{\mathbb{N}_0}))$ parametrized by ${\bf y}\in{\cal Y}^{\mathbb{N}_0}$ satisfying consistency condition {\bf C4}, there exists a collection of stochastic kernels $\{r_n(\cdot|\cdot,\cdot)\in{\cal Q}({\cal X}_n|{\cal X}_{0,n-1}\times{\cal Y}_{0,n}):n\in\mathbb{N}_0\}$ connected to ${\bf R}(\cdot|{\bf y})$ as follows.\\
\begin{align}
{\bf R}(E|{\bf y})&=\int_{E_0}r_0(dx_0|y_0)\int_{E_1}r_1(dx_1|x_0,y^1)\ldots\int_{E_n}r_n(dx_n|x^{n-1},y^{n})\equiv{\overrightarrow{R}}_{0,n}(\times_{i=0}^n{E_i}|y^n)\label{equation111}
\end{align}
where
\begin{align}
E\tri\{{\bf x}\in{\cal X}^{\mathbb{N}_0}:x_0\in{E}_0,x_1\in{E}_1,\ldots,x_n\in{E}_n\},~E_i\in{\cal B}({\cal X}_i),~\forall{i}\in\mathbb{N}_0^n.\nonumber
\end{align}
\noi The joint distribution on $\big{(}{\cal X}^{\mathbb{N}_0}\times{\cal Y}^{\mathbb{N}_0},\otimes_{n\in\mathbb{N}_0}{\cal B}({\cal X}_n)\otimes{\cal B}({\cal Y}_n)\big{)}$ constructed from ${\bf S}(\cdot|\cdot)\in{\cal Q}^{\bf C3}({\cal Y}^{\mathbb{N}_0}|{\cal X}^{\mathbb{N}_0})$ and ${\bf R}(\cdot|\cdot)\in{\cal Q}^{\bf C4}({\cal X}^{\mathbb{N}_0}|{\cal Y}^{\mathbb{N}_0})$, is defined uniquely for $D_i\in{\cal B}({\cal Y}_i),$ $E_i\in{\cal B}({\cal X}_i),$ $\forall{i}\in\mathbb{N}_0^n,$ by
\begin{align}
({\overleftarrow S}_{0,n}\otimes{\overrightarrow R}_{0,n})({\times^n_{i=0}}(D_i\times{E}_i))&=\int_{D_0}s_0(dy_0)\int_{E_0}r_0(dx_0|y_0)\ldots\nonumber\\
&\qquad\ldots\int_{D_n}s_n(dy_n|y^{n-1},x^{n-1})\int_{E_n}r_n(dx_n|x^{n-1},y^{n}).\label{variational:equalities:joint:equation1}
\end{align}
Formally, the $(n+1)$ fold compound joint distribution defined by (\ref{variational:equalities:joint:equation1}) is written as $({\overleftarrow S}_{0,n}\otimes{\overrightarrow R}_{0,n})(dx^n,dy^n)$.\\
Note the difference between the stochastic kernels $\{p_i(dx_i|x^{i-1},y^{i-1}): i\in\mathbb{N}_0\}$, $\{q_i(dy_i|y^{i-1},x^{i}): i\in\mathbb{N}_0\}$, which define $\overleftarrow{P}_{0,n}(dx^n|y^{n-1})$, $\overrightarrow{Q}_{0,n}(dy^n|x^n)$, respectively, as well as the joint measure $(\overleftarrow{P}_{0,n}\otimes\overrightarrow{Q}_{0,n})(dx^n,dy^n)$, and the stochastic kernels $\{r_i(dx_i|x^{i-1},y^i):~i\in\mathbb{N}^n_0\}$, $\{s_i(dy_i|y^{i-1},y^{i-1}):~i\in\mathbb{N}^n_0\}$  which define $\overrightarrow{R}_{0,n}(dx^n|y^n)$, $\overleftarrow{S}_{0,n}(dy^n|x^{n-1})$, respectively, and the joint measure $(\overleftarrow{S}\otimes\overrightarrow{R})(dx^n,dy^n)$. 
\vspace*{0.2cm}\\
\noi The following theorem gives two variational equalities of directed information, including their sequential versions, which are analogous to (\ref{equation2d}), (\ref{equation2c}).
\begin{theorem}(Variational equalities)\label{variational_equalities}{\ \\}
Let $\{{\cal X}_n:~n\in\mathbb{N}_0\}$ and $\{{\cal Y}_n:~n\in\mathbb{N}_0\}$ be Polish spaces. Let ${\bf P}(\cdot|\cdot)\in{\cal Q}^{\bf C1}({\cal X}^{\mathbb{N}_0}|{\cal Y}^{\mathbb{N}_0})$ and ${\bf Q}(\cdot|\cdot)\in{\cal Q}^{\bf C2}({\cal Y}^{\mathbb{N}_0}|{\cal X}^{\mathbb{N}_0})$, and for any $n\in\mathbb{N}_0$, construct from them the joint distribution $P_{0,n}(dx^n,dy^n)=({\overleftarrow{P}}_{0,n}\otimes{\overrightarrow{Q}}_{0,n})(dx^n,dy^n)$, and the distributions ${\nu_{0,n}}(dy^n)=P_{0,n}({\cal X}_{0,n},dy^n)=\otimes_{i=0}^n\nu_{i}(dy_i|y^{i-1})$, $\{\nu_{i}(dy_i|y^{i-1})\in{\cal M}_1({\cal Y}_i):~i=0,1,\ldots,n\}$,  $\overrightarrow{\Pi}(dx^n,dy^n)={\overleftarrow{P}}_{0,n}(dx^n|y^{n-1})\otimes{\nu_{0,n}}(dy^n)$, (defined by  (\ref{equation18}), (\ref{equation20a}), (\ref{equation21})).\\
\noi Then the following variational equalities hold.\\
{\bf Part A.} {\bf (i)} For any arbitrary distribution ${V}_{0,n}(dy^n)\in{\cal M}_1({\cal Y}_{0,n})$ we have
\begin{align}
I(X^n\rightarrow{Y^n})&={\mathbb{I}}_{X^n\rightarrow{Y^n}}({\overleftarrow P}_{0,n}, {\overrightarrow Q}_{0,n})\tri\mathbb{D}(P_{0,n}||{\overrightarrow\Pi}_{0,n})\nonumber\\
&=\inf_{V_{0,n}(dy^n)\in{\cal M}_1({\cal Y}_{0,n})}\mathbb{D}({\overleftarrow P}_{0,n} \otimes {\overrightarrow Q}_{0,n}||{\overleftarrow P}_{0,n} \otimes{V}_{0,n} )\label{equation11}\\
&=\inf_{V_{0,n}(dy^n)\in{\cal M}_1({\cal Y}_{0,n})}\Big\{\int_{{\cal X}_{0,n} \times {\cal Y }_{0,n}}\log \Big( \frac{d{\overrightarrow Q}_{0,n}(\cdot|x^n)}{dV_{0,n}(\cdot)}(y^n)\Big)({\overleftarrow P}_{0,n}\otimes {\overrightarrow Q}_{0,n})(dx^n,dy^n)\Big\}\label{equation10}
\end{align}
and the infimum is achieved at $V_{0,n}(dy^n)\equiv{\nu}_{0,n}(dy^n)\in{\cal M}_1({\cal Y}_{0,n})$ given by
\begin{align}
{\nu}_{0,n}(dy^n)=\int_{{\cal X}_{0,n}}({\overleftarrow P}_{0,n}\otimes {\overrightarrow Q}_{0,n})(dx^n,{dy^n}).\label{variational:equalities:minimization:eq.1}
\end{align}
\noi{\bf (ii)} For any arbitrary conditional distribution $V_{i}(dy_i|y^{i-1})\in{\cal M}_1({\cal Y}_{i}),~i=0,1,\ldots,n$, we have
\begin{align}
I(X^n\rightarrow{Y}^n)&\equiv{\mathbb{I}}_{X^n\rightarrow{Y^n}}({p_i},{q_i}:~i=0,1,\ldots,n)\nonumber\\
&=\inf_{\big\{V_{i}(dy_i|y^{i-1})\in{\cal M}_1({\cal Y}_i):i=0,1,\ldots,n\big\}}\sum^n_{i=0}\int_{ {\cal X}_{0,i}\times {\cal Y}_{0,i-1}}\log\Big(\frac{dq_i(\cdot|y^{i-1},x^i)}{dV_{i}(\cdot|y^{i-1})}(y_i)\Big)\nonumber\\
&\qquad{p}_i(dx_i|x^{i-1},y^{i-1}){\otimes}(\overleftarrow{P}_{0,i-1}\otimes\overrightarrow{Q}_{0,n-1})(dy^{i-1},dx^{i-1})\label{equation15}
\end{align}
and the infimum is achieved at $V_{i}(dy_i|y^{i-1})={\nu}_{i}(dy_i|y^{i-1})$ given by
\begin{align}
{\nu}_{i}(dy_i|y^{i-1})=\int_{{\cal X}_{0,i}}q_i(dy_i|y^{i-1},x^i)\otimes{p}_i(dx_i|x^{i-1},y^{i-1})\otimes(\overleftarrow{P}_{0,i-1}\otimes\overrightarrow{Q}_{0,i-1})(dx^{i-1},dy^{i-1}),~i=0,1,\ldots,n.\label{variational:equalities:sequential:versions:equation1}
\end{align}
{\bf Part B.} {\bf(i)} For any ${\bf S}(\cdot|\cdot)\in{\cal Q}^{\bf C3}({\cal Y}^{\mathbb{N}_0}|{\cal X}^{\mathbb{N}_0})$ and ${\bf R}(\cdot|\cdot)\in{\cal Q}^{\bf C4}({\cal X}^{\mathbb{N}_0}|{\cal Y}^{\mathbb{N}_0})$ then
\begin{align}
&{\mathbb{I}}_{X^n\rightarrow{Y^n}}({\overleftarrow P}_{0,n},{\overrightarrow Q}_{0,n})=\mathbb{D}(P_{0,n}||{\overrightarrow \Pi}_{0,n})\nonumber\\
&=\sup_{\substack{({\overleftarrow S}_{0,n}\otimes{\overrightarrow R}_{0,n})(dx^n,dy^n)\in{\cal M}_1({\cal X}_{0,n}\times{\cal Y}_{0,n}):\\{\overleftarrow S}_{0,n}(dy^n|x^{n-1})\in{\cal M}_1^{\bf C3}({\cal Y}_{0,n}), {\overrightarrow R}_{0,n}(dx^n|y^n)\in{\cal M}_1^{\bf C4}({\cal X}_{0,n})}}\int_{{\cal X}_{0,n} \times {\cal Y }_{0,n}}\log \Big( \frac{d({\overleftarrow S}_{0,n}\otimes {\overrightarrow R}_{0,n})}{d{\overrightarrow \Pi}_{0,n}}(x^n,y^n)\Big)\nonumber\\
&\qquad\qquad\qquad\qquad\qquad({\overleftarrow P}_{0,n}\otimes {\overrightarrow Q}_{0,n})(dx^n,dy^n)\label{equation34}
\end{align}
and the supremum is achieved at $({\overleftarrow S}_{0,n}\otimes {\overrightarrow R}_{0,n})(dx^n,dy^n)=({\overleftarrow P}_{0,n}\otimes {\overrightarrow Q}_{0,n})(dx^n,dy^n)$, given by the RND 
\begin{align}
\Lambda_{0,n}(x^n,y^n)\tri\frac{d({\overleftarrow P}_{0,n}\otimes {\overrightarrow Q}_{0,n})}{d({\overleftarrow S}_{0,n}\otimes {\overrightarrow R}_{0,n})}(x^n,y^n)=1-a.s.,~n\in\mathbb{N}_0.\label{equation105}
\end{align}
Equivalently,
\begin{align}
\lambda_i(x^i,y^i)\tri\frac{dp_i(\cdot|x^{i-1},y^{i-1})}{d{r}_i(\cdot|x^{i-1},y^i)}(x_i).\frac{d{q}_i(\cdot|y^{i-1},x^{i})}{ds_i(\cdot|y^{i-1},x^{i-1})}(y_i)=1-a.s.,~i=0,1,\ldots,n.\label{equation102}
\end{align}
Moreover, if $q_i(\cdot|y^{i-1},x^{i})\ll{s}_i(\cdot|y^{i-1},x^{i-1})$-a.a. $(y^{i-1},x^i)$ and $p_i(\cdot|x^{i-1},y^{i-1})\ll{r}_i(\cdot|x^{i-1},y^i)$-a.a.~$(x^{i-1},y^i)$, $i=0,1,\ldots,n$, then 
\begin{align}
\Pi^n_{i=0}\frac{d{q}_i(\cdot|y^{i-1},x^i)}{ds_i(\cdot|y^{i-1},x^{i-1})}(y_i)=\Pi^n_{i=0}
\bigg(\frac{d{p}_i(\cdot|x^{i-1},y^{i-1})}{dr_i(\cdot|x^{i-1},y^{i})}(x_i)\bigg)^{-1}-a.s.,~{n}\in\mathbb{N}_0\label{equation103}
\end{align}
or equivalently,
\begin{align}
\frac{d{q}_i(\cdot|y^{i-1},x^i)}{ds_i(\cdot|y^{i-1},x^{i-1})}(y_i)=\bigg(\frac{d{p}_i(\cdot|x^{i-1},y^{i-1})}
{dr_i(\cdot|x^{i-1},y^{i})}(x_i)\bigg)^{-1}-a.s.,~i=0,1,\ldots,n.\label{equation104}
\end{align}
\noi{\bf(ii)} For any arbitrary collection of stochastic kernels $\{r_{i}(\cdot|\cdot,\cdot)\in{\cal Q}({\cal X}_{i}|{\cal X}_{0,i-1}\times{\cal Y}_{0,i-1}),~i=0,1,\ldots,n\}$, and $\{s_{i}(\cdot|\cdot,\cdot)\in{\cal Q}({\cal Y}_{i}|{\cal Y}_{0,i-1}\times{\cal X}_{0,i-1}),~i=0,1,\ldots,n\}$, define
\begin{align*}
\mathbb{I}(p_i,q_i,s_i,{r}_i:~i=0,1,\ldots,n)&\tri\sum^n_{i=0}\int_{ {\cal X}_{0,i}\times {\cal Y}_{0,i}}\log\Bigg(\frac{d{r}_i(\cdot|x^{i-1},y^{i})}{dp_i(\cdot|x^{i-1},y^{i-1})}(x_i).\frac{ds_i(\cdot|y^{i-1},x^{i-1})}{d\nu_{i}(\cdot|y^{i-1})}(y_i)\Bigg)\\
&\qquad{\otimes}_{k=0}^i\big({p}_k(dx_k|x^{k-1},y^{k-1}){\otimes}q_k(dy_k|y^{k-1},x^k)\big).
\end{align*}
Then 
\begin{align}
I(X^n\rightarrow{Y}^n)&\equiv{\mathbb{I}}_{X^n\rightarrow{Y^n}}({p_i}(\cdot|\cdot,\cdot),{q_i}(\cdot|\cdot,\cdot):~i=0,1,\ldots,n)\nonumber\\
&=\sup_{\substack{\big\{s_i(dy_i|y^{i-1},x^{i-1})\otimes{r}_i(dx_i|x^{i-1},y^{i})\in{\cal M}({\cal X}_i\times{\cal Y}_i)\big\},~i=0,1,\ldots,n\\\big\{s_i(dy_i|y^{i-1},x^{i-1})\in{\cal M}_1({\cal Y}_i),~r_i(dx_i|x^{i-1},y^{i})\in{\cal M}_1({\cal X}_i)\big\}}}\mathbb{I}(p_i,q_i,s_i,{r}_i:~i=0,1,\ldots,n)\label{equation15a}
\end{align}
and the supremum is achieved when (\ref{equation102}) or (\ref{equation104}) hold.
\end{theorem}
\begin{proof}
{\bf Part A.} {\bf(i)} From Theorem~\ref{equivalent1}, then
\begin{align}
&\mathbb{D}({\overleftarrow P}_{0,n}\otimes{\overrightarrow Q}_{0,n}||{\overleftarrow P}_{0,n}\otimes{V_{0,n}})=\int_{{\cal X}_{0,n} \times {\cal Y }_{0,n}}\log \Big( \frac{d{\overrightarrow Q}_{0,n}(\cdot|x^n)}{dV_{0,n}(\cdot)}(y^n)\Big)({\overleftarrow P}_{0,n}\otimes {\overrightarrow Q}_{0,n})(dx^n,dy^n)\label{equation8}\\
&=\int_{{\cal X}_{0,n} \times {\cal Y }_{0,n}}\log \Big( \frac{d{\overrightarrow Q}_{0,n}(\cdot|x^n)}{d{\nu}_{0,n}(\cdot)}(y^n)\Big)({\overleftarrow P}_{0,n}\otimes {\overrightarrow Q}_{0,n})(dx^n,dy^n)+\mathbb{D}(\nu_{0,n}|| V_{0,n})\label{equation9}\\
&\geq\mathbb{D}({\overleftarrow P}_{0,n}\otimes{\overrightarrow Q}_{0,n}||{\overleftarrow P}_{0,n}\otimes{\nu_{0,n}}).\label{equation9i}
\end{align}
Moreover, equality holds in (\ref{equation9i}) when $V_{0,n}=\nu_{0,n}$ given by (\ref{variational:equalities:minimization:eq.1}). Hence, $\mathbb{D}(P_{0,n}||{\overrightarrow\Pi}_{0,n})$ in (\ref{equation33}) can be expressed via variational equality (\ref{equation10}).\\
\noi{\bf(ii)} The derivation of (\ref{equation15}) is similar to (\ref{equation11}), (\ref{equation10}), but it is done with respect to each component $V_{i}(dy_i|y^{i-1})\in{\cal M}_1({\cal Y}_i)$, starting at $i=n$ and moving sequentially backward to $i=0$.\\
{\bf Part B.} {\bf(i)} Consider the difference between $I(X^n\rightarrow{Y}^n)=\mathbb{D}({\overleftarrow P}_{0,n}\otimes{\overrightarrow Q}_{0,n}||{\overrightarrow\Pi}_{0,n})$ given by (\ref{equation33}) and the LHS of (\ref{equation34}) (without the supremum). Then
\begin{align}
&{\mathbb{I}}_{X^n\rightarrow{Y^n}}({\overleftarrow P}_{0,n}, {\overrightarrow Q}_{0,n}) -\int_{{\cal X}_{0,n} \times {\cal Y }_{0,n}}\log \Big( \frac{d({\overleftarrow S}_{0,n}\otimes {\overrightarrow R}_{0,n})}{d({\overleftarrow P}_{0,n}\otimes \nu_{0,n})}(x^n,y^n)\Big)({\overleftarrow P}_{0,n}\otimes {\overrightarrow Q}_{0,n})(dx^n,dy^n)\nonumber\\
&= \int_{{\cal X}_{0,n} \times {\cal Y }_{0,n}}\log \Big( \frac{d({\overleftarrow P}_{0,n}\otimes {\overrightarrow Q}_{0,n})}{d ( {\overleftarrow S}_{0,n}\otimes {\overrightarrow R}_{0,n} ) }(x^n,y^n)\Big) ({\overleftarrow P}_{0,n}\otimes {\overrightarrow Q}_{0,n})(dx^n,dy^n)\nonumber\\
&\geq\int_{{\cal X}_{0,n} \times {\cal Y }_{0,n}}\Big(1-\frac{d({\overleftarrow S}_{0,n}\otimes {\overrightarrow R}_{0,n})}{d({\overleftarrow P}_{0,n}\otimes {\overrightarrow Q}_{0,n})}(x^n,y^n)\Big)({\overleftarrow P}_{0,n}\otimes {\overrightarrow Q}_{0,n})(dx^n,dy^n)=0\label{equation25}
\end{align}
where (\ref{equation25}) follows from the inequality $\log{x}\geq{1-\frac{1}{x}},$ $x>0,$ which holds with equality if and only if $x=1.$ Furthermore, equality holds  in (\ref{equation25}), when the RND $\Lambda_{0,n}(x^n,y^n)\tri\frac{d({\overleftarrow P}_{0,n}\otimes {\overrightarrow Q}_{0,n})}{d({\overleftarrow S}_{0,n}\otimes {\overrightarrow R}_{0,n})}(x^n,y^n)=1,$ ${\overleftarrow S}_{0,n}\otimes {\overrightarrow R}_{0,n}-a.s.$ in $(x^n,y^n).$ Since $({\overleftarrow P}_{0,n}\otimes {\overrightarrow Q}_{0,n})({\cal X}_{0,n}\times{\cal Y}_{0,n})=({\overleftarrow S}_{0,n}\otimes {\overrightarrow R}_{0,n})({\cal X}_{0,n}\times{\cal Y}_{0,n})=1,$ this condition is equivalent to ${\overleftarrow P}_{0,n}\otimes {\overrightarrow Q}_{0,n}={\overleftarrow S}_{0,n}\otimes {\overrightarrow R}_{0,n}$. By conditioning (\ref{equation105}) on ${\cal B}({\cal X}_{0,n-1})\otimes{\cal B}({\cal Y}_{0,n-1})$ one obtains (\ref{equation102}). Furthermore, (\ref{equation103}) is obtained from (\ref{equation105}), while (\ref{equation104}) is obtained by conditioning.\\
\noi{\bf(ii)} The derivation of (\ref{equation15a}) is similar to (\ref{equation34}) but it is done with respect to each component $s_i\otimes{r}_i$,~starting at $i=n$ and moving backward sequentially to $i=0$.
\end{proof}
\vspace*{0.2cm}
\noi Note that Theorem~\ref{variational_equalities}, {\bf Part A. (ii)}, {\bf Part B. (ii)} are sequential versions of {\bf Part A. (i)}, {\bf Part B. (i)}, respectively. \\
\noi Next, we discuss the relation between the variational equality of directed information (\ref{equation34}) and the variational equality of mutual information (\ref{equation2c}). Clearly, (\ref{equation2c}) is also equivalent to 
\begin{align}
\sup_{V_{X^n|Y^n}\otimes{P}_{Y^n}}\int_{{\cal X}_{0,n}\times{\cal Y}_{0,n}}\log\bigg(\frac{d\big(V_{X^n|Y^n}(\cdot|y^n)\otimes{P}_{Y^n}(\cdot)\big)}{d\big(P_{X^n}(\cdot)\times{P}_{Y^n}(\cdot)\big)}(x^n,y^n)\bigg)P_{Y^n|X^n}(dy^n|x^n)\otimes{P}_{X^n}(dx^n) \label{equation112}
\end{align}
since the RND in (\ref{equation112}) is another version of the one in (\ref{equation2c}). Hence, (\ref{equation34}) is the analogue of (\ref{equation112}). Further, to obtain the analogue of the maximizing measure in (\ref{equation2c}), given by (\ref{section:variational_equalities:eq2}), suppose $q_i(\cdot|y^{i-1},x^i)\ll{s}_i(\cdot|y^{i-1},x^{i-1})-a.a. (x^i,y^{i-1})$, $i=0,1,\ldots,n$, and $\{{s}_i(\cdot|y^{i-1},x^{i-1}):~i=0,1,\ldots,n\}$ is fixed, and generated by $\overleftarrow{P}_{0,n}(\cdot|y^{n-1})\in{\cal M}^{\bf C1}({\cal X}_{0,n})$ and $\overrightarrow{Q}_{0,n}(\cdot|x^{n})\in{\cal M}^{\bf C2}({\cal Y}_{0,n})$. Then from (\ref{equation102}) we obtain 
\begin{align}
r_i(dx_i|x^{i-1},y^i)&=\bigg(\frac{dq_i(\cdot|y^{i-1},x^i)}{ds_i(\cdot|y^{i-1},x^{i-1})}(y_i)\bigg){p}_i(dx_i|x^{i-1},y^{i-1}),\hst i=0,1,\ldots,n\label{variationa:equalities:sequential:version:equation2}\\
&=\frac{q_i(dy_i|y^{i-1},x^i)}{\int_{{\cal X}_i}q_i(dy_i|y^{i-1},x^i)\otimes{p}_i(dx_i|x^{i-1},y^{i-1})}{p}_i(dx_i|x^{i-1},y^{i-1}),\hst~i=0,1,\ldots,n.\label{variationa:equalities:sequential:version:equation3}
\end{align}
\noi Obviously, for a fixed  $\{{s}_i(\cdot|y^{i-1},x^{i-1}):~i=0,1,\ldots,n\}$, (\ref{variationa:equalities:sequential:version:equation2}), (\ref{variationa:equalities:sequential:version:equation3}) are the sequential versions of maximizing distribution satisfying \eqref{equation105}, given by
\begin{align}
\overrightarrow{R}_{0,n}(dx^n|y^n)=\otimes_{i=0}^n\frac{q_i(dy_i|y^{i-1},x^i)}{\int_{{\cal X}_i}q_i(dy_i|y^{i-1},x^i)\otimes{p}_i(dx_i|x^{i-1},y^{i-1})}{p}_i(dx_i|x^{i-1},y^{i-1}),\hst n\in\mathbb{N}_0.\label{variationa:equalities:sequential:version:equation4}
\end{align}
Clearly, (\ref{variationa:equalities:sequential:version:equation4}) is the analogue of the maximizing distribution $P_{X^n|Y^n}$ in (\ref{equation2c}).\\
\noi Note that the optimization in (\ref{equation34}) can be done by keeping $\overleftarrow{S}_{0,n}(\cdot|x^{n-1})$ fixed, generated by ${\bf P}(\cdot|\cdot)\in{\cal Q}^{\bf C1}({\cal X}^{\mathbb{N}_0}|{\cal Y}^{\mathbb{N}_0})$ and ${\bf Q}(\cdot|\cdot)\in{\cal Q}^{\bf C2}({\cal Y}^{\mathbb{N}_0}|{\cal X}^{\mathbb{N}_0})$, and maximizing only over $\overrightarrow{R}_{0,n}(\cdot|y^n)\in{\cal M}_1({\cal X}_{0,n})$ as demonstrated above.
\vspace*{0.2cm}\\

\noi For extremum problems of directed information, such as,  the channel capacity with memory with and without feedback, it is desirable to invoke a sequential version of variational equalities, in order to derive sequential algorithms. This point is illustrated in the next section.

\subsection{Applications of Sequential Variational Equalities to  Feedback Capacity Computations}\label{section:variational:equality:sequential:applications}

\par Consider the extremum problem of feedback capacity  given by (\ref{equation1dd}), without transmission cost constraint. Expressed in terms of channel distributions $\{q_i(dy_i|y^{i-1},x^i)\in{\cal M}_1({\cal Y}_i):~i=0,1,\ldots,n\}$ and the channel input distributions $\{p_i(dx_i|x^{i-1},y^{i-1})\in{\cal M}_1({\cal X}_i):~i=0,1,\ldots,n\}$, then $C^{fb}\tri\liminf_{n\longrightarrow\infty}\frac{1}{n+1}C^{fb}_{0,n}$ where 
\begin{align}
C^{fb}_{0,n}\tri\sup_{\big\{p_i(dx_i|x^{i-1},y^{i-1})\in{\cal M}_1({\cal X}_i):~i=0,1,\ldots,n\big\}}\sum_{i=0}^n{I}(X^i;Y_i|Y^{i-1}).\label{variational:equality:sequential:applications:equation2}
\end{align}
Given a specific channel, Theorem~\ref{variational_equalities}, {\bf Part B. (ii)} can be used to develop a sequential alternating double maximization algorithm over appropriate sets of distributions, which computes $C^{fb}$ via (\ref{variational:equality:sequential:applications:equation2}) (i.e., $\frac{C_{0,n}^{fb}}{n+1}$), for large enough $n$, starting at $n$ and moving sequentially in time to $n-1,n-2,\ldots,0$. This is illustrated next, by considering a simple example.
\vspace*{0.2cm}\\
\noi{\bf Unit Memory Channel.} Consider a channel defined by $\{q_i(dy_i|y_{i-1},x_i)\in{\cal M}_1({\cal Y}_i):~i=0,1,\ldots,n\}$, called Unit Memory Channel Output (UMCO). Then, (\ref{variational:equality:sequential:applications:equation2}) reduces to 
\begin{align}
C^{fb,UMCO}_{0,n}\tri\sup_{\big\{p_i(dx_i|x^{i-1},y^{i-1})\in{\cal M}_1({\cal X}_i):~i=0,1,\ldots,n\big\}}\sum_{i=0}^n\mathbb{E}\bigg\{\log\Big(\frac{dq_i(\cdot|Y_{i-1},X_i)}{d\nu_{i}(\cdot|Y^{i-1})}(Y_i)\Big)\bigg\}.\label{variational:equality:sequential:applications:equation3}
\end{align}
It is conjectured by Chen and Berger \cite{chen-berger2005} (see also \cite{kourtellaris-charalambous2015itw,kourtellaris-charalambous-boutros2015isit}) that the optimal channel input distribution in (\ref{variational:equality:sequential:applications:equation3}) satisfies the conditional independence $p_i(dx_i|x^{i-1},y^{i-1})=\pi_i(dx_i|y_{i-1})-a.a.~(x^{i-1},y^{i-1})\in{\cal X}_{0,n-1}\times{\cal Y}_{0,n-1}$, which then implies the corresponding joint process $\{(X_i,Y_i):~i=0,1,\ldots,n\}$ is first order Markov, the output process $\{Y_i:~i=0,1,\ldots,n\}$ is first order Markov, and consequently, (\ref{variational:equality:sequential:applications:equation3}) reduces to the following expression\footnote{superscript $\pi$ on various distributions indicates their dependence on $\{\pi_i(dx_i|y_{i-1}):~i=0,1,\ldots,n\}$.}. 
\begin{align}
C^{fb,UMCO}_{0,n}&\tri\sup_{\big\{\pi_i(dx_i|y_{i-1})\in{\cal M}_1({\cal X}_i):~i=0,1,\ldots,n\big\}}\sum_{i=0}^n\int_{{\cal Y}_{i-1,i}\times{\cal X}_{i}}\log\Big(\frac{dq_i(\cdot|y_{i-1},x_i)}{d\nu^{\pi}_{i}(\cdot|y_{i-1})}(y_i)\Big)\nonumber\\
&\qquad\qquad\qquad{q}_i(dy_i|y_{i-1},x_i)\otimes\pi_i(dx_i|y_{i-1})\otimes\nu^{\pi}_{i}(dy_{i-1})\label{variational:equality:sequential:applications:equation4}\\
&=\sup_{\big\{\pi_i(dx_i|y_{i-1})\in{\cal M}_1({\cal X}_i):~i=0,1,\ldots,n\big\}}\sum_{i=0}^n{I}(X_i;Y_i|Y_{i-1})\label{variational:equality:sequential:applications:equation5}
\end{align} 
where
\begin{align}
\nu^{\pi}_{i}(\cdot|y_{i-1})=\int_{{\cal X}_{i}}{q}_i(dy_i|y_{i-1},x_i)\otimes\pi_i(dx_i|y_{i-1}),~i=0,1,\ldots,n.\label{variational:equality:sequential:applications:equation6}
\end{align}
The conjecture by Chen and Berger \cite{chen-berger2005} (i.e., \eqref{variational:equality:sequential:applications:equation4}-\eqref{variational:equality:sequential:applications:equation6}) is recently shown in \cite{kourtellaris-charalambous2015aieeeit}, by invoking the variational equality \eqref{equation15} in extremum problems of feedback capacity, to identify information structures of the optimal channel input distribution for general channels with finite memory.\\
\noi By Theorem~\ref{variational_equalities}, {\bf Part B. (ii)}, for  a fixed $\{\pi_i(dx_i|y_{i-1})\in{\cal M}_1({\cal X}_i):~i=0,1,\ldots,n\}$, the expression inside the maximization in (\ref{variational:equality:sequential:applications:equation4}) or (\ref{variational:equality:sequential:applications:equation5}) is expressed as
\begin{align}
\sum_{i=0}^n{I}(X_i;Y_i|Y_{i-1})=&\sup_{\big\{r_i(dx_i|y_{i-1},y_i)\in{\cal M}_1({\cal X}_i):~i=0,1,\ldots,n\big\}}\sum_{i=0}^n\int_{{\cal Y}_{i-1,i}\times{\cal X}_{i}}\log\Big(\frac{dr_i(\cdot|y_{i-1},y_i)}{d{\pi}_{i}(\cdot|y_{i-1})}(x_i)\Big)\nonumber\\
&\qquad\qquad{q}_i(dy_i|y_{i-1},x_i)\otimes\pi_i(dx_i|y_{i-1})\otimes\nu^{\pi}_{i}(dy_{i-1})
\label{variational:equality:sequential:applications:equation7}
\end{align}
where the supremum in (\ref{variational:equality:sequential:applications:equation7}) is achieved at
\begin{align}
r^{\pi}_i(dx_i|y_{i-1},y_i)=\Big(\frac{dq_i(\cdot|y_{i-1},x_i)}{d\nu^{\pi}_{i}(\cdot|y_{i-1})}(y_i)\Big){\pi}_{i}(dx_i|y_{i-1}),~i=0,1,\ldots,n.\label{variational:equality:sequential:applications:equation8}
\end{align}
Next, we convert $C^{fb,UMCO}_{0,n}$ into a sequential alternating maximization problem over appropriate sets of distributions, by using dynamic programming.\\
\noi Let $C_t: {\cal Y}_{t-1}\longmapsto[0,\infty)$ represent the maximum expected total pay-off in (\ref{variational:equality:sequential:applications:equation4}) on the future time horizon $\{t,t+1,\ldots,n\}$, given $Y_{t-1}=y_{t-1}$ at time $t-1$, defined by 
\begin{align}
C_t(y_{t-1})=&\sup_{\big\{\pi_i(dx_i|y_{i-1})\in{\cal M}_1({\cal X}_i):~i=t,t+1,\ldots,n\big\}}\mathbb{E}^{\pi}\bigg\{\sum_{i=t}^n\log\Big(\frac{dq_i(\cdot|y_{i-1},x_i)}{d{\nu}^{\pi}_{i}(\cdot|y_{i-1})}(y_i)\Big)q_i(dy_i|y_{i-1},x_i)\nonumber\\
&\qquad\qquad\otimes{\pi}_i(dx_i|y_{i-1})\Big{|}Y_{t-1}=y_{t-1}\bigg\}.\label{variational:equality:sequential:applications:equation9}
\end{align}
\noi By standard arguments (see \cite{bertsekas-shreve2007}), and in view of the Markov property of $\{Y_i:~i=0,1,\ldots,n\}$, it follows that (\ref{variational:equality:sequential:applications:equation9}) satisfies the following dynamic programming recursions.
\begin{align}
C_n(y_{n-1})&=\sup_{\pi_n(dx_n|y_{n-1})\in{\cal M}_1({\cal X}_n)}\int_{{\cal X}_n\times{\cal Y}_n}\log\Big(\frac{dq_n(\cdot|y_{n-1},x_n)}{d{\nu}^{\pi}_{n}(\cdot|y_{n-1})}(y_n)\Big)q_n(dy_n|y_{n-1},x_n)\otimes{\pi}_n(dx_n|y_{n-1})\label{variational:equality:sequential:applications:equation10}\\
C_t(y_{t-1})&=\sup_{\pi_t(dx_t|y_{t-1})\in{\cal M}_1({\cal X}_t)}\Big\{\int_{{\cal X}_t\times{\cal Y}_t}\log\Big(\frac{dq_t(\cdot|y_{t-1},x_t)}{d{\nu}^{\pi}_{t}(\cdot|y_{t-1})}(y_t)\Big)q_t(dy_t|y_{t-1},x_t)\otimes{\pi}_t(dx_t|y_{t-1})\nonumber\\
&+\int_{{\cal X}_t\times{\cal Y}_t}C_{t+1}(y_t)q_t(dy_t|y_{t-1},x_t)\otimes{\pi}_t(dx_t|y_{t-1})\Big\},~t=0,1,\ldots,n-1.\label{variational:equality:sequential:applications:equation11}
\end{align}
\noi It is well-known that the computation of the optimal channel input distribution in (\ref{variational:equality:sequential:applications:equation10}), (\ref{variational:equality:sequential:applications:equation11}) suffers from the so-called, curse of dimensionality (i.e., it is often computationally prohibitive, even for finite alphabet spaces). However, by applying Theorem~\ref{variational_equalities}, {\bf Part B. (ii)}, to the dynamic programming recursions (\ref{variational:equality:sequential:applications:equation10}), (\ref{variational:equality:sequential:applications:equation11}), we can show that these can be converted to equivalent alternating maximizations over convex sets. Consequently, \eqref{variational:equality:sequential:applications:equation4} can be expressed via sequential alternating maximizations, of concave functionals over convex sets, as stated in the next theorem. 
\begin{theorem}(Sequential double maximization of feedback capacity of UMCO)\label{applications:sequential:double maximization:theorem}{\ \\}
Consider the UMCO defined by $\{q_i(dy_i|y_{i-1},x_i)\in{\cal M}_1({\cal Y}_i):~i=0,1,\ldots,n\}$, and $C^{fb,UMCO}_{0,n}$ defined by (\ref{variational:equality:sequential:applications:equation4}), for a fixed $Prob\{Y_{-1}\in{d}y_{-1}\}\tri{\nu}_{{-1}}(dy_{i-1})$.\\
\noi{\bf Part A.} The dynamic programming recursions (\ref{variational:equality:sequential:applications:equation10}), (\ref{variational:equality:sequential:applications:equation11}) are equivalent to the following sequential double maximization dynamic programming recursions.
\begin{align}
C_n(y_{n-1})&=\sup_{\pi_n(dx_n|y_{n-1})\in{\cal M}_1({\cal X}_n)}\sup_{r_n(dx_n|y_{n-1},y_n)\in{\cal M}_1({\cal X}_n)}\Big\{\int_{{\cal X}_{n}\times{\cal Y}_{n}}\log\Big(\frac{dr_n(\cdot|y_{n-1},y_n)}{d{\pi}_{n}(\cdot|y_{n-1})}(x_n)\Big)\nonumber\\
&\qquad\qquad{q}_n(dy_n|y_{n-1},x_n)\otimes\pi_n(dx_n|y_{n-1})\Big\}\label{variational:equality:sequential:applications:equation12}\\
C_t(y_{t-1})&=\sup_{\pi_t(dx_t|y_{t-1})\in{\cal M}_1({\cal X}_t)}\sup_{r_t(dx_t|y_{t-1},y_t)\in{\cal M}_1({\cal X}_t)}\Big\{\int_{{\cal X}_t\times{\cal Y}_t}\log\Big(\frac{dr_t(\cdot|y_{t-1},y_t)}{d{\pi}_{t}(\cdot|y_{t-1})}(x_t)\Big)q_t(dy_t|y_{t-1},x_t)\nonumber \\
&\otimes{\pi}_t(dx_t|y_{t-1})
+\int_{{\cal X}_t\times{\cal Y}_t}C_{t+1}(y_t)q_t(dy_t|y_{t-1},x_t)\otimes{\pi}_t(dx_t|y_{t-1})\Big\},~t=0,1,\ldots,n-1.\label{variational:equality:sequential:applications:equation13}
\end{align}
and $C^{fb, UMCO}_{0,n}$ is given by
\begin{align*}
C^{fb, UMCO}_{0,n}=\int_{{\cal Y}_{-1}}C_0(y_{-1})\nu_{-1}(dy_{-1}).
\end{align*}
Moreover, the following hold.\\
\noi{\bf Maximizations in (\ref{variational:equality:sequential:applications:equation12}).}\\ \noi{\bf(i)} For a fixed $\pi_n(dx_n|y_{n-1})\in{\cal M}_1({\cal X}_n)$, the maximum in (\ref{variational:equality:sequential:applications:equation12}) over $r_n(dx_n|y_{n-1},y_n)\in{\cal M}_1({\cal X}_n)$ occurs at $r_n(\cdot|\cdot,\cdot)=r^{*,\pi}_n(\cdot|\cdot,\cdot)$ given by
\begin{align}
r^{*,\pi}_n(dx_n|y_{n-1},y_n)=\Big(\frac{dq_n(\cdot|y_{n-1},x_n)}{d\nu^{\pi}_{n}(\cdot|y_{n-1})}(y_n)\Big){\pi}_{n}(dx_n|y_{n-1}).
\label{variational:equality:sequential:applications:equation14}
\end{align}
\noi{\bf(ii)} For a fixed $r_n(dx_n|y_{n-1},y_n)\in{\cal M}_1({\cal X}_n)$, the maximum in (\ref{variational:equality:sequential:applications:equation12}) over $\pi_n(dx_n|y_{n-1})\in{\cal M}_1({\cal X}_n)$ occurs at $\pi_n(\cdot|\cdot)=\pi^{*,r}_n(\cdot|\cdot)$\footnote{superscript $r$ indicates the dependence on the distribution $\{r_i(dx_i|y_{i-1},y_i):~i=0,1,\ldots,n\}$.} given by
\begin{align}
\pi^{*,r}_n(dx_n|y_{n-1})&=\frac{\exp{\Big\{\int_{{\cal Y}_n}\log\Big(\frac{dr_n(\cdot|y_{n-1},y_n)}{d\pi_n(\cdot|y_{n-1})}(x_n)\Big)q_n(dy_n|y_{n-1},x_n)\Big\}}\pi_n(dx_n|y_{n-1})}{\int_{{\cal X}_n}\exp{\Big\{\int_{{\cal Y}_n}\log\Big(\frac{r^{\pi}_n(\cdot|y_{n-1},y_n)}{\pi_n(\cdot|y_{n-1})}(x_n)\Big)q_n(dy_n|y_{n-1},x_n)\Big\}}\pi_n(dx_n|y_{n-1})}\label{variational:equality:sequential:applications:equation14aa}
\end{align}
Moreover, when (\ref{variational:equality:sequential:applications:equation14aa}) is evaluated at $r_n(\cdot|\cdot,\cdot)=r^{*,\pi}_n(\cdot|\cdot,\cdot)$ given by (\ref{variational:equality:sequential:applications:equation14}) then
\begin{align}
\pi^{*,r^{*}}_n(dx_n|y_{n-1})=\frac{\exp{\Big\{\int_{{\cal Y}_n}\log\Big(\frac{dq_n(\cdot|y_{n-1},x_n)}{d\nu^{\pi}_{n}(\cdot|y_{n-1})}(y_n)\Big)q_n(dy_n|y_{n-1},x_n)\Big\}}\pi_n(dx_n|y_{n-1})}{\int_{{\cal X}_n}\exp{\Big\{\int_{{\cal Y}_n}\log\Big(\frac{dq_n(\cdot|y_{n-1},x_n)}{d\nu^{\pi}_{n}(\cdot|y_{n-1})}(y_n)\Big)q_n(dy_n|y_{n-1},x_n)\Big\}}\pi_n(dx_n|y_{n-1})}.\label{variational:equality:sequential:applications:equation15}
\end{align}
\noi{\bf Maximizations in (\ref{variational:equality:sequential:applications:equation13}).}\\ \noi {\bf(iii)} For a fixed $\pi_t(dx_t|y_{t-1})\in{\cal M}_1({\cal X}_t)$, the maximum in (\ref{variational:equality:sequential:applications:equation13}) over $r_t(dx_t|y_{t-1},y_t)\in{\cal M}_1({\cal X}_t)$ occurs at $r_t(\cdot|\cdot,\cdot)=r^{*,\pi}_t(\cdot|\cdot,\cdot)$ given by
\begin{align}
r^{*,\pi}_t(dx_t|y_{t-1},y_t)=\Big(\frac{dq_t(\cdot|y_{t-1},x_t)}{d\nu^{\pi}_{t}(\cdot|y_{t-1})}(y_t)\Big){\pi}_{t}(dx_t|y_{t-1}),~t=n-1,n-2,\ldots,0.
\label{variational:equality:sequential:applications:equation16}
\end{align}
\noi{\bf(iv)} For a fixed $r_t(dx_t|y_{n-1},y_t)\in{\cal M}_1({\cal X}_t)$, the maximum in (\ref{variational:equality:sequential:applications:equation13}) over $\pi_t(dx_t|y_{t-1})\in{\cal M}_1({\cal X}_t)$, occurs at $\pi_t(\cdot|\cdot)=\pi^{*,r}_t(\cdot|\cdot),~t=n-1,n-2,\ldots,0,$ given by
\begin{align}
\pi^{*,r}_t(dx_t|y_{t-1})&=\frac{\exp{\Big\{\int_{{\cal Y}_t}\Big\{\log\Big(\frac{dr_t(\cdot|y_{t-1},y_t)}{d{\pi}_{t}(\cdot|y_{t-1})}(x_t)\Big)+C_{t+1}(y_{t})\Big\}q_t(dy_t|y_{t-1},x_t)\Big\}}\pi_t(dx_t|y_{t-1})}{\int_{{\cal X}_t}\exp{\Big\{\int_{{\cal Y}_t}\Big\{\log\Big(\frac{dr_t(\cdot|y_{t-1},y_t)}{d{\pi}_{t}(\cdot|y_{t-1})}(x_t)\Big)+C_{t+1}(y_{t})\Big\}q_t(dy_t|y_{t-1},x_t)\Big\}}\pi_t(dx_t|y_{t-1})}.\label{variational:equality:sequential:applications:equation17aa}
\end{align}
Moreover, when (\ref{variational:equality:sequential:applications:equation17aa}) is evaluated at $r_t(\cdot|\cdot,\cdot)=r^{*,\pi}_t(\cdot|\cdot,\cdot),~t=n-1,n-2,\ldots,0,$ given by (\ref{variational:equality:sequential:applications:equation16}) then
\begin{align}
\pi^{*,r^{*}}_t(dx_t|y_{t-1})=\frac{\exp{\bigg\{\int_{{\cal Y}_t}\Big\{\log\Big(\frac{dq_t(\cdot|y_{t-1},x_t)}{d\nu^{\pi}_{t}(\cdot|y_{t-1})}(y_t)\Big)+C_{t+1}(y_{t})\Big\}q_t(dy_t|y_{t-1},x_t)\bigg\}}\pi_t(dx_t|y_{t-1})}{\int_{{\cal X}_t}\exp{\bigg\{\int_{{\cal Y}_t}\Big\{\log\Big(\frac{dq_t(\cdot|y_{t-1},x_t)}{d\nu^{\pi}_{t}(\cdot|y_{t-1})}(y_t)\Big)+C_{t+1}(y_{t})\Big\}q_t(dy_t|y_{t-1},x_t)\bigg\}}\pi_t(dx_t|y_{t-1})}.\label{variational:equality:sequential:applications:equation17}
\end{align}
\noi{\bf Part B.} The extremum problem $C^{fb,UMCO}_{0,n}$ defined by (\ref{variational:equality:sequential:applications:equation4}) is equivalent to the following sequential double maximization problem.
\begin{align}
C^{fb,UMCO}_{0,n}&=\sup_{\pi_0(dx_0|y_{-1})\in{\cal M}_1({\cal X}_0)}\sup_{r^{\pi}_0(dx_0|y_{-1},y_0)\in{\cal M}_1({\cal X}_0)}\ldots\sup_{\pi_n(dx_n|y_{n-1})\in{\cal M}_1({\cal X}_n)}\sup_{r^{\pi}_n(dx_n|y_{n-1},y_n)\in{\cal M}_1({\cal X}_n)}\nonumber\\
&\bigg\{\sum_{i=0}^n\int_{{\cal Y}_{i-1,i}\times{\cal X}_{i}}\log\Big(\frac{dr^{\pi}_i(\cdot|y_{i-1},y_i)}{d\pi_{i}(\cdot|y_{i-1})}(x_i)\Big){q}_i(dy_i|y_{i-1},x_i)\otimes\pi_i(dx_i|y_{i-1})\otimes\nu^{\pi}_{i}(dy_{i-1})\bigg\}\label{variational:equality:sequential:applications:equation18}
\end{align}
and statements {\bf (i)-(iv)} hold.
\end{theorem}
\begin{proof}
{\bf Part A. (i)} (\ref{variational:equality:sequential:applications:equation12}) and (\ref{variational:equality:sequential:applications:equation14}) follow directly from (\ref{variational:equality:sequential:applications:equation10}). {\bf (ii)} (\ref{variational:equality:sequential:applications:equation14aa}) is obtained as follows. Fix $r_n(dx_n|y_{n-1},y_n)\in{\cal M}_1({\cal X}_n)$, calculate the G\^ateaux differential inside the maximization in (\ref{variational:equality:sequential:applications:equation12}) at $\pi^{*,r}_n(dx_n|y_{n-1})$ in the direction $\pi^r_n(dx_n|y_{n-1})-\pi^{*,r}_n(dx_n|y_{n-1})$, i.e., $\pi^{\epsilon,r}_n(dx_n|y_{n-1})\tri\pi^{*,r}_n(dx_n|y_{n-1})-\epsilon\big\{\pi^r_n(dx_n|y_{n-1})-\pi^{*,r}_n(dx_n|y_{n-1})\big\},~\epsilon\in[0,1]$,  by incorporating the constraint $\int_{{\cal X}_n}\pi^r_n(dx_n|y_{n-1})=1$ via a Lagrange multiplier $\lambda_n(y_{n-1})$. The G\^ateaux differential gives (\ref{variational:equality:sequential:applications:equation14aa}). Then substitute (\ref{variational:equality:sequential:applications:equation14}) into (\ref{variational:equality:sequential:applications:equation14aa}) to obtain (\ref{variational:equality:sequential:applications:equation15}). {\bf(iii)} For fixed $\pi_t(dx_t|y_{t-1})\in{\cal M}_1({\cal X}_t)$, the second RHS term in (\ref{variational:equality:sequential:applications:equation11}) is a function of the channel distribution, hence (\ref{variational:equality:sequential:applications:equation13}) and (\ref{variational:equality:sequential:applications:equation16}) follow directly as in {\bf (i)}. {\bf(iv)} To show (\ref{variational:equality:sequential:applications:equation17aa}), (\ref{variational:equality:sequential:applications:equation17}), compute the G\^ateux differential as in {\bf (ii)}, by tracking the additional  second RHS term in (\ref{variational:equality:sequential:applications:equation13}).\\
\noi{\bf Part B.} Since $\nu_{-1}(dy_{-1})\in{\cal M}_1({\cal Y}_{-1})$ is fixed, then (\ref{variational:equality:sequential:applications:equation18}) follows directly from {\bf Part A.}, and the definition of $C_t(y_{t-1})$ evaluated at $t=0$.
\end{proof}
\noi Theorem~\ref{applications:sequential:double maximization:theorem}, specifically (\ref{variational:equality:sequential:applications:equation15}), (\ref{variational:equality:sequential:applications:equation17}), are the equations, which should be used to derive a sequential algorithm to compute numerically the optimal channel input distribution.\\
\noi Below, we discuss applications of Theorem~\ref{applications:sequential:double maximization:theorem}, and identify generalizations, and directions for future research.
\begin{remark}(Sequential algorithms for feedback capacity)
\begin{itemize}
\item[(1)]For the UMCO, Theorem~\ref{applications:sequential:double maximization:theorem} provides all necessary ingredients to derive a sequential algorithm at each time step, $t=n,n-1,\ldots,0$, similar to the BAA. It remains to show at each time step, $t=n,n-1,\ldots,0$, that (\ref{variational:equality:sequential:applications:equation15}), (\ref{variational:equality:sequential:applications:equation17aa}) have fixed points corresponding to the optimal channel input distribution, and to derive upper and lower bounds on $C_t(y_{t-1}),~t=n,n-1,\ldots,0$, to stop the iterations at each time step of the algorithm. For finite alphabet spaces $\{({\cal X}_i, {\cal Y}_i):~i=0,1,\ldots,n\}$, these additional steps can be carried out using Theorem~\ref{applications:sequential:double maximization:theorem} and the procedure in \cite{blahut1972}.
\item[(2)] For the UMCO, if the alphabet spaces ${\cal X}_i\equiv{\cal X}$, ${\cal Y}_i\equiv{\cal Y},~i=0,1,\ldots$, and the joint process $\{(X_i,Y_i):~i=0,1,\ldots\}$ is stationary ergodic or directed information stable, then the per unit time limiting version of dynamic programming recursive equations (\ref{variational:equality:sequential:applications:equation10}), (\ref{variational:equality:sequential:applications:equation11}) can be derived \cite{lerma-lasserre1996}, and these involve only a single stage maximization over $\pi(dx_i|y_{i-1})\in{\cal M}_1({\cal X}),~\forall{i}$. Hence, a theorem similar to Theorem~\ref{applications:sequential:double maximization:theorem} can be derived.
\item[(3)] For general channels, it is possible to derive the analogue of Theorem~\ref{applications:sequential:double maximization:theorem}, provided the set of optimal channel input distributions, which maximize $\sum_{i=0}^n{I}(X^i;Y_i|Y^{i-1})$ is identified, as in the case of UMCO (see \cite{kourtellaris-charalambous2015aieeeit}).
\end{itemize}
\end{remark}
%%%%%%%%%%%%%%%%%%%%%%%%%%%%%%%%%%%%%%%%%%%%%%%%%%%%%%%%%%

\section{Conclusion}\label{conclusion}
In this paper we derive functional and topological properties of directed information, for  abstract alphabet spaces (i.e., complete separable metric spaces). These include, convexity of the set of consistent family of distributions, which uniquely define causally conditioned compound distributions, convexity and concavity of directed information with respect to consistent family of  distributions, and  a general theorem on weak compactness of causally conditioned distributions, their joint distributions, and marginals, which are utilized to define directed information. Further, we use this main theorems  to show lower semicontinuity of directed information as a functional of two causally conditioned distributions, and under additional conditions continuity of directed information. In addition, we derive sequential  variational  equalities for directed information. Throughout the paper, we discuss application examples  in the context of  extremum problems of directed information, such as, in feedback capacity, nonanticipative RDF, and in developing sequential  computational algorithms, similar to the Blahut-Arimoto algorithm \cite{blahut1972}.

%\vspace*{1.cm}

\appendices

\section{Background material}\label{backround_material}

In this section, we introduce some of the basic analytical concepts which are used throughout the paper.
\vspace*{0.2cm}\\
\noi{\bf Weak Convergence and Compactness.}\\
The main notions discussed are weak convergence of probability measures, the relation to convergence with respect to Prohorov metric, tightness of a family of probability measures and relative compactness\cite{billingsley1999}.\\
Let $({\cal X},d)$ be a metric space, ${\cal B}({\cal X})$ the $\sigma-$algebra of Borel subsets of ${\cal X}$, and ${\cal M}_1({\cal X})$ the family of probability measures on ${\cal X}$. Let $BC({\cal X})$ denote the set of bounded, continuous real-valued function $f$ on $({\cal X},d)$, endowed with the supremum norm $||f||=\sup_{x\in{\cal X}}|f(x)|$.  A sequence of probability measures $\{P_n:n=1,2,\ldots\}\subset{\cal M}_1({\cal X})$ is said to {\it converge weakly} to a probability measure $P\in{\cal M}_1({\cal X})$ if
\begin{align}
\lim_{n\rightarrow\infty}\int_{\cal X}f(x)dP_n(x)=\int_{\cal X}f(x)dP(x),~\forall{f}\in{BC}({\cal X}).\nonumber  
\end{align}
Weak convergence of $\{P_n:n=1,2,\ldots\}$ to $P$ is denoted by $P_n\buildrel w \over \longrightarrow{P}$. The space of probability measures ${\cal M}_1({\cal X})$ is metrizable with respect to the Prohorov metric (see \cite{dupuis-ellis97}).\\

\noi A crucial result for the characterization of compact subsets of ${\cal M}_1({\cal X})$ is the next theorem due to Prohorov, which relates compactness and tightness of a family of measures.
\begin{definition} $($Tightness and Relative Compactness$)$\cite[p. 308]{dupuis-ellis97}{\ \\}
Let $M\subset{\cal M}_1({\cal X})$ be a family of probability measures on a metric space $({\cal X}, d)$. $M$ is said to be
\item[1)] {\it tight} or {\it uniformly tight} if for every $\epsilon>0$ there exists a compact set $K^{(\epsilon)}\subset{\cal X}$ such that
$\inf_{P\in{M}}P(K^{(\epsilon)})\geq{1-\epsilon}$;
\item[2)] relatively compact or weakly compact if every sequence in $M$ contains a weakly convergent subsequence, that is, for every sequence $\{P_n:n=1,2,\ldots\}$ in $M$ there is a subsequence $\{P_{n_{i}}:i\in\{1,2,\ldots\}\}$ and a $P\in{\cal M}_1({\cal X})$ such that $P_{n_i}\buildrel w \over \longrightarrow{P}$. Here, the limit $P$ is not required to belong to $M$, but all is required is to belong to ${\cal M}_1({\cal X})$.
\end{definition}
\noi Prohorov states that for $({\cal X},d)$ a metric space and ${\cal X}$ compact, then any sequence $\{P_n:n=1,2,\ldots\}$ of probability measures on ${\cal X}$ possess a convergent subsequence. The following theorem due to Prohorov, relates weak compactness and tightness of a family of probability measures.
\begin{theorem}$($Prohorov's Theorem$)$\label{prohorov}\cite[Theorem A.3.15,~p.~309]{dupuis-ellis97}{\ \\}
Let $M\subset{\cal M}_1({\cal X})$ be a family of probability measures on a metric space $({\cal X},d)$.
\item[1)] If $M$ is tight, then it is relative compact.
\item[2)] Suppose ${\cal X}$ is separable and complete. If $M$ is relatively compact, then it is tight.
\end{theorem}
\noi Thus, a family of probability measures $M\subset{\cal M}_1({\cal X})$ on a complete separable metric space $({\cal X},d)$ is weakly compact or relatively compact with respect to weak convergence if and only if it is tight. Moreover, if $P_n \buildrel w \over \longrightarrow{P}$, then the family $\{P_n:n=1,2,\ldots\}$ is tight.\\
Finally, we give another version due to Prohorov for a family of measures $M\subset{\cal M}_1({\cal X})$ to be compact.
\begin{theorem}(Corollary of Prohorov's Theorem)\label{corollary_of_prohorov}{\ \\}
Let $({\cal X},d)$ be a separable metric and $M\subset{\cal M}_1({\cal X})$ a set of measures. The following hold.
\item[(a)] If $M$ is closed and tight, then $M$ is compact.
\item[(b)] Suppose ${\cal X}$ is complete. If $M$ is compact then $M$ is closed and tight.
\end{theorem} 

\noi In what follows, we give the definition of weak continuity of conditional distributions, which is often associated with proving results using weak convergence of probability distributions, and we distinguish it from strong continuity.

\begin{definition}(Strong and weak continuity){\ \\}
\label{App_SW}
Let $({\cal X},d)$, $({\cal Y},d')$ be metric spaces, $Q(\cdot|\cdot)\in{\cal Q}({\cal Y}|{\cal X})$ a conditional distribution, and define by $BM({\cal Y})$ the set of bounded measurable functions on ${\cal Y}$. Then $Q(\cdot|\cdot)\in{\cal Q}({\cal Y}|{\cal X})$ is said to be\\
1) strongly continuous if the function mapping
\begin{align*}
x\longmapsto\int_{\cal Y}f(y)Q(dy|x)\in{BC}({\cal Y})
\end{align*}
whenever $f(\cdot)\in{BM}({\cal Y})$,\\
2) weakly continuous if the function mapping 
\begin{align*}
x\longmapsto\int_{\cal Y}f(y)Q(dy|x)\in{BC}({\cal Y})
\end{align*}
whenever $f(\cdot)\in{BC}({\cal Y})$.
\end{definition}
\noi It can be shown that strong continuity is equivalent to  $Q(B|\cdot)$ is continuous on ${\cal Y}$ for every set $B\in{\cal B}({\cal Y})$ (i.e., its conditional distribution is continuous), and this is much stronger than weak continuity of  $Q(\cdot|\cdot)\in{\cal Q}({\cal Y}|{\cal X})$.\\ 
\noi{\bf Uniform Integrability.}\\
\noi In this paper we shall also need stronger sufficient conditions to verify convergence of a sequence of integrals using the concept of uniform integrability. We state this next. 
\begin{definition}$($Uniform Integrability of RV's$)$\label{Uniform-Integrability}\cite[Definition 4, p.~188]{shiryaev1984}{\ \\}
Let $(\Omega,{\cal F},\mathbb{P})$ be a probability space. A sequence of RV's $\{X_n:n\in\mathbb{N}_1\}$, $\mathbb{N}_1\triangleq\{1,2,\ldots\}$, is said to be uniformly $\mathbb{P}$-integrable if
\begin{align}
\lim_{c\rightarrow{\infty}}\sup_{n\in\mathbb{N}_1}\int_{\{\omega:|X_n(\omega)|\geq{c}\}}|X_n(\omega)|d\mathbb{P}(\omega)={0}.\nonumber
\end{align}
\end{definition}
\noi Note that if $\{X_n:n\in\mathbb{N}_1\}$ satisfy $|X_n|\leq{Y}$ and $\mathbb{E}\{Y\}<\infty$, then the sequence $\{X_n:n\in\mathbb{N}_1\}$ is uniformly integrable. \\
The following theorem gives some properties for a family of uniformly integrable RV's.
\begin{theorem}$($Uniform Integrability of RV's$)$\label{uniform_integrability_1}\cite[Theorem 4, pp.~188-189]{shiryaev1984}{\ \\}
Let $(\Omega,{\cal F},\mathbb{P})$ be a probability space and $\{X_n:n\in\mathbb{N}_1\}$ a uniformly $\mathbb{P}$-integrable family of RV's. Then
\begin{enumerate}
\item[(a)] $\mathbb{E}\liminf_{n}X_n\leq\liminf_{n}\mathbb{E}X_n\leq\limsup_{n}\mathbb{E}X_n\leq\mathbb{E}\limsup_{n}X_n$.
\item[(b)]If $X_n\buildrel a.s. \over\Longrightarrow{X}$, then $\mathbb{E}|X|<\infty$, $\lim_{n\rightarrow\infty}\mathbb{E}|X_n|=\mathbb{E}|X|$ and $\lim_{n\rightarrow\infty}\mathbb{E}\Big\{|X_n-X|\Big\}={0}$.
\end{enumerate}
\end{theorem}
\noi The next definition of uniform integrability is with respect to a family of probability measures for a fixed integrand.
\begin{definition}\label{Uniform-Integrability-2}$($Uniform Integrability for a family of probability measures$)${\ \\}
Let ${M}\subset{\cal M}_1({\cal X})$ be a family of probability measures on $\big{(}{\cal X},{\cal B}({\cal X})\big{)}$. A measurable function $f$ on ${\cal X}$ is said to be uniformly integrable over ${M}$ if
\begin{align*}
\lim_{c\rightarrow\infty}\sup_{P\in{M}}\int_{\{x\in{\cal X}:|f(x)|>c\}}|f(x)|d{P}(x)={0}.
\end{align*}
\end{definition} 
\noi A sufficient condition for the convergence of a sequence of integrals of a function with respect to a weakly convergent sequence of measures is the following.
\begin{theorem}\label{continuity-uniform_integrability}\cite[Appendix, Theorem~A.2, p.~3084]{fozunbal}{\ \\}
Let $M\subset{\cal M}_1({\cal X})$ be a closed family of probability measures on $\big{(}{\cal X},{\cal B}({\cal X})\big{)}$, and let $\{P_n:n\in\mathbb{N}_1\}\subset{M}$ be a weakly convergent sequence in $M$. If $f$ is a continuous function on ${\cal X}$ and uniformly integrable over $\{P_n:n\in\mathbb{N}_1\}$ then $\lim_{n\rightarrow\infty}\int{f}(x)dP_n(x)=\int{f}(x)dP(x)$. 
\end{theorem}

\noi{\bf Absolute Continuity of Probability Measures.}\label{regular}\\
Let $({\Omega},{\cal F})$ be a measurable space. Given two probability measures $P, Q$ on $({\Omega},{\cal F}),$ $Q$ is said to be {\it absolutely continuous} with respect to $P$ (denoted $P\ll{Q}$) if for every $A\in{\cal F}$ such that $P(A)=0$ then $Q(A)=0$. If ${Q}\ll{P},$ by Radon-Nikodym Derivative theorem, there exists a $P-$integrable and ${\cal F}-$measurable function $f$ such that for every $A\in{\cal F},$ $Q(A)=\int_Af(\omega)dP(\omega)$. Let $({\Omega},{\cal F},\mathbb{P})$ be a probability space and ${\cal G}$ be a sub-$\sigma$-algebra of ${\cal F}.$ A {\it regular conditional probability distribution} $P(\cdot|{\cal G})$ on $({\Omega},{\cal F})$ exist, when ${\cal G}$ is generated by a countable partition of $\Omega$. Moreover, if $(\Omega,d)$ is a metric space which is complete and separable (Polish space), and ${\cal F}$ is a Borel $\sigma-$algebra, then for any probability measure $P$ on $(\Omega,{\cal F})$ and any sub-$\sigma$-algebra ${\cal G}\subseteq{\cal F},$ a regular conditional probability measure of $P$ given ${\cal G}$ always exists.\\
The next lemma summarizes certain relationships between the absolute continuity of probability measures. 
\begin{lemma}\label{absolute_contunuity}(Absolute Continuity of Probability Measures)\cite[Lemma 4.4.7, pp. 149-150]{deuschel-stroock1989}
\item[a)] Suppose $Q_{\cal G}\ll{P}_{\cal G}.$ If $Q(\cdot|{\cal G})(\omega)\ll{P}(\cdot|{\cal G})(\omega),~Q_{\cal G}-a.s.,$ then $Q\ll{P}.$
\item[b)] Conversely, if $Q\ll{P},$ then $Q(\cdot|{\cal G})(\omega)\ll{P}(\cdot|{\cal G})(\omega),~{P}(\cdot|{\cal G})(\omega)-a.s.$
\end{lemma}

\noi If $Y:(\Omega,{\cal F})\longmapsto({\cal Y},{\cal A})$ is a RV on $(\Omega,{\cal F})$ into a measurable space $({\cal Y},{\cal A})$ and ${\cal Y}$ is a Polish space, then a regular conditional distribution for $Y$ given the sub-$\sigma$-algebra ${\cal G}$ of ${\cal F}$ denoted by $P(dy|{\cal G})(\omega)$, always exists. Additionally, if $X:(\Omega,{\cal F})\longmapsto({\cal X},{\cal B})$ is a RV on $(\Omega,{\cal F})$ into a measurable space $({\cal X},{\cal B}),$ and ${\cal G}$ is the sub-$\sigma$-algebra of ${\cal F}$ generated by X, then $P(dy|X)(\omega)$ is called the {\it regular conditional distribution} of $Y$ given $X.$ One can go one step further to define an equivalent definition of a regular conditional distribution for $Y$ given $X=x$ as a quantity $P(dy|X=x)$ called stochastic kernel.

%%%%%%%%%%%%%%%%%%%%%%%%%%%%%%%%%%%%%%%%%%%%%%%%%%%%%%%%%%%%%%%%%%

\section{Proof of Theorem~\ref{convexity1}}\label{convexity_of_functionals}
1) Fix ${\overleftarrow P}_{0,n}(\cdot|y^{n-1})\in{\cal M}_1^{\bf C1}({\cal X}_{0,n})$ and let ${\overrightarrow Q}^1_{0,n}(\cdot|x^n)$, ${\overrightarrow Q}^2_{0,n}(\cdot|x^n)\in{\cal M}_1^{\bf C2}({\cal Y}_{0,n})$. Then, the joint distributions corresponding  to ${\overrightarrow Q}^1_{0,n}(\cdot|x^n)$, ${\overrightarrow Q}^2_{0,n}(\cdot|x^n)$ are
\begin{align}
({\overleftarrow P}_{0,n}\otimes{\overrightarrow Q}^1_{0,n})(dx^n,dy^n)~\mbox{and}~({\overleftarrow P}_{0,n}\otimes{\overrightarrow Q}^2_{0,n})(dx^n,dy^n),\nonumber
\end{align}
and the marginals are
\begin{align}
\nu_{0,n}^1(dy^n)=({\overleftarrow P}_{0,n}\otimes{\overrightarrow Q}^1_{0,n})({\cal X}_{0,n},dy^n),\hst\nu_{0,n}^2(dy^n)=({\overleftarrow P}_{0,n}\otimes{\overrightarrow Q}^2_{0,n})({\cal X}_{0,n},dy^n).\nonumber
\end{align}
Since the set ${\cal M}_1^{\bf C2}({\cal Y}_{0,n})$ is convex, given $\lambda\in(0,1)$ there exists a probability measure $\tilde {\bf P}$ on $({\cal X}^{\mathbb{N}_0}\times{\cal Y}^{\mathbb{N}_0},{\cal B}({\cal X}^{\mathbb{N}_0})\otimes{\cal B}({\cal Y}^{\mathbb{N}_0}))$ whose regular conditional measure ${\bf Q}(\cdot|{\bf x})\in{\cal M}_1({\cal Y}^{\mathbb{N}_0})$ satisfies
\begin{align}
{\overrightarrow Q}_{0,n}(\cdot|x^n)=\lambda{\overrightarrow Q}_{0,n}^1(\cdot|x^n)+(1-\lambda){\overrightarrow Q}_{0,n}^2(\cdot|x^n),~{\bar {\bf P}}\big{|}_{{\cal B}({\cal X}_{0,n})}-a.e.~x^n\nonumber
\end{align}
and ${\bf C1}$ holds. Define
\begin{align}
\nu_{0,n}(dy^n)&={\lambda}\nu_{0,n}^1(dy^n)+(1-\lambda)\nu_{0,n}^2(dy^n).\nonumber
\end{align}
Introduce the RNDs
${\Lambda}_{0,n}^i(x^n,y^n)=\frac{d{\overrightarrow Q}^i_{0,n}(\cdot|x^n)}{d{\nu}^i_{0,n}(\cdot)}(y^n)$, ${\Psi}_{0,n}^i(x^n,y^n)=\frac{d{\overrightarrow Q}^i_{0,n}(\cdot|x^n)}{d{\nu}_{0,n}(\cdot)}(y^n)$, ${K}_{0,n}^i(y^n)=\frac{d{\nu}^i_{0,n}(\cdot)}{d{\nu}_{0,n}(\cdot)}(y^n)$ and ${\Lambda}_{0,n}(x^n,y^n)=\frac{d{\overrightarrow Q}_{0,n}(\cdot|x^n)}{d{\nu}_{0,n}(\cdot)}(y^n),~i=1,2$. Then,
\begin{align}
\lambda{\Psi}_{0,n}^1(x^n,y^n)+(1-\lambda){\Psi}_{0,n}^2(x^n,y^n)&=\lambda\frac{d{\overrightarrow Q}^1_{0,n}(\cdot|x^n)}{d{\nu}_{0,n}(\cdot)}(y^n)+(1-\lambda)\frac{d{\overrightarrow Q}^2_{0,n}(\cdot|x^n)}{d{\nu}_{0,n}(\cdot)}(y^n)\nonumber\\
&=\frac{d\big({\lambda}{\overrightarrow Q}^1_{0,n}(\cdot|x^n)+(1-\lambda){\overrightarrow Q}^2_{0,n}(\cdot|x^n)\big)}{d\big({\lambda}{\nu}^1_{0,n}(\cdot)+(1-\lambda){\nu}^2_{0,n}(\cdot)\big)}(y^n)=\Lambda_{0,n}(x^n,y^n)\nonumber
\end{align}
and
\begin{align}
\lambda{K}_{0,n}^1(y^n)+(1-\lambda){K}_{0,n}^2(y^n)&=\lambda\frac{d{\nu}^1_{0,n}(\cdot)}{d{\nu}_{0,n}(\cdot)}(y^n)+(1-\lambda)\frac{d{\nu}^2_{0,n}(\cdot)}{d{\nu}_{0,n}(\cdot)}(y^n)=\frac{d\big({\lambda}{\nu}^1_{0,n}(\cdot)+(1-\lambda){\nu}^2_{0,n}(\cdot)\big)}{d\big({\lambda}{\nu}^1_{0,n}(\cdot)+(1-\lambda){\nu}^2_{0,n}(\cdot)\big)}(y^n)=1.\nonumber
\end{align}
Applying the log-sum formula \cite[Theorem 2.7.1, p.~31]{cover-thomas} yields
\begin{align}
&\lambda{\Psi}_{0,n}^1(x^n,y^n)\log\Lambda_{0,n}^1(x^n,y^n)+(1-\lambda){\Psi}_{0,n}^2(x^n,y^n)\log\Lambda_{0,n}^2(x^n,y^n)\nonumber\\
&=\lambda{\Psi}_{0,n}^1(x^n,y^n)\log\Bigg(\frac{\frac{d{\overrightarrow Q}^1_{0,n}(\cdot|x^n)}{d{\nu}_{0,n}(\cdot)}(y^n)}{\frac{d{\nu}^1_{0,n}(\cdot)}{d{\nu}_{0,n}(\cdot)}(y^n)}\Bigg)
+(1-\lambda){\Psi}_{0,n}^2(x^n,y^n)\log\Bigg(\frac{\frac{d{\overrightarrow Q}^2_{0,n}(\cdot|x^n)}{d{\nu}_{0,n}(\cdot)}(y^n)}{\frac{d{\nu}^2_{0,n}(\cdot)}{d{\nu}_{0,n}(\cdot)}(y^n)}\Bigg)\nonumber\\
&=\lambda\frac{d{\overrightarrow Q}^1_{0,n}(\cdot|x^n)}{d{\nu}_{0,n}(\cdot)}(y^n)\log\Bigg(\frac{\lambda\frac{d{\overrightarrow Q}^1_{0,n}(\cdot|x^n)}{d{\nu}_{0,n}(\cdot)}(y^n)}{\lambda\frac{d{\nu}^1_{0,n}(\cdot)}{d{\nu}_{0,n}(\cdot)}(y^n)}\Bigg)
+(1-\lambda)\frac{d{\overrightarrow Q}^2_{0,n}(\cdot|x^n)}{d{\nu}_{0,n}(\cdot)}(y^n)\log\Bigg(\frac{(1-\lambda)\frac{d{\overrightarrow Q}^2_{0,n}(\cdot|x^n)}{d{\nu}_{0,n}(\cdot)}(y^n)}{(1-\lambda)\frac{d{\nu}^2_{0,n}(\cdot)}{d{\nu}_{0,n}(\cdot)}(y^n)}\Bigg)\nonumber\\
&\geq\Bigg(\lambda\frac{d{\overrightarrow Q}^1_{0,n}(\cdot|x^n)}{d{\nu}_{0,n}(\cdot)}(y^n)+(1-\lambda)\frac{d{\overrightarrow Q}^2_{0,n}(\cdot|x^n)}{d{\nu}_{0,n}(\cdot)}(y^n)\Bigg)\log\Bigg(\frac{\lambda\frac{d{\overrightarrow Q}^1_{0,n}(\cdot|x^n)}{d{\nu}_{0,n}(\cdot)}(y^n)+(1-\lambda)\frac{d{\overrightarrow Q}^2_{0,n}(\cdot|x^n)}{d{\nu}_{0,n}(\cdot)}(y^n)}{\lambda\frac{d{\nu}^1_{0,n}(\cdot)}{d{\nu}_{0,n}(\cdot)}(y^n)+(1-\lambda)\frac{d{\nu}^2_{0,n}(\cdot)}{d{\nu}_{0,n}(\cdot)}(y^n)}\Bigg)\nonumber\\
&=\frac{d{\overrightarrow Q}_{0,n}(\cdot|x^n)}{d{\nu}_{0,n}(\cdot)}(y^n)\log\frac{d{\overrightarrow Q}_{0,n}(\cdot|x^n)}{d{\nu}_{0,n}(\cdot)}(y^n).\nonumber
\end{align}
Integrating the above with respect to $\nu_{0,n}(dy^n)\otimes{\overleftarrow P}_{0,n}(dx^n|y^{n-1})$ yields:
\begin{align}
&\int_{{\cal X}_{0,n}\times{\cal Y}_{0,n}}\log\Big(\frac{d{\overrightarrow Q}_{0,n}(\cdot|x^n)}{d{\nu}_{0,n}(\cdot)}(y^n)\Big)\frac{d{\overrightarrow Q}_{0,n}(\cdot|x^n)}{d{\nu}_{0,n}(\cdot)}(y^n)\big(\nu_{0,n}(dy^n)\otimes{\overleftarrow P}_{0,n}(dx^n|y^{n-1})\big)\nonumber\\
&=\int_{{\cal X}_{0,n}\times{\cal Y}_{0,n}}\log\Big(\frac{d{\overrightarrow Q}_{0,n}(\cdot|x^n)}{d{\nu}_{0,n}(\cdot)}(y^n)\Big)({\overrightarrow Q}_{0,n}\otimes{\overleftarrow P}_{0,n})(dx^n,dy^{n})\nonumber\\
&\leq\lambda\int_{{\cal X}_{0,n}\times{\cal Y}_{0,n}}\log\Big(\frac{d{\overrightarrow Q}^1_{0,n}(\cdot|x^n)}{d{\nu}^1_{0,n}(\cdot)}(y^n)\Big)({\overrightarrow Q}_{0,n}^1\otimes{\overleftarrow P}_{0,n})(dx^n,dy^{n})\nonumber\\
&+(1-\lambda)\int_{{\cal X}_{0,n}\times{\cal Y}_{0,n}}\log\Big(\frac{d{\overrightarrow Q}^2_{0,n}(\cdot|x^n)}{d{\nu}^2_{0,n}(\cdot)}(y^n)\Big)({\overrightarrow Q}_{0,n}^2\otimes{\overleftarrow P}_{0,n})(dx^n,dy^{n}).\nonumber
\end{align}
Hence,
\begin{align}
{\mathbb{I}}_{X^n\rightarrow{Y^n}}\big({\overleftarrow P}_{0,n},\lambda{\overrightarrow Q}_{0,n}^1+(1-\lambda){\overrightarrow Q}_{0,n}^2\big)\leq\lambda{\mathbb{I}}_{X^n\rightarrow{Y^n}}\big({\overleftarrow P}_{0,n},{\overrightarrow Q}_{0,n}^1\big)+(1-\lambda){\mathbb{I}}_{X^n\rightarrow{Y^n}}\big({\overleftarrow P}_{0,n},{\overrightarrow Q}_{0,n}^2\big).\nonumber
\end{align}
This completes the derivation of 1).\\
2) Fix ${\overrightarrow Q}_{0,n}(\cdot|x^n)\in{\cal M}_1^{\bf C2}({\cal Y}_{0,n})$ and let ${\overleftarrow P}^1_{0,n}(\cdot|y^{n-1})$, ${\overleftarrow P}^2_{0,n}(\cdot|y^{n-1})\in{\cal M}_1^{\bf C1}({\cal X}_{0,n})$. Then, the joint distributions corresponding to ${\overleftarrow P}^1_{0,n}(\cdot|y^{n-1})$, ${\overleftarrow P}^2_{0,n}(\cdot|y^{n-1})$ are
\begin{align}
({\overleftarrow P}^1_{0,n}\otimes{\overrightarrow Q}_{0,n})(dx^n,dy^n)~\mbox{and}~({\overleftarrow P}^2_{0,n}\otimes{\overrightarrow Q}_{0,n})(dx^n,dy^n).\nonumber
\end{align}
The marginals corresponding to ${\overleftarrow P}_{0,n}^1(\cdot|y^{n-1})$, ${\overleftarrow P}_{0,n}^2(\cdot|y^{n-1})$ are
\begin{align}
\nu_{0,n}^1(dy^n)&=({\overleftarrow P}^1_{0,n}\otimes{\overrightarrow Q}_{0,n})({\cal X}_{0,n},dy^n),~\nu_{0,n}^2(dy^n)=({\overleftarrow P}^2_{0,n}\otimes{\overrightarrow Q}_{0,n})({\cal X}_{0,n},dy^n).\nonumber
\end{align}
Since the set ${\cal M}_1^{\bf C1}({\cal X}_{0,n})$ is convex, given $\lambda\in(0,1)$ there exists a probability measure $\tilde {\bf P}$ on $({\cal X}^{\mathbb{N}_0}\times{\cal Y}^{\mathbb{N}_0},{\cal B}({\cal X}^{\mathbb{N}_0})\otimes{\cal B}({\cal Y}^{\mathbb{N}_0}))$ whose regular conditional measure ${\bf P}(\cdot|{\bf y})\in{\cal M}_1({\cal X}^{\mathbb{N}_0})$ satisfies
\begin{align}
{\overleftarrow P}_{0,n}(\cdot|y^{n-1})=\lambda{\overleftarrow P}^1_{0,n}(\cdot|y^{n-1})+(1-\lambda){\overleftarrow P}^2_{0,n}(\cdot|y^{n-1}),~\bar{\bf P}\big{|}_{{\cal B}({\cal Y}_{0,n-1})}-a.e.~y^{n-1}\nonumber
\end{align}
and ${\bf C2}$ holds. Then, corresponding to ${\overleftarrow P}_{0,n}(\cdot|y^{n-1})$ and ${\overrightarrow Q}_{0,n}(\cdot|x^n)$ we have
\begin{align}
\nu_{0,n}(dy^n)&=\int_{{\cal X}_{0,n}}\big(\lambda{\overleftarrow P}^1_{0,n}(dx^n|y^{n-1})+(1-\lambda){\overleftarrow P}^1_{0,n}(dx^n|y^{n-1})\big)\otimes{\overrightarrow Q}_{0,n}(dy^n|x^n)\nonumber\\
&=\lambda({\overleftarrow P}^1_{0,n}\otimes{\overrightarrow Q}_{0,n})({\cal X}_{0,n},dy^n)+(1-\lambda)({\overleftarrow P}^2_{0,n}\otimes{\overrightarrow Q}_{0,n})({\cal X}_{0,n},dy^n)=\lambda\nu_{0,n}^1(dy^n)+(1-\lambda)\nu_{0,n}^2(dy^n).\nonumber
\end{align}
Pick any measure $U_{0,n}(dy^n)\in{\cal M}_1({\cal Y}_{0,n})$ with $\mathbb{D}(\nu_{0,n}||U_{0,n})<\infty,$ e.g., such that $\nu_{0,n}(\cdot){\ll}U_{0,n}(\cdot).$ Since ${\overrightarrow Q}(\cdot|x^n){\ll}\nu_{0,n}(\cdot)$, for almost all $x^n\in{\cal X}_{0,n}$, and $\nu_{0,n}(\cdot){\ll}U_{0,n}(\cdot)$, then ${\overrightarrow Q}_{0,n}(\cdot|x^n){\ll}U_{0,n}(\cdot)$, for almost all $x^n\in{\cal X}_{0,n}$. Consider
\begin{align}
{\mathbb{I}}_{X^n\rightarrow{Y^n}}&\big({\overleftarrow P}_{0,n},{\overrightarrow Q}_{0,n}\big)=\int_{{\cal X}_{0,n}\times{\cal Y}_{0,n}}\log\Big(\frac{d{\overrightarrow Q}_{0,n}(\cdot|x^n)}{d{\nu}_{0,n}(\cdot)}(y^n)\Big)({\overrightarrow Q}_{0,n}\otimes{\overleftarrow P}_{0,n})(dx^n,dy^{n})\nonumber\\
&=\int_{{\cal X}_{0,n}\times{\cal Y}_{0,n}}\log\Big(\frac{d\big({\overrightarrow Q}_{0,n}(\cdot|x^n){\times}U_{0,n}(\cdot)\big)}{d\big({\nu}_{0,n}(\cdot)\times{U}_{0,n}(\cdot)\big)}(y^n)\Big)({\overrightarrow Q}_{0,n}\otimes{\overleftarrow P}_{0,n})(dx^n,dy^{n})\nonumber\\
&=\int_{{\cal X}_{0,n}\times{\cal Y}_{0,n}}\log\Big(\frac{d{\overrightarrow Q}_{0,n}(\cdot|x^n)}{dU_{0,n}(\cdot)}(y^n)\Big)({\overrightarrow Q}_{0,n}\otimes{\overleftarrow P}_{0,n})(dx^n,dy^{n})\nonumber\\
&\qquad-\int_{{\cal X}_{0,n}\times{\cal Y}_{0,n}}\log\Big(\frac{d{\nu}_{0,n}(\cdot)}{dU_{0,n}(\cdot)}(y^n)\Big)({\overrightarrow Q}_{0,n}\otimes{\overleftarrow P}_{0,n})(dx^n,dy^{n})\nonumber\\
&=\int_{{\cal X}_{0,n}\times{\cal Y}_{0,n}}\log\Big(\frac{d{\overrightarrow Q}_{0,n}(\cdot|x^n)}{dU_{0,n}(dy^n)}(y^n)\Big)({\overrightarrow Q}_{0,n}\otimes{\overleftarrow P}_{0,n})(dx^n,dy^{n})\nonumber\\
&\qquad-\int_{{\cal Y}_{0,n}}\log\Big(\frac{d\nu_{0,n}(\cdot)}{dU_{0,n}(\cdot)}(y^n)\Big)\Bigg(\int_{{\cal X}_{0,n}}({\overrightarrow Q}_{0,n}\otimes{\overleftarrow P}_{0,n})(dx^n,dy^{n})\Bigg)\nonumber\\
&=\int_{{\cal X}_{0,n}\times{\cal Y}_{0,n}}\log\Big(\frac{d{\overrightarrow Q}_{0,n}(\cdot|x^n)}{dU_{0,n}(\cdot)}(y^n)\Big)({\overrightarrow Q}_{0,n}\otimes{\overleftarrow P}_{0,n})(dx^n,dy^{n})-\int_{{\cal Y}_{0,n}}\log\Big(\frac{d\nu_{0,n}(\cdot)}{dU_{0,n}(\cdot)}(y^n)\Big)\nu_{0,n}(dy^n).\nonumber
\end{align}
Hence,
\begin{align}
&{\mathbb{I}}_{X^n\rightarrow{Y^n}}\big(\lambda{\overleftarrow P}^1_{0,n}+(1-\lambda){\overleftarrow P}^2_{0,n},{\overrightarrow Q}_{0,n}\big)=\int_{{\cal X}_{0,n}\times{\cal Y}_{0,n}}\log\Big(\frac{d{\overrightarrow Q}_{0,n}(\cdot|x^n)}{dU_{0,n}(\cdot)}(y^n)\Big)\nonumber\\
&\times{\overrightarrow Q}_{0,n}(dy^n|x^n)\otimes\big(\lambda{\overleftarrow P}^1_{0,n}(dx^n|y^{n-1})+(1-\lambda){\overleftarrow P}^2_{0,n}(dx^n|y^{n-1})\big)-\int_{{\cal Y}_{0,n}}\log\Bigg(\frac{d\nu_{0,n}(\cdot)}{dU_{0,n}(\cdot)}(y^n)\Bigg)\nu_{0,n}(dy^n).\nonumber
\end{align}
Moreover, relative entropy is convex in both arguments \big{(}e.g., $\mathbb{D}(\cdot||U_{0,n})$ is convex for fixed $U_{0,n}$\big{)}, hence
\begin{align}
{\mathbb{I}}_{X^n\rightarrow{Y^n}}\big(\lambda{\overleftarrow P}^1_{0,n}+(1-\lambda){\overleftarrow P}^2_{0,n},{\overrightarrow Q}_{0,n}\big)&\geq\lambda\int_{{\cal X}_{0,n}\times{\cal Y}_{0,n}}\log\Big(\frac{d{\overrightarrow Q}_{0,n}(\cdot|x^n)}{dU_{0,n}(\cdot)}(y^n)\Big)({\overrightarrow Q}_{0,n}\otimes{\overleftarrow P}^1_{0,n})(dx^n,dy^{n})\nonumber\\
&-\lambda\int_{{\cal Y}_{0,n}}\log\Big(\frac{d\nu^1_{0,n}(\cdot)}{dU_{0,n}(\cdot)}(y^n)\Big)\nu^1_{0,n}(dy^n)\nonumber\\
&+(1-\lambda)\int_{{\cal X}_{0,n}\times{\cal Y}_{0,n}}\log\Big(\frac{d{\overrightarrow Q}_{0,n}(\cdot|x^n)}{dU_{0,n}(\cdot)}(y^n)\Big)({\overrightarrow Q}_{0,n}\otimes{\overleftarrow P}^2_{0,n})(dx^n,dy^{n})\nonumber\\
&-(1-\lambda)\int_{{\cal Y}_{0,n}}\log\Big(\frac{d\nu^2_{0,n}(\cdot)}{dU_{0,n}(\cdot)}\Big)\nu^2_{0,n}(dy^n).\nonumber
\end{align}
Finally, since $\nu^1_{0,n}(\cdot){\ll}U_{0,n}(\cdot)$ and $\nu^2_{0,n}(\cdot){\ll}U_{0,n}(\cdot)$ by substituting the following versions\\
$\frac{d\big(\overrightarrow{Q}_{0,n}(\cdot|x^n)\times\nu_{0,n}^i(\cdot)\big)}{d\big(U_{0,n}(\cdot)\times\nu_{0,n}^i(\cdot)\big)}(y^n)$,~$i=1,2$, of the RND for $\frac{d\overrightarrow{Q}_{0,n}(\cdot|x^n)}{dU_{0,n}(\cdot)}(y^n)$ in the first and third RHS expression in the preceding equations yields
\begin{align}
{\mathbb{I}}_{X^n\rightarrow{Y^n}}\big(\lambda{\overleftarrow P}^1_{0,n}+(1-\lambda){\overleftarrow P}^2_{0,n},{\overrightarrow Q}_{0,n}\big)\geq\lambda{\mathbb{I}}_{X^n\rightarrow{Y^n}}\big({\overleftarrow P}^1_{0,n},{\overrightarrow Q}_{0,n}\big)+(1-\lambda){\mathbb{I}}_{X^n\rightarrow{Y^n}}\big({\overleftarrow P}^2_{0,n},{\overrightarrow Q}_{0,n}\big).\nonumber
\end{align}
This completes the derivation of 2).\\
3) Here, it will be shown that for ${\overrightarrow Q}_{0,n}^1(\cdot|x^n),$ ${\overrightarrow Q}_{0,n}^2(\cdot|x^n)\in{\cal M}_1^{\bf C2}({\cal Y}_{0,n})$ such that ${\overrightarrow Q}_{0,n}^1(\cdot|x^n)\neq{\overrightarrow Q}_{0,n}^2(\cdot|x^n)$, and $\lambda\in(0,1)$, then ${\mathbb{I}}_{X^n\rightarrow{Y^n}}\big({\overleftarrow P}_{0,n},\lambda{\overrightarrow Q}_{0,n}^1+(1-\lambda){\overrightarrow Q}_{0,n}\big)<\lambda{\mathbb{I}}_{X^n\rightarrow{Y^n}}\big({\overleftarrow P}_{0,n},{\overrightarrow Q}_{0,n}^1\big)+(1-\lambda){\mathbb{I}}_{X^n\rightarrow{Y^n}}\big({\overleftarrow P}_{0,n}$, ${\overrightarrow Q}_{0,n}^2\big)$, for a fixed~${\overleftarrow P}_{0,n}(\cdot|y^{n-1})\in{\cal M}_1^{\bf C1}({\cal X}_{0,n})$.\\
It is already known that ${\mathbb{I}}_{X^n\rightarrow{Y^n}}({\overleftarrow P}_{0,n},{\overrightarrow Q}_{0,n})$ is a convex functional on ${\overrightarrow Q}_{0,n}(\cdot|x^n)\in{\cal M}_1^{\bf C2}({\cal Y}_{0,n})$ for a fixed ${\overleftarrow P}_{0,n}(\cdot|y^{n-1})\in{\cal M}_1^{\bf C1}({\cal X}_{0,n})$. All is required to show in order to have strict convexity is that ${\mathbb{I}}_{X^n\rightarrow{Y^n}}({\overleftarrow P}_{0,n},{\overrightarrow Q}_{0,n})<\infty$.
This can be easily obtained from part 1) since ${\overleftarrow P}_{0,n}\otimes{\overrightarrow Q}_{0,n}\ll{\overleftarrow P}_{0,n}\otimes\nu_{0,n}$ if and only if ${\overrightarrow Q}_{0,n}(\cdot|x^n)\ll\nu_{0,n}(\cdot),$ for $\mu_{0,n}-$almost all $x^n\in{\cal X}_{0,n}.$ Hence, from the strict convexity of the function $s{\log}s,~s\in[0,\infty)$, and the expression of directed information as a functional of $\{{\overleftarrow P}_{0,n}(\cdot|y^{n-1}),{\overrightarrow Q}_{0,n}(\cdot|x^n)\}\in{\cal M}_1^{\bf C1}({\cal X}_{0,n})\times{\cal M}_1^{\bf C2}({\cal Y}_{0,n})$, with $\overrightarrow{Q}_{0,n}(\cdot|x^n)=\lambda\overrightarrow{Q}_{0,n}^1(\cdot|x^n)+(1-\lambda){\overrightarrow Q}_{0,n}^2(\cdot|x^n)$ it follows that
\begin{align}
{\mathbb{I}}_{X^n\rightarrow{Y^n}}({\overleftarrow P}_{0,n},\overrightarrow{Q}_{0,n})&=\int_{{\cal X}_{0,n}\times{\cal Y}_{0,n}}\log\Big(\frac{d({\overleftarrow P}_{0,n}\otimes{\overrightarrow Q}_{0,n})(\cdot,\cdot)}{{\overrightarrow\Pi}(\cdot,\cdot)}(x^n,y^n)\Big)({\overrightarrow Q}_{0,n}\otimes{\overleftarrow P}_{0,n})(dx^n,dy^n)\nonumber\\
&=\int_{{\cal X}_{0,n}\times{\cal Y}_{0,n}}\log\Big(\frac{d{\overrightarrow Q}_{0,n}(\cdot|x^n)}{d{\nu}_{0,n}(\cdot)}(y^n)\Big)({\overrightarrow Q}_{0,n}\otimes{\overleftarrow P}_{0,n})(dx^n,dy^n)\nonumber\\
&\leq\int_{{\cal X}_{0,n}\times{\cal Y}_{0,n}}\lambda\log\Big(\frac{d{\overrightarrow Q}^1_{0,n}(\cdot|x^n)}{d{\nu}_{0,n}(\cdot)}(y^n)\Big)({\overrightarrow Q}_{0,n}^1\otimes{\overleftarrow P}_{0,n})(dx^n,dy^n)\nonumber\\
&\qquad+\int_{{\cal X}_{0,n}\times{\cal Y}_{0,n}}(1-\lambda)\log\Big(\frac{d{\overrightarrow Q}^2_{0,n}(\cdot|x^n)}{d{\nu}_{0,n}(\cdot)}(y^n)\Big)({\overrightarrow Q}_{0,n}^2\otimes{\overleftarrow P}_{0,n})(dx^n,dy^n)\nonumber\\
&=\lambda{\mathbb{I}}_{X^n\rightarrow{Y^n}}\big({\overleftarrow P}_{0,n},{\overrightarrow Q}_{0,n}^1\big)+(1-\lambda){\mathbb{I}}_{X^n\rightarrow{Y^n}}\big({\overleftarrow P}_{0,n},{\overrightarrow Q}_{0,n}^2\big)<\infty.\nonumber
\end{align}
This completes the derivation of 3).

\section{Proof of Theorem~\ref{weak_convergence}}\label{proof_weak_convergence}

{\bf Part~A.} Let $\overrightarrow{Q}^{\alpha}_{0,n}(\cdot|\cdot)\in{\cal Q}^{\bf C2}({\cal Y}_{0,n}|{\cal X}_{0,n}),~\alpha=1,2,\ldots$, be a sequence of forward channels and $(X^{n,(\alpha)},Y^{n,(\alpha)}),~\alpha=1,2,\ldots$ a sequence of the basic joint process corresponding to the backward channel ${\overleftarrow P}_{0,n}(\cdot|\cdot)\in{\cal Q}^{\bf C1}({\cal X}_{0,n}|{\cal Y}_{0,n-1})$ and the sequence of forward channels ${\overrightarrow Q}^{\alpha}_{0,n}(\cdot|\cdot)\in{\cal Q}^{\bf C2}({\cal Y}_{0,n}|{\cal X}_{0,n}),~\alpha=1,2,\ldots$. The important steps for the derivation of A1) are outlined in \cite{gihman-skorohod1979} for stochastic control problems with randomized controls. Since we shall use A1) and parts of its derivation to show A2)--A4), we give the details of the derivation. 
\vspace*{0.2cm}\\
A1) First, it is shown that the joint distribution of the basic joint process $\{(X^{(\alpha)}_i,Y^{(\alpha)}_i):i\in\mathbb{N}_0\}$ converges as $\alpha\longrightarrow\infty$ to the joint distribution of a joint process $\{(X^{(o)}_i,Y^{(o)}_i):i\in\mathbb{N}_0\}$ and secondly, that this limiting joint process $\{(X^{(o)}_i,Y^{(o)}_i):i\in\mathbb{N}_0\}$ is also a basic joint process corresponding to the backward channel ${\overleftarrow P}_{0,n}(\cdot|\cdot)\in{\cal Q}^{\bf C1}({\cal X}_{0,n}|{\cal Y}_{0,n-1})$, that is, $({\overleftarrow{P}}_{0,n}\otimes{\overrightarrow{Q}}_{0,n}^{\alpha})(dx^n,dy^n)\buildrel w \over\longrightarrow ({\overleftarrow{P}}_{0,n}\otimes{\bar{Q}}_{0,n}^o)(dx^n,dy^n)\in{\cal M}_1({\cal X}_{0,n}\times{\cal Y}_{0,n}) $ and that $({\overleftarrow{P}}_{0,n}\otimes{\bar{Q}}_{0,n}^o)(dx^n,dy^n)$ has backward channel ${\overleftarrow{P}}_{0,n}(\cdot|\cdot)\in{\cal Q}^{\bf C1}({\cal X}_{0,n}|{\cal Y}_{0,n-1})$, but $\bar{Q}^o_{0,n}(\cdot|x^n)\in{\cal M}_1({\cal Y}_{0,n})$ is not necessarily an element of ${\cal M}_1^{\bf C2}({\cal Y}_{0,n})$.\\
\noi For any $g(\cdot)\in{BC}({\cal X}_n)$, by condition {\bf CA}, the function
\begin{align*}
f:~{\cal X}_{0,n-1}\times{\cal Y}_{0,n-1}\longmapsto\mathbb{R},~f(x^{n-1},y^{n-1})\tri\int_{{\cal X}_n}g(x)p_n(dx_n|x^{n-1},y^{n-1})
\end{align*}
is continuous, and hence for any compact sets $K_i\in{\cal X}_i,~i=0,1,\ldots,n-1$, and by the compactness of ${\cal Y}_{0,n-1}$, the image of $f(\cdot,\cdot)$ under $K_{0,n-1}\times{\cal Y}_{0,n-1}\tri{K}_0\times{K}_1\times\ldots\times{K}_{n-1}\times{\cal Y}_{0,n-1}$, $f(K_{0,n-1}\times{\cal Y}_{0,n-1})={\cal R}\subset\mathbb{R}$, and ${\cal R}$ is compact (since the image of any real-valued continuous function on a compact set is compact). Thus, by condition ${\bf CA}$ and the compactness of $\{{\cal Y}_i:~i\in\mathbb{N}_0^n\}$, for any compact sets $K_0\in{\cal X}_0, K_1\in{\cal X}_{1},\ldots,K_{n-1}\in{\cal X}_{n-1}$ the family of distributions $\{p_n(\cdot|x^{n-1},y^{n-1}):x_0{\in}K_0,x_1{\in}K_1,\ldots,x_{n-1}{\in}K_{n-1},y^{n-1}\in{\cal Y}_{0,n-1}\}$ is compact. Indeed, given any sequence $\{x_0^{(\alpha)},\ldots,x_{n-1}^{(\alpha)},$ $y_0^{(\alpha)},\ldots,y_{n-1}^{(\alpha)}\}$, by selecting a subsequence $\alpha_i$ such that the subsequence $\{x_0^{(\alpha_i)},\ldots,x_{n-1}^{(\alpha_i)},$ $y_0^{(\alpha_i)},\ldots,y_{n-1}^{(\alpha_i)}\}$ converges  to $\{x_0^{(o)},\ldots,x_{n-1}^{(o)},y_0^{(o)},\ldots,y_{n-1}^{(o)}\}$, a weakly convergent subsequence of measures \\$p_n(\cdot|x_0^{(\alpha_i)},\ldots,x_{n-1}^{(\alpha_i)},y_0^{(\alpha_i)},\ldots,y_{n-1}^{(\alpha_i)})$ is obtained. Utilizing Prohorov's theorem (see Theorem~\ref{prohorov}), we verify that for any sequence of compact sets $K_0\subset{\cal X}_0,K_1\subset{\cal X}_1,\ldots,K_{n-1}\subset{\cal X}_{n-1},$ and $\epsilon_1>0$ a compact set $K_n\subset{\cal X}_n$ can be constructed such that $p_n(K_n|x^{n-1},y^{n-1})\geq1-\epsilon_1$, for any $y^{n-1}\in{\cal Y}_{0,n-1}$. To this end, pick $\epsilon_1>0$ and construct the compact sets as follows. Choose compact set $K_0\subset{\cal X}_0$ such that $p_0(K_0)\geq1-\frac{\epsilon_1}{2},$ compact set $K_1\subset{\cal X}_1$ such that $p_1(K_1|x_0,y_0)\geq1-\frac{\epsilon_1}{2^2}$, for any $x_0\in{K}_0,y_0\in{\cal Y}_0$, compact set $K_2\subset{\cal X}_2$ such that $p_2(K_2|x_0,x_1,y_0,y_1)\geq1-\frac{\epsilon_1}{2^3},$ for any $x_0\in{K}_0,x_1\in{K}_1,y_0\in{\cal Y}_0,y_1\in{\cal Y}_1$, and compact set $K_n$ such that 
\begin{align}
p_n(K_n|x^{n-1},y^{n-1})\geq1-\frac{\epsilon_1}{2^{n+1}}.\label{equation39}
\end{align}
Utilizing (\ref{equation39}) then
\begin{align}
\mathbb{P}&\Big\{X_0^{(\alpha)}\in{K}_0,\ldots,X_n^{(\alpha)}\in{K}_n\Big\}=\mathbb{P}\Big\{X_0^{(\alpha)}\in{K}_0,\ldots,X_{n}^{(\alpha)}\in{K}_{n},Y_0^{(\alpha)}\in{\cal Y}_{0},\ldots,Y_{n-1}^{(\alpha)}\in{\cal Y}_{n-1}\Big\}\nonumber\\
&=\int_{\times_{i=0}^n{K}_i}\int_{{\cal Y}_{0,n-1}}\mathbb{P}\Big\{X_n^{(\alpha)}\in{K}_n|X_0^{(\alpha)}=x_0,\ldots,X_{n-1}^{(\alpha)}=x_{n-1},Y_0^{(\alpha)}=y_0,\ldots,Y_{n-1}^{(\alpha)}=y_{n-1}\Big\}\nonumber\\
&\qquad\mathbb{P}\Big\{X_0^{(\alpha)}\in{dx}_0,\ldots,X_{n-1}^{(\alpha)}\in{dx}_{n-1},Y_0^{(\alpha)}\in{dy}_0,\ldots,Y_{n-1}^{(\alpha)}\in{dy}_{n-1}\Big\}\nonumber\\
&\geq\Bigg(1-\frac{\epsilon_1}{2^{n+1}}\Bigg)\int_{\times_{i=0}^{n-1}{K}_i}\mathbb{P}\Big\{X_0^{(\alpha)}\in{dx}_0,\ldots,X_{n-1}^{(\alpha)}\in{dx}_{n-1}\Big\}\nonumber\\
&=\Bigg(1-\frac{\epsilon_1}{2^{n+1}}\Bigg)\mathbb{P}\Big\{X_0^{(\alpha)}\in{K}_0,\ldots,X_{n-1}^{(\alpha)}\in{K}_{n-1}\Big\}\nonumber\\
&\geq\Bigg(1-\frac{\epsilon_1}{2^{n+1}}\Bigg)\Bigg(1-\frac{\epsilon_1}{2^{n}}\Bigg)\mathbb{P}\Big\{X_0^{(\alpha)}\in{K}_0,\ldots,X_{n-2}^{(\alpha)}\in{K}_{n-2}\Big\}\nonumber\\
&=\Bigg(1-\frac{\epsilon_1}{2^{n+1}}-\frac{\epsilon_1}{2^{n}}+\frac{\epsilon^2_1}{2^{2n+1}}\Bigg)\mathbb{P}\Big\{X_0^{(\alpha)}\in{K}_0,\ldots,X_{n-2}^{(\alpha)}\in{K}_{n-2}\Big\}\nonumber\\
&\geq\Bigg(1-\frac{\epsilon_1}{2^{n+1}}-\frac{\epsilon_1}{2^{n}}\Bigg)\mathbb{P}\Big\{X_0^{(\alpha)}\in{K}_0,\ldots,X_{n-2}^{(\alpha)}\in{K}_{n-2}\Big\}\nonumber
\end{align}
\begin{align}
&\geq\Bigg(1-\frac{\epsilon_1}{2^{n+1}}-\frac{\epsilon_1}{2^{n}}\Bigg)\Bigg(1-\frac{\epsilon_1}{2^{n-1}}\Bigg)
\mathbb{P}\Big\{X_0^{(\alpha)}\in{K}_0,\ldots,X_{n-3}^{(\alpha)}\in{K}_{n-3}\Big\}\nonumber\\
&=\Bigg(1-\frac{\epsilon_1}{2^{n+1}}-\frac{\epsilon_1}{2^{n}}-\frac{\epsilon_1}{2^{n-1}}+\frac{\epsilon^2_1}{2^{2n}}+\frac{\epsilon^2_1}{2^{2n-1}}\Bigg)
\mathbb{P}\Big\{X_0^{(\alpha)}\in{K}_0,\ldots,X_{n-3}^{(\alpha)}\in{K}_{n-3}\Big\}\nonumber\\
&\geq\Bigg(1-\frac{\epsilon_1}{2^{n+1}}-\frac{\epsilon_1}{2^{n}}-\frac{\epsilon_1}{2^{n-1}}\Bigg)
\mathbb{P}\Big\{X_0^{(\alpha)}\in{K}_0,\ldots,X_{n-3}^{(\alpha)}\in{K}_{n-3}\Big\}.\label{appendix:proof:weak_convergence:eq1}
\end{align}
Iterating the RHS of (\ref{appendix:proof:weak_convergence:eq1}) we obtain
\begin{align}
\mathbb{P}\Big\{X_0^{(\alpha)}\in{K}_0,\ldots,X_n^{(\alpha)}\in{K}_n\Big\}&\geq1-\frac{\epsilon_1}{2^{n+1}}-\frac{\epsilon_1}{2^{n}}-\frac{\epsilon_1}{2^{n-1}}-\ldots-\frac{\epsilon_1}{2^1}=1-\epsilon_1\sum^n_{i=1}\frac{1}{2^{i+1}}\nonumber\\
&\geq1-\epsilon_1,\hst~\mbox{for all}~{\alpha=1,2,\ldots},~\mbox{and any}~n\in\mathbb{N}_0.\label{appendix:proof:weak_convergence:eq2}
\end{align}
By (\ref{appendix:proof:weak_convergence:eq2}), the family of marginal distributions of the joint process $\{(X_i^{(\alpha)},Y_i^{(\alpha)}):i\in\mathbb{N}_0\},~\alpha=1,2,\ldots$ on ${\cal X}_{0,n}$ is uniformly tight, and by Prohorov's theorem \cite{parthasarathy1967} it has a weakly convergent subsequence. On the other hand, since $\{{\cal Y}_i:~i\in\mathbb{N}_0^n\}$ are compact metric spaces, the family of marginal distributions of the joint sequence $\{(X_i^{(\alpha)},Y_i^{(\alpha)}):i\in\mathbb{N}_0\}$ on ${\cal Y}_{0,n}$ is uniformly tight. Utilizing the uniform tightness of the marginal distribution of the joint process $\{(X_i^{(\alpha)},Y_i^{(\alpha)}):i\in\mathbb{N}_0\}$, then the family of joint distributions of the joint process $\{(X_i^{(\alpha)},Y_i^{(\alpha)}):i\in\mathbb{N}_0\}$ is uniformly tight. By Prohorov's theorem \cite{parthasarathy1967}, the sequence of joint distribution of the joint process $\{(X_i^{(\alpha)},Y_i^{(\alpha)}):i\in\mathbb{N}_0\}$ possess a weakly convergent subsequence to a joint process $\{(X_i^{(o)},Y_i^{(o)}):i\in\mathbb{N}_0\}$. A restatement of Prohorov's theorem states that, if ${\cal Z}$ is a separable metric space then every uniformly tight sequence of measures $\{\gamma^{\alpha}:~\alpha=1,2,\ldots\}$ on ${\cal Z}$ has a subsubsequence which is weakly convergent. Moreover, by \cite{parthasarathy1967}, if each subsequence $\{\gamma^{\alpha_i}:~i=1,2,\ldots\}$ of $\{\gamma^{\alpha}:~\alpha=1,2,\ldots\}$ contains a further subsequence $\{\gamma^{\alpha_{i_m}}:~m=1,2,\ldots\}$ such that $\gamma^{\alpha_{i_m}}\buildrel w \over\longrightarrow\gamma^o$ as $m\longrightarrow\infty$, then $\gamma^{\alpha}\buildrel w \over\longrightarrow\gamma^o$ as $\alpha\longrightarrow\infty$. Utilizing these facts, then the joint distribution of the joint process $\{(X^{(\alpha)}_i,Y^{(\alpha)}_i):i\in\mathbb{N}_0\}$ converges weakly to 
a joint process $\{(X^{(o)}_i,Y^{(o)}_i):i\in\mathbb{N}_0\}$. Next, we show that the limiting joint process $\{(X_i^{(o)},Y_i^{(o)}):i\in\mathbb{N}_0\}$ is a basic joint process with the same backward channel ${\overleftarrow P}(\cdot|\cdot)\in{\cal Q}^{\bf C1}({\cal X}_{0,n}|{\cal Y}_{0,n-1}).$
For any $n\in\mathbb{N}_0$, consider bounded and continuous real-valued functions $g_n(\cdot)\in{BC}({\cal X}_n)$ and $\Psi_{0,n-1}(\cdot,\cdot)\in{BC}({\cal X}_{0,n-1}\times{\cal Y}_{0,n-1}).$ By the weak convergence of the joint measures corresponding to $\{(X_i^{(\alpha)},Y_i^{(\alpha)}):i\in\mathbb{N}_0\}$ to the joint measures corresponding to $\{(X_i^{(o)},Y_i^{(o)}):i\in\mathbb{N}_0\}$ denoted by $(\overleftarrow{P}_{0,n}\otimes\overrightarrow{Q}_{0,n}^{\alpha})(dx^n,dy^n)\buildrel w \over\longrightarrow{P}_{0,n}^{o}(dx^n,dy^n)$, the continuity of $g_n(\cdot)$ and the continuity of the function mapping $(x^{n-1},y^{n-1})\in{\cal X}_{0,n-1}\times{\cal Y}_{0,n-1}\longmapsto\int_{{\cal X}_n}g_n(x)p_n(dx|x^{n-1},y^{n-1})\in\mathbb{R}$, given $\epsilon>0$ there exists $N\in\mathbb{N}_0$ such that for all $\alpha\geq{N}$
\begin{align}
\Bigg|&\int_{{\cal X}_{0,n-1}\times{\cal Y}_{0,n-1}}\Bigg(\int_{{\cal X}_n}g_n(x)p_n(dx|x^{n-1},y^{n-1})\Bigg)\Psi_{0,n-1}(x^{n-1},y^{n-1}){P}_{0,n-1}^{o}(dx^{n-1},dy^{n-1})\nonumber\\
&-\int_{{\cal X}_{0,n-1}\times{\cal Y}_{0,n-1}}\Bigg(\int_{{\cal X}_n}g_n(x)p_n(dx|x^{n-1},y^{n-1})\Bigg)\Psi_{0,n-1}(x^{n-1},y^{n-1}){P}_{0,n-1}^{\alpha}(dx^{n-1},dy^{n-1})\Bigg|\leq\epsilon.\nonumber
\end{align}
Since $\epsilon>0$ is arbitrary, then
\begin{align}
\lim_{\alpha\rightarrow\infty}{\mathbb{E}}\bigg\{g_n(X_n^{(\alpha)})&\Psi(X_0^{(\alpha)},\ldots,X_{n-1}^{(\alpha)},Y_0^{(\alpha)},\ldots,Y_{n-1}^{(\alpha)})\bigg\}={\mathbb{E}}\bigg\{g_n(X_n^{(o)})\Psi(X_0^{(o)},\ldots,X_{n-1}^{(o)},Y_0^{(o)},\ldots,Y_{n-1}^{(o)})\bigg\}.\label{equation4c}
\end{align}
Moreover, for all $\alpha=1,2,\ldots$, then
\begin{align}
{\mathbb{E}}&\bigg\{g_n(X_n^{(\alpha)})\Psi(X_0^{(\alpha)},\ldots,X_{n-1}^{(\alpha)},Y_0^{(\alpha)},\ldots,Y_{n-1}^{(\alpha)})\bigg\}\nonumber\\
&={\mathbb{E}}\bigg\{\Psi(X_0^{(\alpha)},\ldots,X_{n-1}^{(\alpha)},Y_0^{(\alpha)},\ldots,Y_{n-1}^{(\alpha)})\mathbb{E}\bigg\{g_n(X_n^{(\alpha)})|
X_0^{(\alpha)},\ldots,X_{n-1}^{(\alpha)},Y_0^{(\alpha)},\ldots,Y_{n-1}^{(\alpha)})\bigg\}\bigg\}\nonumber\\
&={\mathbb{E}}\bigg\{\bigg(\int_{{\cal X}_n}g_n(x)p_n(dx|X_0^{(\alpha)},\ldots,X_{n-1}^{(\alpha)},Y_0^{(\alpha)},\ldots,Y_{n-1}^{(\alpha)})\bigg)\Psi(X_0^{(\alpha)},\ldots,X_{n-1}^{(\alpha)},
Y_0^{(\alpha)},\ldots,Y_{n-1}^{(\alpha)})\bigg\}.\nonumber
\end{align}
Hence, (\ref{equation4c}) is equivalent to
\begin{align}
\lim_{\alpha\rightarrow\infty}&{\mathbb{E}}\bigg\{\int_{{\cal X}_n}g_n(x)p_n(dx|X_0^{(\alpha)},\ldots,X_{n-1}^{(\alpha)},Y_0^{(\alpha)},\ldots,Y_{n-1}^{(\alpha)})\Psi(X_0^{(\alpha)},\ldots,X_{n-1}^{(\alpha)},Y_0^{(\alpha)},\ldots,Y_{n-1}^{(\alpha)})\bigg\}\nonumber\\
&={\mathbb{E}}\bigg\{\int_{{\cal X}_{n}}g_n(x)p_n(dx|X_0^{(o)},\ldots,X_{n-1}^{(o)},Y_0^{(o)},\ldots,Y_{n-1}^{0})\Psi(X_0^{(o)},\ldots,X_{n-1}^{(o)},Y_0^{(o)},\ldots,Y_{n-1}^{(o)})\bigg\}.\nonumber
\end{align}
From the previous equality, the following identity is obtained.
\begin{align}
{\mathbb{E}}&\bigg\{g_n(X_n^{(o)})|X_0^{(o)},\ldots,X_{n-1}^{(o)},Y_0^{(o)},\ldots,Y_{n-1}^{(o)})\bigg\}=\int_{{\cal X}_n}g_n(x)p_n(dx|X_0^{(o)},\ldots,X_{n-1}^{(o)},Y_0^{(o)},\ldots,Y_{n-1}^{(o)})-a.s.\label{appendix:proof:weak_convergence:eq3}
\end{align}
Since for any indicator function $I_{E},$ $E\in{\cal B}({\cal X}_n)$ there exists a sequence $\{g_{n,j}:~j=1,2,\ldots\}\subset{BC}({\cal X}_n)$ which is nondecreasing  such that $g_{n,j}\uparrow{I}_E$, by utilizing such a sequence in (\ref{appendix:proof:weak_convergence:eq3}), and by invoking Lebesgue's monotone convergence theorem then
\begin{align}
\mathbb{P}\bigg\{X_n^{(o)}\in{E}|X_0^{(o)},\ldots,X_{n-1}^{(o)},Y_0^{(o)},\ldots,Y_{n-1}^{(o)}\bigg\}=p_n(E|X_0^{(o)},\ldots,X_{n-1}^{(o)},Y_0^{(o)},\ldots,Y_{n-1}^{(o)}).\nonumber
\end{align}
This shows that the limiting joint process $\{(X_i^{(o)},Y_i^{(o)}):i\in\mathbb{N}_0\}$ is a basic process corresponding to the backward channel ${\overleftarrow P}_{0,n}(\cdot|y^{n-1})\in{\cal M}_1^{\bf C1}({\cal X}_{0,n})$ and a forward channel $\bar{Q}_{0,n}^o(\cdot|x^n)\in{\cal M}_1({\cal Y}_{0,n})$. Moreover, the marginal distributions of the basic joint process $\{(X_i^{(\alpha)},Y_i^{(\alpha)}):i\in\mathbb{N}_0\}$ converge to the marginal distributions of the basic joint process $\{(X_i^{(o)},Y_i^{(o)}):i\in\mathbb{N}_0\}$ corresponding to the backward channel ${\overleftarrow P}_{0,n}(\cdot|y^{n-1})\in{\cal M}_1^{\bf C1}({\cal X}_{0,n})$ and a forward channel $\bar{Q}^o_{0,n}(\cdot|x^{n})\in{\cal M}_1({\cal Y}_{0,n}).$ This completes the derivation of A1).
\vspace*{0.2cm}\\
\noi A2) By consistency condition ${\bf C1}$, any ${\overleftarrow P}_{0,n}(\cdot|\cdot)\in{\cal Q}^{\bf C1}({\cal X}_{0,n}|{\cal Y}_{0,n-1})$ uniquely defines a family $\{p_i(\cdot|\cdot,\cdot)\in{\cal Q}({\cal X}_i|{\cal X}_{0,i-1}\times{\cal Y}_{0,i-1}),i\in\mathbb{N}_0^n\}$ via (\ref{equation2}). Hence, (\ref{equation2}) can be used to relate tightness of  $p_i(\cdot|x^{i-1},y^{i-1})\in{\cal M}_1({\cal X}_i), ~(x^{i-1},y^{i-1})\in{\cal X}_{0,i-1}\times{\cal Y}_{0,i-1},~i\in\mathbb{N}_0^n$, to tightness  of $\overleftarrow{P}_{0,n}(\cdot|y^{n-1})\in{\cal M}_1^{\bf C1}({\cal X}_{0,n}), ~y^{n-1}\in{\cal Y}_{0,n-1}$. \\
By recalling the derivation A1), condition (\ref{equation39}), for $K_{0,n}=\times_{i=0}^n{K_i},$ $K_i\in{\cal B}({\cal X}_i)$ compact sets, $i\in\mathbb{N}_0^n$, then
\begin{align}
{\bf P}(K_{0,n}|{\bf y})&\tri\int_{K_0}p_0(dx_0)\int_{K_1}p_1(dx_1|x^{0},y^0)\ldots\int_{K_n}p_n(dx_n|x^{n-1},y^{n-1})\nonumber\\
&\geq\bigg(1-\frac{\epsilon_1}{2^{n+1}}\bigg)\int_{K_0}p_0(dx_0)\int_{K_1}p_1(dx_1|x^{0},y^0)\ldots\int_{K_{n-1}}p_{n-1}(dx_{n-1}|x^{n-2},y^{n-2})\nonumber\\
&\geq\bigg(1-\frac{\epsilon_1}{2^{n+1}}\bigg)\bigg(1-\frac{\epsilon_1}{2^{n}}\bigg)\int_{K_0}p_0(dx_0)\int_{K_1}p_1(dx_1|x^{0},y^0)\ldots\int_{K_{n-2}}p_{n-2}(dx_{n-2}|x^{n-3},y^{n-3})\nonumber\\
&=\Bigg(1-\frac{\epsilon_1}{2^{n+1}}-\frac{\epsilon_1}{2^{n}}+\frac{\epsilon^2_1}{2^{2n+1}}\Bigg)\int_{K_0}p_0(dx_0)\int_{K_1}p_1(dx_1|x^{0},y^0)\ldots\int_{K_{n-2}}p_{n-2}(dx_{n-2}|x^{n-3},y^{n-3})\nonumber\\
&\geq\Bigg(1-\frac{\epsilon_1}{2^{n+1}}-\frac{\epsilon_1}{2^{n}}\Bigg)\int_{K_0}p_0(dx_0)\int_{K_1}p_1(dx_1|x^{0},y^0)\ldots\int_{K_{n-2}}p_{n-2}(dx_{n-2}|x^{n-3},y^{n-3}).\nonumber
\end{align}
By repeating the above procedure the following bound is obtained.
\begin{align}
{\bf P}({K}_{0,n}|{\bf y})&\geq1-\frac{\epsilon_1}{2^{n+1}}-\frac{\epsilon_1}{2^{n}}-\frac{\epsilon_1}{2^{n-1}}-\ldots-\frac{\epsilon_1}{2^1}=1-\epsilon_1\sum^n_{i=1}\frac{1}{2^{i+1}}\nonumber\\
&\geq1-\epsilon_1\nonumber,\hst~\mbox{for any}~n\in\mathbb{N}_0~\mbox{and for every}~{\bf y}\in{\cal Y}^{\mathbb{N}_0}.
\end{align}
Since $\{K_i:i=0,1,\ldots,n\}$ are compact, from the last inequality it follows that the family of measures $\overleftarrow{P}_{0,n}(\cdot|y^{n-1})\in{\cal M}_1^{\bf C1}({\cal X}_{0,n}), ~y^{n-1}\in{\cal Y}_{0,n-1}$ is uniformly tight. This completes the derivation of A2).
\vspace*{0.2cm}\\
\noi A3) Weak compactness of the family of measures ${\overrightarrow Q}_{0,n}(\cdot|x^{n})\in{\cal M}_1^{\bf C2}({\cal Y}_{0,n})$ for fixed $x^n\in{\cal X}_{0,n}$ follows from the fact that ${\cal Y}_{0,n}$ is a compact Polish space.
\vspace*{0.2cm}\\
\noi A4) Utilizing the weak convergence  $\nu_{0,n}^{\alpha}\buildrel w \over \longrightarrow\nu_{0,n}^o$ \big{(}shown in A2)\big{)}, we shall show weak convergence of the convolution of measures $\overrightarrow{\Pi}^{\alpha}_{0,n}(dx^n,dy^n)\equiv{\overleftarrow P}_{0,n}(dx^n|y^{n-1})\otimes\nu_{0,n}^{\alpha}(dy^n)\buildrel w \over \longrightarrow{\overleftarrow P}_{0,n}(dx^n|y^{n-1})\otimes\nu_{0,n}^o(dy^n)\equiv\overrightarrow{\Pi}_{0,n}^o(dx^n,dy^n),$ when ${\overleftarrow P}_{0,n}(\cdot|y^{n-1})\in{\cal M}_1^{\bf C1}({\cal X}_{0,n})$ is fixed. We show weak convergence by considering integrals with respect to $g_{0,n}(x^n)h_{0,n}(y^n),$ where $g_{0,n}(\cdot)\in{BC}({\cal X}_{0,n})$ and $h_{0,n}(\cdot)\in{BC}({\cal Y}_{0,n}).$ Let $\epsilon>0$ be given. Condition ${\bf CA}$ implies that the function mapping 
\begin{align}
y^{n-1}\in{\cal Y}_{0,n-1}\longmapsto\int_{{\cal X}_{0,n}}g(x^n){\overleftarrow P}_{0,n}(dx^n|y^{n-1})\in\mathbb{R}\label{equation40}
\end{align}
is continuous.
Hence, by the weak convergence $\nu_{0,n}^{\alpha}\buildrel w \over \longrightarrow\nu_{0,n}^o$ and the continuity of the function mapping (\ref{equation40}) then there exists $N\in\mathbb{N}_0$ such that for all $\alpha\geq{N}$
\begin{align}
\Bigg|&\int_{{\cal Y}_{0,n}}\bigg(\int_{{\cal X}_{0,n}}g(x^n){\overleftarrow P}_{0,n}(dx^n|y^{n-1})\bigg)h(y^n)\nu_{0,n}^o(dy^n)-\int_{{\cal Y}_{0,n}}\bigg(\int_{{\cal X}_{0,n}}g(x^n){\overleftarrow P}_{0,n}(dx^n|y^{n-1})\bigg)h(y^n)\nu_{0,n}^{\alpha}(dy^n)\Bigg|\nonumber\\
&\leq\epsilon.\nonumber
\end{align}
Since $\epsilon>0$ is arbitrary, then the derivation of A5) is complete.
\vspace*{0.2cm}\\
\noi {\bf Part B}. The methodology is similar to that of {\bf Part A.}, hence it is omitted.
%%%%%%%%%%%%%%%%%%%%%%%%%%%%%%%%%%%%%%%%%%%%%%%%%%%%%%%%%%%%%%%%%%%%%

\section{Proof of Lemma~\ref{capacity:closedness:lemma}}\label{section:proof:closedness:feedback:capacity}

\par By Theorem~\ref{weak_convergence}, {\bf Part A.}, A2), the family of measures $\overleftarrow{P}_{0,n}(\cdot|y^{n-1})\in{\cal M}_1^{\bf C1}({\cal X}_{0,n}), y^{n-1} \in {\cal Y}_{0,n-1}$ are tight, and by Appendix~\ref{proof_weak_convergence}, (\ref{equation39}), $\{p_i(\cdot|x^{i-1},y^{i-1})\in{\cal M}_1^{\bf C1}({\cal X}_i):~i=0,1,\ldots,n\}$ are tight. Since $p_i(\cdot|x^{i-1},y^{i-1})$ are probability measures on ${\cal M}_1^{\bf C1}({\cal X}_i)$, $i=0,1,\ldots,n$, for any sequence $\overleftarrow{P}_{0,n}^{\alpha}(\cdot|y^{n-1})\in{\cal M}_1^{\bf C1}({\cal X}_{0,n}),~\alpha=1,2,\ldots,$ there is a collection $\{p_i^{\alpha}(\cdot|x^{i-1},y^{i-1}):~i=0,1,\ldots,n\},~\alpha=1,2,\ldots$, such that
\begin{align*}
p_i^{\alpha}(\cdot|x^{i-1},y^{i-1})\buildrel w \over \longrightarrow{p}_i^{o}(\cdot|x^{i-1},y^{i-1}),~i=0,1,\ldots,n.
\end{align*}
\noi Hence, to show closedness of $\overleftarrow{P}_{0,n}(\cdot|y^{n-1})\in{\cal M}_1^{\bf C1}({\cal X}_{0,n}), y^{n-1} \in {\cal Y}_{0,n-1}$ it suffices to show that 
\begin{align*}
\otimes_{i=0}^np_i^{\alpha}(\cdot|x^{i-1},y^{i-1})\buildrel w \over \longrightarrow\otimes_{i=0}^n{p}_i^{o}(\cdot|x^{i-1},y^{i-1})
\end{align*}
whenever $p_i^{\alpha}(\cdot|x^{i-1},y^{i-1})\buildrel w \over \longrightarrow{p}_i^{o}(\cdot|x^{i-1},y^{i-1})$, for each $(x^{i-1},y^{i-1})$,~$i=0,1,\ldots,n$. This will be shown by induction.\\
Consider $n=0$. For any $h_{0}(\cdot)\in{BC}({\cal X}_{0})$, by definition of weak convergence we have
\begin{align}
\lim_{\alpha\longrightarrow\infty}\int_{{\cal X}_0}h_0(x)p_0^{\alpha}(dx_0)=\int_{{\cal X}_0}h_0(x)p_0^{o}(dx_0).\nonumber
\end{align}
Consider $n=1$. For any $h_{0}(\cdot)\in{BC}({\cal X}_{0})$, $h_{1}(\cdot)\in{BC}({\cal X}_{1})$, we need to show $\forall\epsilon>0$, there exists an $N\in\mathbb{N_{+}}\tri\{1,2,\ldots\}$ such that for $\alpha>N$
\begin{align}
\Bigg{|}\int_{{\cal X}_0}h_0(x_0)p_0^{\alpha}(dx_0)\int_{{\cal X}_1}h_1(x_1)p_1^{\alpha}(dx_1|x_0,y_0)-\int_{{\cal X}_0}h_0(x_0)p_0^{o}(dx_0)\int_{{\cal X}_1}h_1(x_1)p_1^{o}(dx_1|x_0,y_0)\Bigg{|}\leq\epsilon.\label{application:capacity:closedness:equation1}
\end{align}
From the left hand side (LHS) of (\ref{application:capacity:closedness:equation1}), by adding and subtracting terms, we have the following upper bound.
\begin{align}
&A_{0,1}\tri\Bigg{|}\int_{{\cal X}_0\times{\cal X}_1}h_0(x_0)h_1(x_1)p_1^{\alpha}(dx_1|x_0,y_0)p_0^{\alpha}(dx_0)-\int_{{\cal X}_0\times{\cal X}_1}h_0(x_0)h_1(x_1)p_1^{o}(dx_1|x_0,y_0)p_0^{o}(dx_0)\Bigg{|}\nonumber\\
&\leq{\underbrace{\Bigg{|}\int_{{\cal X}_0\times{\cal X}_1}h_0(x_0)h_1(x_1)p_1^{o}(dx_1|x_0,y_0)p_0^{\alpha}(dx_0)-\int_{{\cal X}_0\times{\cal X}_1}h_0(x_0)h_1(x_1)p_1^{o}(dx_1|x_0,y_0)p_0^{o}(dx_0)\Bigg{|}}_{Term-1}}\nonumber\\
&+\underbrace{\Bigg{|}\int_{{\cal X}_0\times{\cal X}_1}h_0(x_0)h_1(x_1)p_1^{\alpha}(dx_1|x_0,y_0)p_0^{\alpha}(dx_0)-\int_{{\cal X}_0\times{\cal X}_1}h_0(x_0)h_1(x_1)p_1^{o}(dx_1|x_0,y_0)p_0^{\alpha}(dx_0)\Bigg{|}}_{Term-2}.\label{application:capacity:closedness:equation2}
\end{align}
\noi {\it Term-1}:~Let $\epsilon_0>0$ be given, and consider {\it Term-1}. By the continuity of the function mapping $(x_0,y_0)\in{\cal X}_0\times{\cal Y}_0\longmapsto\int_{{\cal X}_1}h(x_1)p_1(dx_1|x_0,y_0)$ and the weak convergence $p_1^{\alpha}(\cdot|x_{0},y_{0})\buildrel w \over \longrightarrow{p}_1^{o}(\cdot|x_{0},y_{0})$, for each $(x_0,y_0)\in{\cal X}_0\times{\cal Y}_0$, then there exists an $N_1\in\mathbb{N_{+}}$ such that for all $\alpha\geq{N_1}$
\begin{align}
\Bigg{|}\int_{{\cal X}_0}h_0(x_0)\bigg(\int_{{\cal X}_1}h_1(x_1)p_1^{o}(dx_1|x_0,y_0)\bigg)\big(p_0^{\alpha}(dx_0)-p_0^{o}(dx_0)\big)\Bigg{|}\leq\epsilon_0.\label{application:capacity:closedness:equation2aa}
\end{align}
\noi {\it Term-2}:~Consider {\it Term-2}. By the weak convergence, $p_0^{\alpha}(dx_0)\buildrel w \over \longrightarrow{p}_0^{o}(dx_0)$, $p_1^{\alpha}(dx_1|x_{0},y_{0})\buildrel w \over \longrightarrow{p}_1^{o}(dx_1|x_{0},y_{0})$, for each $(x_0,y_0)\in{\cal X}_0\times{\cal Y}_0$. According to Prohorov's theorem there exist compact subset $K_0\subset{\cal X}_0$ such that $p_0^{\alpha}(K^c_0)\leq{\epsilon_1},~\alpha=1,2,\ldots$, and compact subset $K_1\subset{\cal X}_1$ such that $p_1^{\alpha}(K_1^c|x_0,y_0)\leq\epsilon_2,~\alpha=1,2,\ldots$, for each $(x_0,y_0)\in{\cal X}_0\times{\cal Y}_0$.\\
Hence, {\it Term-2} is written as follows.
\begin{align}
&\Bigg{|}\int_{K_0\cup{K}^c_0}h_0(x_0)\Big(\int_{{\cal X}_1}h_1(x_1)p_1^{\alpha}(dx_1|x_0,y_0)\Big)p_0^{\alpha}(dx_0)-\int_{K_0\cup{K}^c_0}h_0(x_0)\Big(\int_{{\cal X}_1}h_1(x_1)p_1^{o}(dx_1|x_0,y_0)\Big)p_0^{\alpha}(dx_0)\Bigg{|}\nonumber\\
&=\Bigg{|}\int_{{K}^c_0}h_0(x_0)\Big(\int_{{\cal X}_1}h_1(x_1)p_1^{\alpha}(dx_1|x_0,y_0)\Big)p_0^{\alpha}(dx_0)-\int_{{K}^c_0}h_0(x_0)\Big(\int_{{\cal X}_1}h_1(x_1)p_1^{o}(dx_1|x_0,y_0)\Big)p_0^{\alpha}(dx_0)\\
&+\int_{K_0}h_0(x_0)\Big(\int_{{\cal X}_1}h_1(x_1)p_1^{\alpha}(dx_1|x_0,y_0)\Big)p_0^{\alpha}(dx_0)-\int_{K_0}h_0(x_0)\Big(\int_{{\cal X}_1}h_1(x_1)p_1^{o}(dx_1|x_0,y_0)\Big)p_0^{\alpha}(dx_0)\Bigg{|}
\nonumber\\
&\leq\int_{{K}^c_0}||h_0(\cdot)||_{\infty}||h_1(\cdot)||_{\infty}p_0^{\alpha}(dx_0)+\int_{{K}^c_0}||h_0(\cdot)||_{\infty}||h_1(\cdot)||_{\infty}p_0^{\alpha}(dx_0)\nonumber\\
&\qquad+\Bigg{|}\int_{K_0}h_0(x_0)\Big(\int_{{\cal X}_1}h_1(x_1)p_1^{\alpha}(dx_1|x_0,y_0)-\int_{{\cal X}_1}h_1(x_1)p_1^{o}(dx_1|x_0,y_0)\Big)p_0^{\alpha}(dx_0)\Bigg{|}\nonumber\\
&\leq{2}.||h_0(\cdot)||_{\infty}||h_1(\cdot)||_{\infty}p^{\alpha}_0(K_0^c)\nonumber\\
&\qquad\qquad+\Bigg{|}\int_{K_0}h_0(x_0)\Big(\int_{{\cal X}_1}h_1(x_1)p_1^{\alpha}(dx_1|x_0,y_0)-\int_{{\cal X}_1}h_1(x_1)p_1^{o}(dx_1|x_0,y_0)\Big)p_0^{\alpha}(dx_0)\Bigg{|}.\nonumber\\
&\leq{2}.||h_0(\cdot)||_{\infty}||h_1(\cdot)||_{\infty}.\epsilon_1+||h_0(\cdot)||_{\infty}\sup_{x_0\in{K}_0}\Bigg{|}\int_{{\cal X}_1}h_1(x_1)p_1^{\alpha}(dx_1|x_0,y_0)-\int_{{\cal X}_1}h_1(x_1)p_1^{o}(dx_1|x_0,y_0)\Bigg{|}.\label{application:capacity:closedness:equation3aa}
\end{align}
\noi By utilizing condition (\ref{capacity:applications:lemma:closedness}), $\forall\epsilon_2>0$ there exists $N_2\in\mathbb{N}_1$ such that for all $\alpha\geq{N}_2$ 
\begin{align}
\sup_{x_0\in{K}_0}\Bigg{|}\int_{{\cal X}_1}h_1(x_1)p_1^{\alpha}(dx_1|x_0,y_0)-\int_{{\cal X}_1}h_1(x_1)p_1^{o}(dx_1|x_0,y_0)\Bigg{|}<\epsilon_2,~\forall{y_0}\in{\cal Y}_0\label{application:capacity:closedness:equation4aa}
\end{align}
Hence, by (\ref{application:capacity:closedness:equation2aa}), (\ref{application:capacity:closedness:equation3aa}), (\ref{application:capacity:closedness:equation4aa}), there exists an $N\in\mathbb{N}_1$ large enough such that for all $\alpha\geq{N}_2$, expression \eqref{application:capacity:closedness:equation2} is further bounded by
\begin{align*}
A_{0,1}\leq\epsilon_0+{2}.||h_0(\cdot)||_{\infty}||h_1(\cdot)||_{\infty}.\epsilon_1+||h_0(\cdot)||_{\infty}.\epsilon.
\end{align*}
Since $\epsilon_0,\epsilon_1,\epsilon_2>0$ are arbitrary, the claim holds for $n=1$, as well.\\
\noi Suppose that for $n=k$, and for each $h_i(\cdot)\in{BC}({\cal X}_i)$,~$i=0,1,\ldots,k$, and $\forall~\epsilon>0$, there exists $N^k\in\mathbb{N}_1$ such that for each $\alpha\geq{N^k}$
\begin{align}
\Bigg{|}\int_{{\cal X}_{0,k}}\otimes_{i=0}^k{h}_i(x_i){p}_i^{\alpha}(dx_i|x^{i-1},y^{i-1})-\int_{{\cal X}_{0,k}}\otimes_{i=0}^k{h}_i(x_i)p_i^{o}(dx_i|x^{i-1},y^{i-1})\Bigg{|}\leq\epsilon.\label{capacity:application:closedness:equation3}
\end{align}
We need to show that (\ref{capacity:application:closedness:equation3}) holds for $n=k+1$, i.e., 
\begin{align}
\otimes_{i=0}^{k+1}{p}_i^{\alpha}(\cdot|x^{i-1},y^{i-1})\buildrel w \over \longrightarrow\otimes_{i=0}^{k+1}{p}_i^{o}(\cdot|x^{i-1},y^{i-1})\nonumber
\end{align} 
whenever ${p}_i^{\alpha}(\cdot|x^{i-1},y^{i-1})\buildrel w \over \longrightarrow{p}_i^{o}(\cdot|x^{i-1},y^{i-1}),~i=0,1,\ldots,k+1$, and provided that $\otimes_{i=0}^k{p}_i^{\alpha}(\cdot|x^{i-1},y^{i-1})\buildrel w \over \longrightarrow\otimes_{i=0}^k{p}_i^{o}(\cdot|x^{i-1},y^{i-1})$. The derivation is similar to showing (\ref{application:capacity:closedness:equation1}), hence it is omitted.\\
This shows (\ref{application:capacity:closedness:equation00}), hence the set $\overleftarrow{P}_{0,n}(\cdot|y^{n-1})\in{\cal M}_1^{\bf C1}({\cal X}_{0,n}), y^{n-1} \in {\cal Y}_{0,n-1}$ is closed. By Theorem~\ref{weak_convergence}, {\bf Part A.} A2), this set is also tight, hence by Prohorov's theorem (Appendix~\ref{backround_material}, Theorem~\ref{corollary_of_prohorov}) it is compact. This completes the derivation.\qed

%%%%%%%%%%%%%%%%%%%%%%%%%%%%%%%%%%%%%%%%%%%%%%%%%%%%%%%%%%%%%%%%%%%%%%

\section{Proof of Lemma~\ref{closedness:nrdf:lemma}}\label{proof:closedness:nrdf}

\par {\bf(1)} Since every probability measure on a compact metric space is weakly compact, then the set  $\overrightarrow{Q}_{0,n}(\cdot|x^{n})\in {\cal M}_1^{\bf C2}({\cal Y}_{0,n}), x^n \in {\cal X}_{0,n}$ is weakly compact. This means that any sequence $\{\overrightarrow{Q}^{\alpha}_{0,n}(\cdot|x^n):~\alpha=1,2,\ldots\}$, possesses a weakly convergent subsequence $\overrightarrow{Q}^{\alpha_i}_{0,n}(dy^n|x^{n})\buildrel w \over\longrightarrow\bar{Q}_{0,n}^o(dy^n|x^{n})$, for each $x^n\in{\cal X}_{0,n}$, and hence tight (by Prohorov's theorem, see Appendix~\ref{backround_material}, Theorem~\ref{prohorov}), but $\bar{Q}_{0,n}^o(dy^n|x^{n})$ may not be an element of ${\cal M}_1^{\bf C2}({\cal Y}_{0,n})$ (i.e., it may fail to satisfy consistency condition {\bf C2}). By Prohorov's theorem, to show compactness of $\overrightarrow{Q}_{0,n}(\cdot|x^{n})\in {\cal M}_1^{\bf C2}({\cal Y}_{0,n}), x^n \in {\cal X}_{0,n}$, we need to show $\bar{Q}_{0,n}^o(\cdot|x^n)=\overrightarrow{Q}_{0,n}^o(\cdot|x^{n})\tri\otimes_{i=0}^n{q}^{o}_i(dy_i|y^{i-1},x^i)$, whenever $q_i^{\alpha}(dy_i|y^{i-1},x^i)\buildrel w \over\longrightarrow{q}_i^o(dy_i|y^{i-1},x^i),~i=0,1,\ldots,n$ (since ${\cal Y}_i,~i=0,1,\ldots,n$ are compact Polish spaces). The method is precisely the same as in Lemma~\ref{capacity:closedness:lemma}, hence it is omitted. Therefore, the set $\overrightarrow{Q}_{0,n}(\cdot|x^{n})\in {\cal M}_1^{\bf C2}({\cal Y}_{0,n}), x^n \in {\cal X}_{0,n}$  is closed, and since it is also tight, it is compact.  \\
\noi{\bf(2)} Next, we discuss how the fidelity set ${\cal{Q}}_{0,n}(D)$ is a closed subset of the compact set ${\cal M}_1^{\bf C2}({\cal Y}_{0,n})$, hence compact itself, that is, for each sequence $\{\overrightarrow{Q}^{\alpha}_{0,n}(\cdot|x^{n}):\alpha=1,2,\ldots\}\in{\cal Q}_{0,n}(D)$ there is a subsequence such that $\overrightarrow{Q}^{\alpha}_{0,n}(\cdot|x^{n})\buildrel w \over\longrightarrow\overrightarrow{Q}^{o}_{0,n}(\cdot|x^{n})\in{\cal Q}_{0,n}(D)$. We outline the derivation. Let $\{\overrightarrow{Q}^{\alpha}_{0,n}(\cdot|x^{n}):\alpha=1,2,\ldots\}\in{\cal Q}_{0,n}(D)\subset{\cal M}^{\bf C2}({\cal Y}_{0,n})$. Since ${\cal M}_1^{\bf C2}({\cal Y}_{0,n})$ is closed and uniformly tight, and hence compact, there exists a subsequence $\{\overrightarrow{Q}^{\alpha_{i}}_{0,n}(\cdot|x^n):i=1,2,\ldots\}\in{\cal M}_1^{\bf C2}({\cal Y}_{0,n})$ and a measure $\overrightarrow{Q}^{o}_{0,n}(\cdot|x^{n})\in{\cal M}_1^{\bf C2}({\cal Y}_{0,n})$ such that $\overrightarrow{Q}^{\alpha_i}_{0,n}(\cdot|x^{n})\buildrel w \over\longrightarrow\overrightarrow{Q}^{o}_{0,n}(\cdot|x^{n})$ for each $x^{n}\in{\cal X}_{0,n}$. Recall that $d_{0,n}:{\cal X}_{0,n}\times{\cal Y}_{0,n}\longmapsto[0,\infty]$ is a Borel measurable, non-negative, and continuous function on $y^{n}\in{\cal Y}_{0,n}$. Consider the sequence $\{d_{0,n}^{(k)}\tri{d}_{0,n}\wedge{k}:k\in\mathbb{N}_0\},~\mathbb{N}_1\tri\{1,2,\ldots\}$, which is bounded, and continuous function in the second argument $y^{n}\in{\cal Y}_{0,n}$. By Lebesgue's monotone convergence theorem and Fatou's lemma it can be shown that ${\cal Q}_{0,n}(D)$ is closed with respect to the topology of weak convergence. Since a closed subset of a compact set is compact, then ${\cal Q}_{0,n}(D)$ is compact. This completes the derivation.\qed

%%%%%%%%%%%%%%%%%%%%%%%%%%%%%%%%%%%%%%%%%%%%%%%%%%%%%%%%%%%%%%%%%%%%

\section{Proof of Theorem~\ref{lower_semicontinuity}}\label{proof:lower:semicontinuity}

\par 1) We need to show that for any sequence $\{\overrightarrow{Q}^{\alpha}_{0,n}(\cdot|x^n)\in{\cal M}_1^{\bf C2}({\cal Y}_{0,n}):\alpha=1,2,\ldots\}$, such that $\overrightarrow{Q}^{\alpha}_{0,n}(\cdot|x^n)\buildrel w \over \longrightarrow\overrightarrow{Q}^o_{0,n}(\cdot|x^n)$ for each $x^n\in{\cal X}_{0,n}$ then
\begin{align}
{\mathbb{I}}_{X^n\rightarrow{Y^n}}({\overleftarrow P}_{0,n},\overrightarrow{Q}^o_{0,n})\leq\liminf_{\alpha\rightarrow\infty}{\mathbb{I}}_{X^n\rightarrow{Y^n}}({\overleftarrow P}_{0,n},\overrightarrow{Q}^{\alpha}_{0,n}).\nonumber
\end{align}
Define the sequence of joint distribution $P_{0,n}^{\alpha}(dx^n,dy^n)\tri({\overleftarrow P}_{0,n}\otimes\overrightarrow{Q}^{\alpha}_{0,n})(dx^n,dy^n),~\alpha=1,2,\ldots$. Weak convergence $P_{0,n}^{\alpha}(dx^n,dy^n)\buildrel w \over \longrightarrow({\overleftarrow P}_{0,n}\otimes{\overrightarrow Q}_{0,n}^o)(dx^n,dy^n)\equiv{P}^o_{0,n}(dx^n,dy^n)$ is shown by considering integrals with respect to a test function $\phi_{0,n}(\cdot,\cdot){\in}BC({\cal X}_{0,n}\times{\cal Y}_{0,n})$ via
\begin{align}
\int_{{\cal X}_{0,n}\times{\cal Y}_{0,n}}\phi_{0,n}(x^n,y^n)P^{\alpha}_{0,n}(dx^n,dy^n)=\int_{{\cal X}_{0,n}\times{\cal Y}_{0,n}}\phi_{0,n}(x^n,y^n)({\overleftarrow P}_{0,n}\otimes{\overrightarrow Q}^{\alpha}_{0,n})(dx^n,dy^n).\nonumber
\end{align}
By Theorem~\ref{weak_convergence}, {\bf Part A.}, A1), $P_{0,n}^{\alpha}(dx^n,dy^n)\buildrel w \over \longrightarrow P^o_{0,n}(dx^n,dy^n)$. Similarly, consider ${\overrightarrow \Pi}_{0,n}^{\alpha}\tri{\overleftarrow P}_{0,n}\otimes\nu_{0,n}^{\alpha}~\alpha=1,2,\ldots$, where $\{\nu_{0,n}^{\alpha}:\alpha=1,2,\ldots\}$ are the marginals of $\{P^{\alpha}_{0,n}:\alpha=1,2,\ldots\}$. Then by Theorem~\ref{weak_convergence}, {\bf Part A.}, A4) we have
\begin{align}
\overrightarrow{\Pi}^{\alpha}_{0,n}=\overleftarrow{P}_{0,n}\otimes\nu^{\alpha}_{0,n}\buildrel w \over \longrightarrow\overrightarrow{\Pi}^o_{0,n}=\overleftarrow{P}_{0,n}\otimes\nu^o_{0,n}.\nonumber
\end{align}
Recall the definition of directed information via relative entropy given by
\begin{align}
\mathbb{D}(P_{0,n}||{\overrightarrow\Pi}_{0,n})&=\mathbb{D}({\overleftarrow P}_{0,n}\otimes{\overrightarrow Q}_{0,n}||{\overleftarrow P}_{0,n}\otimes\nu_{0,n})={\mathbb{I}}_{X^n\rightarrow{Y^n}}({\overleftarrow P}_{0,n},{\overrightarrow Q}_{0,n}).\label{equation41}
\end{align}
It is well known that relative entropy is lower semicontinuous, hence
\begin{align}
\mathbb{D}(P^o_{0,n}||\overrightarrow{\Pi}^o_{0,n})=\mathbb{D}({\overleftarrow P}_{0,n}\otimes\overrightarrow{Q}^o_{0,n}||\overrightarrow{\Pi}^o_{0,n})\leq\liminf_{\alpha\rightarrow\infty}\mathbb{D}(P^{\alpha}_{0,n}||\overrightarrow{\Pi}^{\alpha}_{0,n}).\label{equation42}
\end{align}
By (\ref{equation41}) it follows that (\ref{equation42}) is also equivalent to
\begin{align}
{\mathbb{I}}_{X^n\rightarrow{Y^n}}({\overleftarrow P}_{0,n},\overrightarrow{Q}^o_{0,n})\leq\liminf_{\alpha\rightarrow\infty}{\mathbb{I}}_{X^n\rightarrow{Y^n}}({\overleftarrow P}_{0,n},\overrightarrow{Q}_{0,n}^{\alpha})\nonumber
\end{align}
Hence, directed information is lower semicontinuous as a functional of $\overrightarrow{Q}_{0,n}(\cdot|x^n)\in{\cal M}_1^{\bf C2}({\cal Y}_{0,n})$ for a fixed $\overleftarrow{P}_{0,n}(\cdot|y^{n-1})\in{\cal M}_1^{\bf C1}({\cal X}_{0,n})$. This completes the derivation of 1).\\
2) The derivation is similar to 1).\qed

%%%%%%%%%%%%%%%%%%%%%%%%%%%%%%%%%%%%%%%%%%%%%%%%%%%%%%%%%%%%%%%%%%%%%

\section{Proof of Theorem~\ref{continuity}}\label{continuity1}

\par To show continuity of $\mathbb{I}_{X^n\rightarrow{Y^n}}(\cdot,\overrightarrow{Q}_{0,n})$ we need to show that  for every sequence $\{\overleftarrow{P}_{0,n}^{\alpha}(\cdot|y^{n-1}):\alpha=1,2,\ldots\}$ such that $\overleftarrow{P}_{0,n}^{\alpha}\buildrel w \over\longrightarrow \overleftarrow{P}_{0,n}^{o}$, we have
$$\mathbb{I}_{X^n\rightarrow{Y^n}}(\overleftarrow{P}_{0,n}^{\alpha},\overrightarrow{Q}_{0,n})\longrightarrow\mathbb{I}_{X^n\rightarrow{Y^n}}(\overleftarrow{P}_{0,n}^o,\overrightarrow{Q}_{0,n}).$$
The derivation is based on the procedure utilized in \cite{fozunbal} to show continuity for single letter mutual information. First, decompose directed information into two terms as follows.
\begin{align}
\mathbb{I}_{X^n\rightarrow{Y^n}}(\overleftarrow{P}_{0,n},\overrightarrow{Q}_{0,n})&=\int_{{\cal X}_{0,n}\times{\cal Y}_{0,n}}\log\Big(\frac{d{\overrightarrow Q}_{0,n}(\cdot|x^n)}{d\nu_{0,n}(\cdot)}(y^n)\Big)(\overleftarrow{P}_{0,n}\otimes\overrightarrow{Q}_{0,n})(dx^n,dy^n)\nonumber\\
&=\int_{{\cal X}_{0,n}\times{\cal Y}_{0,n}}\log\Big(\frac{d{\overrightarrow Q}_{0,n}(\cdot|x^n)}{d\nu_{0,n}(\cdot)}(y^n)\Big)(\overleftarrow{P}_{0,n}\otimes\overrightarrow{Q}_{0,n})(dx^n,dy^n)\nonumber\\
&-\qquad\int_{{\cal Y}_{0,n}}\log\Big(\frac{d{\overrightarrow Q}_{0,n}(\cdot|x^n)}{d\nu_{0,n}(\cdot)}(y^n)\Big)\nu_{0,n}(dy^n)\nonumber\\
&=\int_{{\cal X}_{0,n}\times{\cal Y}_{0,n}}\Big(\xi_{\bar{\nu}_{0,n}}(x^n,y^n)\log\xi_{\bar{\nu}_{0,n}}(x^n,y^n)\Big)\overleftarrow{P}_{0,n}(dx^n|y^{n-1})\otimes\bar{\nu}_{0,n}(dy^n)\nonumber\\
&-\qquad\int_{{\cal Y}_{0,n}}\Big(\xi_{\bar{\nu}_{0,n},\overleftarrow{P}_{0,n}}(y^n)\log\xi_{\bar{\nu}_{0,n},\overleftarrow{P}_{0,n}}(y^n)\Big)\bar{\nu}_{0,n}(dy^n),\label{equation_1}
\end{align}
where $\xi_{\bar{\nu}_{0,n},\overleftarrow{P}_{0,n}}(y^n)\tri\frac{d\nu_{0,n}(\cdot)}{d\bar{\nu}_{0,n}(\cdot)}(y^n)$ emphasizes the fact that this RND depends on $\overleftarrow{P}_{0,n}(\cdot|y^{n-1})$ via $\bar{\nu}(\cdot)$. For now, assume that both terms in on the RHS of the above formula are finite; the validity of this assumption will be established at the end. Thus, we only need to show that both terms are bounded and continuous in the weak sense over ${\cal M}_1^{{\bf C1},cl}({\cal X}_{0,n})$.\\
{\it Continuity of the first term}. Since $\overleftarrow{P}_{0,n}^{\alpha}(\cdot|y^{n-1})\buildrel w \over\longrightarrow \overleftarrow{P}_{0,n}^{o}(\cdot|y^{n-1})$, by \cite[Theorem~A.5.8, p.~320]{dupuis-ellis97}, utilizing Lebesgue's dominated convergence theorem, we have $\overleftarrow{P}_{0,n}^{\alpha}\otimes\bar{\nu}_{0,n}\buildrel w \over\longrightarrow\overleftarrow{P}_{0,n}^{o}\otimes\bar{\nu}_{0,n}$. Since $\xi_{\bar{\nu}_{0,n}}(x^n,y^n)$ is continuous, then so is $\xi_{\bar{\nu}_{0,n}}(x^n,y^n)\log\xi_{\bar{\nu}_{0,n}}(x^n,y^n)$. By hypothesis, $\xi_{\bar{\nu}_{0,n}}(x^n,y^n)\log\xi_{\bar{\nu}_{0,n}}(x^n,y^n)$ is uniformly integrable over $\big\{\bar{\nu}_{0,n}\otimes\overleftarrow{P}_{0,n}:{\overleftarrow P}_{0,n}(\cdot|y^{n-1})\in{\cal M}_1^{{\bf C1},cl}({\cal X}_{0,n})\big\}$. Therefore, using Theorem~\ref{continuity-uniform_integrability}, Appendix~\ref{backround_material}, we conclude that 
\begin{align}
\lim_{\alpha\rightarrow\infty}\int_{{\cal X}_{0,n}\times{\cal Y}_{0,n}}&\xi_{\bar{\nu}_{0,n}}(x^n,y^n)\log\xi_{\bar{\nu}_{0,n}}(x^n,y^n)\overleftarrow{P}_{0,n}^{\alpha}(dx^n|y^{n-1})\otimes\bar{\nu}_{0,n}(dy^n)\nonumber\\
&=\int_{{\cal X}_{0,n}\times{\cal Y}_{0,n}}\xi_{\bar{\nu}_{0,n}}(x^n,y^n)\log\xi_{\bar{\nu}_{0,n}}(x^n,y^n){P}_{0,n}^{o}(dx^n|y^{n-1})\otimes\bar{\nu}_{0,n}(dy^n).\label{equation101}
\end{align}
This proves the continuity of the first term. The finiteness of the first term is obtained from uniform integrability as follows. For a given $\epsilon>0$ and sufficiently large $c>0$
\begin{align*}
&\sup_{\overleftarrow{P}_{0,n}(\cdot|y^{n-1})\in{\cal M}_1^{{\bf C1},cl}({\cal X}_{0,n})}\Big\{\int_{{\cal X}_{0,n}\times{\cal Y}_{0,n}}\big{|}\xi_{\bar{\nu}_{0,n}}(x^n,y^n)\log\xi_{\bar{\nu}_{0,n}}(x^n,y^n)\big{|}I_{\{|\xi_{\bar{\nu}_{0,n}}(x^n,y^n)\log\xi_{\bar{\nu}_{0,n}}(x^n,y^n)|\geq{c}\}}\nonumber\\
&\qquad\qquad\times\overleftarrow{P}_{0,n}^{\alpha}(dx^n|y^{n-1})\otimes\bar{\nu}_{0,n}(dy^n)\\
&+\int_{{\cal X}_{0,n}\times{\cal Y}_{0,n}}\big{|}\xi_{\bar{\nu}_{0,n}}(x^n,y^n)\log\xi_{\bar{\nu}_{0,n}}(x^n,y^n)\big{|}I_{\{|\xi_{\bar{\nu}_{0,n}}(x^n,y^n)\log\xi_{\bar{\nu}_{0,n}}(x^n,y^n)|<c\}}\overleftarrow{P}_{0,n}^{\alpha}(dx^n|y^{n-1})\otimes\bar{\nu}_{0,n}(dy^n)\Big\}\\
&\leq\sup_{\overleftarrow{P}_{0,n}(\cdot|y^{n-1})\in{\cal M}_1^{{\bf C1},cl}({\cal X}_{0,n})}\Big\{\int_{\{|\xi_{\bar{\nu}_{0,n}}(x^n,y^n)\log\xi_{\bar{\nu}_{0,n}}(x^n,y^n)|\geq{c}\}}\big{|}\xi_{\bar{\nu}_{0,n}}(x^n,y^n)\log\xi_{\bar{\nu}_{0,n}}(x^n,y^n)\big{|}\nonumber\\
&\qquad\qquad\times\overleftarrow{P}_{0,n}^{\alpha}(dx^n|y^{n-1})\otimes\bar{\nu}_{0,n}(dy^n)\Big\}\\
&+\sup_{\overleftarrow{P}_{0,n}(\cdot|y^{n-1})\in{\cal M}_1^{{\bf C1},cl}({\cal X}_{0,n})}\Big\{\int_{\{|\xi_{\bar{\nu}_{0,n}}(x^n,y^n)\log\xi_{\bar{\nu}_{0,n}}(x^n,y^n)|<{c}\}}\big{|}\xi_{\bar{\nu}_{0,n}}(x^n,y^n)\log\xi_{\bar{\nu}_{0,n}}(x^n,y^n)\big{|}\\
&\qquad\qquad\times\overleftarrow{P}_{0,n}^{\alpha}(dx^n|y^{n-1})\otimes\bar{\nu}_{0,n}(dy^n)\Big\}\leq\epsilon+c.
\end{align*}
{\it Continuity of the second term}. For a fixed $y^n\in{\cal Y}_{0,n}$, since $\xi_{\bar{\nu}_{0,n}}(x^n,y^n)$ is uniformly integrable over ${\cal M}_1^{{\bf C1},cl}({\cal X}_{0,n})$, by Theorem~\ref{uniform_integrability_1}, Appendix~\ref{backround_material}, we obtain that $\overleftarrow{P}_{0,n}^{\alpha}\buildrel w \over\longrightarrow \overleftarrow{P}_{0,n}^{o}$, implies pointwise convergence of $\xi_{\bar{\nu}_{0,n},\overleftarrow{P}_{0,n}^{\alpha}}(y^n)\longrightarrow\xi_{\bar{\nu}_{0,n},\overleftarrow{P}_{0,n}^o}(y^n)$. By continuity of the logarithm, we obtain the pointwise convergence of $\xi_{\bar{\nu}_{0,n},\overleftarrow{P}_{0,n}^{\alpha}}(y^n)\log\xi_{\bar{\nu}_{0,n},\overleftarrow{P}_{0,n}^{\alpha}}(y^n)\longrightarrow\xi_{\bar{\nu}_{0,n},\overleftarrow{P}_{0,n}^{o}}(y^n)\log\xi_{\bar{\nu}_{0,n},\overleftarrow{P}_{0,n}^{0}}(y^n)$. It only remains to show convergence under the integral with respect to $\bar{\nu}_{0,n}$. By (\ref{equation155}), then $\forall{\alpha}$
\begin{align}
\Bigg{|}\xi_{\bar{\nu}_{0,n},\overleftarrow{P}_{0,n}^{\alpha}}&(y^n)\log\xi_{\bar{\nu}_{0,n},\overleftarrow{P}_{0,n}^{\alpha}}(y^n)\Bigg{|}\leq\frac{2}{e\ln2}+\xi_{\bar{\nu}_{0,n},\overleftarrow{P}_{0,n}^{\alpha}}(y^n)\log\xi_{\bar{\nu}_{0,n},\overleftarrow{P}_{0,n}^{\alpha}}(y^n)\nonumber\\
&=\frac{2}{e\ln2}\int_{{\cal X}_{0,n}}\xi_{\bar{\nu}_{0,n},\overleftarrow{P}_{0,n}^{\alpha}}(y^n)\log\xi_{\bar{\nu}_{0,n},\overleftarrow{P}_{0,n}^{\alpha}}(y^n)\overleftarrow{P}_{0,n}^{\alpha}(dx^n|y^{n-1})\label{equation_2}\\
&\leq\frac{2}{e\ln2}+\int_{{\cal X}_{0,n}}\Big(\xi_{\bar{\nu}_{0,n}}(x^n,y^n)\log\xi_{\bar{\nu}_{0,n}}(x^n,y^n)\Big)\overleftarrow{P}_{0,n}^{\alpha}(dx^n|y^{n-1}).\nonumber
\end{align}
where (\ref{equation_2}) follows from (\ref{equation_1}) and the nonnegativity of $\mathbb{I}_{X^n\rightarrow{Y^n}}(\overleftarrow{P}_{0,n},\overrightarrow{Q}_{0,n})$.
By (\ref{equation101}), the integration of the RHS over $\bar{\nu}_{0,n}$ converges. Thus, by the generalized Lebesgue's dominated convergence theorem \cite[p.~59]{folland1999}, we conclude that
\begin{align}
\int_{{\cal Y}_{0,n}}&\xi_{\bar{\nu}_{0,n},\overleftarrow{P}_{0,n}^{\alpha}}(y^n)\log\xi_{\bar{\nu}_{0,n},\overleftarrow{P}_{0,n}^{\alpha}}(y^n)\bar{\nu}_{0,n}(dy^n)
\buildrel{\alpha\rightarrow\infty} \over\longrightarrow\int_{{\cal Y}_{0,n}}\xi_{\bar{\nu}_{0,n},\overleftarrow{P}_{0,n}^{0}}(y^n)\log\xi_{\bar{\nu}_{0,n},\overleftarrow{P}_{0,n}^{o}}(y^n)\bar{\nu}_{0,n}(dy^n).\nonumber
\end{align}
This implies the continuity of the second term. Furthermore, its finiteness follows as before. Since both terms are finite and continuous we deduce continuity of the directed information $\mathbb{I}_{X^n\rightarrow{Y^n}}(\cdot,\overrightarrow{Q}_{0,n})$ with respect to $\overleftarrow{P}_{0,n}(\cdot|y^{n-1})$, for fixed $\overrightarrow{Q}(\cdot|x^n)$. This completes the derivation.

\bibliographystyle{IEEEtran}
\bibliography{photis_references_DI_properties}

\end{document}